\documentclass{imsart}

\RequirePackage[OT1]{fontenc}
\RequirePackage{amsthm,amsmath}
\RequirePackage[colorlinks,citecolor=blue,urlcolor=blue]{hyperref}
\usepackage{graphicx}

\usepackage{natbib}
\usepackage{graphicx}
\usepackage[utf8]{inputenc}
\usepackage{graphics}
\usepackage{amsmath}
\usepackage{amssymb}
\usepackage{mathtools}
\usepackage{color}
\usepackage{verbatim}
\usepackage{enumitem}
\usepackage[ruled,vlined]{algorithm2e}
\usepackage{upgreek}
\usepackage{marginnote}
\usepackage{multirow}
\usepackage{diagbox}
\usepackage{slashbox}
\usepackage{graphicx}
\usepackage{caption}
\usepackage{subcaption}
\usepackage{ragged2e} 
\usepackage[english]{babel} 

\newcommand{\Pro}{\mathrm{P}}

\newcommand{\Exp}{\mathrm{E}}

\newcommand{\cA}{A}

\newcommand{\cC}{C}

\newcommand{\cF}{\mathcal{F}}
\newcommand{\cI}{\mathcal{I}}
\newcommand{\cJ}{\mathcal{J}}
\newcommand{\cL}{\mathcal{L}}
\newcommand{\cP}{\mathcal{P}}

\newcommand{\bN}{\mathbb{N}}
\newcommand{\bR}{\mathbb{R}}

\newcommand{\bfT}{\mathbf{T}}
\newcommand{\bfD}{\mathbf{D}}
\newcommand{\bfTs}{\hat{\mathbf{T}}}
\newcommand{\bfDs}{\hat{\mathbf{D}}}
\newcommand{\Ts}{\hat{T}}
\newcommand{\Ds}{\hat{D}}

\newcommand{\bfDc}{\check{\mathbf{D}}}
\newcommand{\Tc}{\check{T}}

\newcommand{\bfDss}{\hat{\mathbf{D}}}
\newcommand{\Tss}{\hat{T}}

\newcommand{\bftheta}{{\boldsymbol\theta}}
\newcommand{\bfTheta}{{\boldsymbol\Theta}}
\newcommand{\lambdas}{\lambda^*}
\newcommand{\Gammas}{\hat\Gamma}
\newcommand{\Pros}{\Pro^*}
\newcommand{\Exps}{\Exp^*}
\newcommand{\bfone}{\mathbf{1}}

\newcommand{\Def}{\equiv}
\newcommand{\bfa}{\mathbf{a}}

\newcommand{\ha}{\hat\alpha}
\newcommand{\hb}{\hat\beta}
\newcommand{\ca}{\check\alpha}
\newcommand{\cb}{\check\beta}
\newcommand{\Ae}{\operatorname{ARE}}
\newcommand{\ta}{\Tilde{\alpha}}
\newcommand{\tb}{\Tilde{\beta}}
\newcommand{\ellg}{\ell}
\newcommand{\ella}{\ell^*}

\newtheorem{theorem}{Theorem}[section]
\newtheorem{corollary}{Corollary}[theorem]
\newtheorem{lemma}{Lemma}[section]
\newtheorem{proposition}{Proposition}[section]
\theoremstyle{remark}
\newtheorem{remark}{Remark}[section]

\startlocaldefs
\numberwithin{equation}{section}
\theoremstyle{plain}

\endlocaldefs

\begin{document}

\begin{frontmatter}
\title{Asymptotically Optimal Sequential Multiple Testing with Asynchronous Decisions
}
\runtitle{Sequential Multiple Testing}

\begin{aug}
\author{\fnms{Yiming} \snm{Xing}
\ead[label=e1]
{yimingx4@illinois.edu}}
\and
\author{\fnms{Georgios} \snm{Fellouris}
\ead[label=e2]
{fellouri@illinois.edu}}
\address{
605 E. Springfield Ave. Champaign, IL 61820, USA \\
University of Illinois at Urbana-Champaign}

\runauthor{Y. Xing and G. Fellouris}

\affiliation{University of Illinois, Urbana-Champaign}
\end{aug}

\begin{abstract}
The problem of simultaneously testing  the marginal distributions of sequentially monitored, independent data streams is considered.  The decisions for the various testing problems can be made at different times, using data from all streams, which can be monitored until all decisions have been made.
Moreover,  arbitrary  a priori bounds are assumed  on the number of signals, i.e., data streams in which the alternative hypothesis is correct. 
A novel sequential multiple testing procedure is proposed and it is  shown to achieve the minimum expected decision time, simultaneously in every data stream and under every signal configuration, asymptotically as certain metrics of global error rates go to zero.
This optimality property is established under 
general parametric composite hypotheses, various  error metrics, and weak distributional assumptions that  allow for temporal dependence. Furthermore, the limit of the  factor by which the expected decision time in a data stream  increases when one is limited to  synchronous or decentralized procedures is evaluated. 
Finally, two existing sequential multiple testing  procedures in the literature  are compared with the proposed one in various simulation studies.
\end{abstract}


\end{frontmatter}

\section{Introduction}
In application areas such as multichannel anomaly detection \cite{Kobi_2015a}, clinical trials with multiple endpoints \cite[Chapter 15]{jennison1999group}; \cite{Bartroff_2010}, 
gene association or expression studies \cite{Zehetmayer_2005, sarkar2013multiple}, postmarketing safety surveillance of  medical products \cite{markatou},  data are  generated by distinct  sources and collected sequentially, and it is of interest to solve, simultaneously and  in real time,  a hypothesis testing problem for each of the resulting data streams. 

For such \textit{sequential multiple testing problems},
there are various plausible setups and formulations as far as it concerns the times at which sampling is terminated and the times at which the  decisions are  made.  

One such setup,  schematically depicted  in Figure \ref{Sketch}.(a), corresponds to the case where  \textit{there is a common time at which the decisions are made for all  testing problems and all data sources are continuously monitored until that time}  \cite{De_Baron_Seq_Bonf, Step_up_down, De_2015, Song_prior, Song_AoS, Bartroff2021, chaudhuri2022joint}. In what follows, we refer to such sequential multiple testing procedures as \textit{synchronous}, as they stop sampling and make a decision in all data streams at the same time. 

A different setup, schematically depicted  in Figure \ref{Sketch}.(b),  arises when \textit{the decisions can be made at different times and  sampling from a data source is terminated once the decision for the corresponding testing problem is made}  \cite{Bartroff_2014, Bartroff_Song_2016}. An interesting  special case, that is considered for example in   \cite{Malloy_Nowak_2014, PaperII},  is when 
the time at which sampling is terminated in a data stream  and the selected hypothesis for  the corresponding  testing problem   \emph{must rely on data only from this data stream}. In what follows, we refer to such sequential multiple testing procedures  as  \textit{decentralized}. 

\begin{figure}
     \centering
     \begin{subfigure}[b]{0.3\textwidth}
         \centering
         \includegraphics[width=\textwidth]{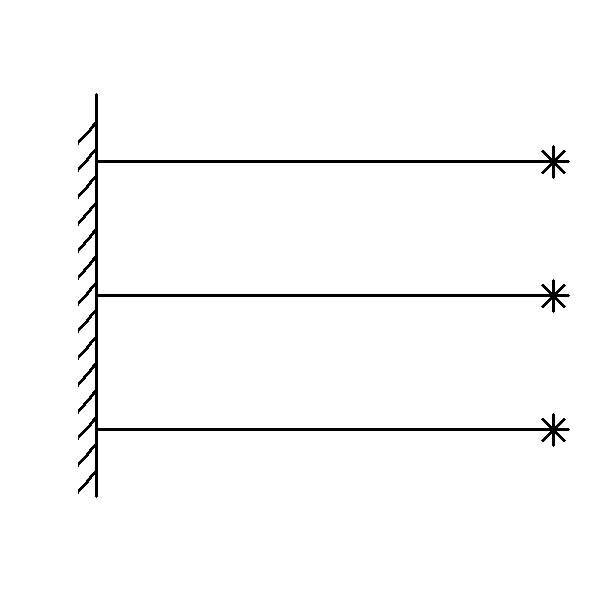}
         \caption{}
         \label{Sketch 1}
     \end{subfigure}
     \hfill
     \begin{subfigure}[b]{0.3\textwidth}
         \centering
         \includegraphics[width=\textwidth]{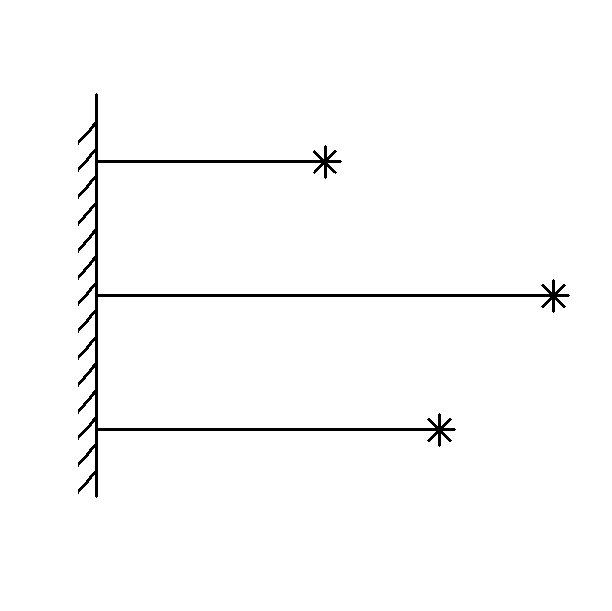}
         \caption{}
         \label{Sketch 2}
     \end{subfigure}
     \hfill
     \begin{subfigure}[b]{0.3\textwidth}
         \centering
         \includegraphics[width=\textwidth]{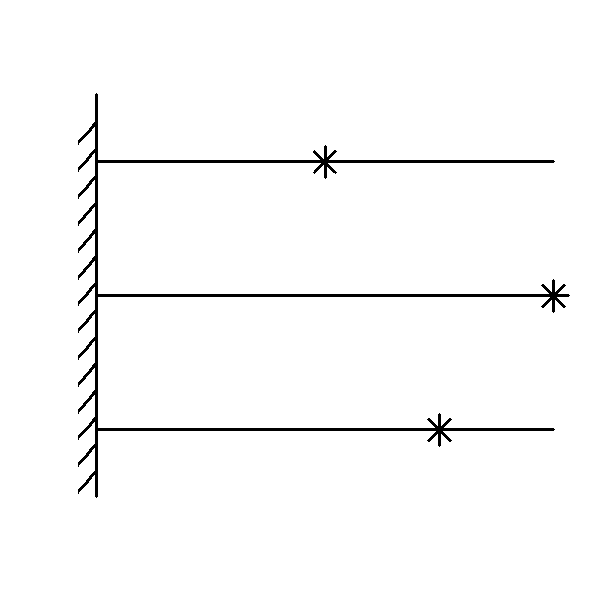}
         \caption{}
         \label{Sketch 3}
     \end{subfigure}
        \caption{Schematic depiction of three setups for  sequential multiple testing when data are generated by distinct  sources.  Each horizontal line segment represents the time span of sampling from the corresponding data source and $*$ the decision time for the corresponding testing problem.}
        \label{Sketch}
\end{figure}
 
In both  these setups, sampling is terminated in a data stream as soon as the decision is made for the corresponding testing problem. The difference is that in the first one the decision times must coincide, whereas in the second they may differ.  As a result, the first setup is suitable when  an action needs to be taken based on the decisions for \textit{all} testing problems, whereas  the second when an action must be taken based on the decision for \emph{each, individual} testing problem. 

Our motivation for the present work is the understanding  of a third  setup, which  is schematically depicted  in Figure \ref{Sketch}.(c).  In this one, the decisions for the various testing problems can be made at different times, using  data  from all  streams, \textit{which are monitored until all decisions have been made}.  As a result, this setup is suitable when a timely action needs to be taken based on the decision for \emph{each} testing problem,  but it is not of primary interest, or maybe even possible, to quickly stop sampling from each data source.  

For example, suppose that  multiple sensors are deployed, each of them monitors a certain environment, and an action needs to be taken in each of these environments depending on  the presence or absence of signal in it. In this context, the advantages of  asynchronous decisions are obvious, but it may not be convenient, or necessary, to shut down a sensor after its corresponding decision has been made. 
This  can be the case  in applications  such as missile identification based on radar measurements \cite{Missile_detection_2004}, spectrum sensing for cognitive radio \cite{gupta2019progression}, intrusion detection in a computer-based system and fraud detection in a commercial organization \cite{chandola2009anomaly}.  However, apart from a brief discussion in \cite[Section 2]{De_2015}, this very natural setup for sequential multiple testing  does not seem to have been studied in the literature. 

In the present paper,  we (i) propose a general formulation  for the sequential multiple testing  problem that  encompasses all the above setups, (ii) introduce  a novel testing procedure that takes advantage of the flexibility of this formulation, (iii) show that this procedure enjoys a  general asymptotic optimality property, and (iv) 
compute its asymptotic gains over  synchronous and decentralized procedures.

To be specific, we assume that there are multiple, independent, sequentially monitored data streams, and we consider  a binary hypothesis testing problem for each of them. These testing problems are coupled by global, user-specified error constraints.  The decisions for the various testing problems can be made at different times, using data from all streams, which can be monitored until all decisions have been made. Thus,  all three setups described  earlier are included in this general formulation.  We introduce a novel, sequential multiple testing procedure that incorporates arbitrary a priori lower and upper bounds on the number of  signals, i.e., data streams in which the alternative hypothesis holds. We show that this procedure achieves the optimal expected decision time \textit{in every data stream and under every signal configuration} to a first-order asymptotic approximation as  global error constraints of both types go to zero. This asymptotic optimality property is established under general parametric composite hypotheses,  various global error metrics,   and weak distributional assumptions that allow for temporal dependence.  Moreover,  we evaluate the factor by which the expected decision time in each stream increases, asymptotically as the error rates go to zero, when one is limited to decentralized or synchronous procedures. Finally, we compute in various simulation studies the  relative efficiencies of two existing sequential multiple testing procedures over the proposed,   and compare them with their limiting values that  are obtained from the asymptotic theory. 

The remainder of this paper is organized as follows: We formulate the sequential multiple testing problem that we consider in this work in Section \ref{section: problem formulation}. We introduce  and analyze the proposed  procedure,  contrasting it with existing procedures in the literature,  in Sections \ref{section: the proposed test} and \ref{sec: analysis}.  We establish our asymptotic optimality  theory in Section \ref{sec: AO} and present the results of various simulation studies in Section \ref{section: simulation studies}.    We extend our methodology and theoretical results  to general parametric composite hypotheses in Section \ref{section: generalization to composite hypotheses} and to various error metrics in Section \ref{section: generalization to other GEM}. We conclude in Section \ref{section: conclusion}. The proofs of all main results, as well as several supporting lemmas, are presented in Appendices \ref{Supporting lemmas}-\ref{Proofs related to other global error metrics}.

We end this introductory section with some  notations that we use throughout the paper. We  set $\bN\equiv \{1, 2, \ldots\}$ and, for any $n\in\bN$,  $[n]\equiv \{1,\ldots, n\}$. For a set $\cA$, we denote by $2^{\cA}$ its power set, by $|\cA|$ its size, and by $\bfone\{\cA\}$ its indicator function.
For any  real numbers $x,y$, we set $x\wedge y\equiv\min\{x,y\}$ and $x\vee y\equiv\max\{x,y\}$.
For any  sequences of positive real numbers $(x_n)$ and $(y_n)$,   $x_n\sim y_n$ stands for $\lim_n (x_n/y_n)= 1$, $x_n\gtrsim y_n$    for $\liminf_n (x_n/y_n) \geq 1$, $x_n\lesssim y_n$  for $\limsup_n (x_n/y_n) \leq 1$, and $x_n\lesssim O(y_n)$  means that  there exists a $C\in\bR$  so that  $x_n\leq C y_n$ for  all $n\in\bN$.  Finally, the minimum or infimum over the empty set is understood as $\infty$. 

\section{Problem formulation} \label{section: problem formulation}
We consider $K$ data sources that generate independent streams of observations, 
$$ X_k \Def \{X_k(n), \, n\in \bN\}, \quad k\in[K]. $$
For each $n\in\bN$ and $k \in [K]$,  we denote by $\cF_k(n)$ the $\sigma$-algebra generated by the observations in stream $k$  up to time $n$, and 
by  $\cF(n)$ the $\sigma$-algebra generated by the observations in all streams up to time $n$, i.e., 
$$\cF_k(n)\Def \sigma\left(X_k(t), \, t\in [n] \right) , \quad  \cF(n)\Def\sigma\left( \cF_k(n), \, k\in [K]\right). $$
For each $k \in [K]$, we  denote by $\Pro_k$ the distribution of $X_k$, consider two hypotheses for it, 
\begin{equation} \label{Testing problem}
\Pro_k \in \cP_k^0 \quad \text{ versus } \quad \Pro_k \in 
 \cP_k^1,
\end{equation}
and refer to the $k^{th}$ data source as 
 \emph{signal} if $\Pro_k \in \cP_k^1$ and as  \emph{noise} if $\Pro_k \in \cP_k^0$.    
 
    The   true subset of signals is unknown and deterministic, but we allow for incorporation of information about it  by assuming  that it is \textit{a priori} known to belong to some  class, $\Pi$, of subsets of $[K]$.  Without loss of generality, we assume that
\begin{equation} \label{no trivial}
    \forall \,  k\in[K]  \quad \exists\, A, B\in\Pi \quad \text{such that} \quad k\in A \text{ and }   k \notin B.
\end{equation}
When, in particular,   a  lower bound, $l$,  and an  upper bound, $u$,  are given  on the true number of signals,  where
\begin{equation} \label{lu}
0\leq l\leq u\leq K, \quad u>0, \quad l<K,
\end{equation}
we set  $\Pi=\Pi_{l,u}$, where 
$$ \Pi_{l,u} \Def \{\cA\subseteq [K]: l\leq |\cA|\leq u\}.$$
In what follows, when referring to  $\Pi$ we assume that 
\eqref{no trivial} holds,  and when referring to  $\Pi_{l,u}$ we assume that \eqref{lu} holds. \\

To solve the above multiple hypothesis testing problem,  we need to determine two vectors,
 $$ \bfT \equiv (T_1,\ldots,T_K), \quad \bfD \equiv (D_1,\ldots,D_K),$$
where, for each $k \in [K]$,  
\begin{itemize}
\item  $T_k$ is a random time, taking values in $\bN$, at which the decision is made for the testing problem in stream $k$, 
\item   $D_k$ is a Bernoulli random variable such that stream $k$ is identified as a signal at time $T_k$ if and only if  $D_k=1$.   
\end{itemize}
With an abuse of notation, we also denote by $\bfD$ the subset of streams that are identified as signals, i.e., $$ \bfD \equiv \{k\in [K]: D_k=1\}.$$

We assume that  each decision time and each selected hypothesis  can depend on the already collected observations from all data streams.  Thus,  we say that
$\chi\equiv(\bfT, \bfD)$  is a \emph{sequential multiple testing} procedure  if,  for each $k \in [K]$,  
\begin{itemize}
\item  $T_k$ is a  stopping time with respect to
 $\{\cF(n), n\in \bN\}$,
\item  $D_k$  is an $\cF(T_k)$-measurable Bernoulli  random variable, 
\end{itemize}
i.e., for each   $n \in \bN$ and $i \in \{0,1\}$, 
$$\{T_k=n \}, \; \{T_k=n, \, D_k=i\} \in \cF(n).$$ 

 We refer to such a sequential multiple testing procedure as \textit{synchronous} if all decisions occur at the same time,  i.e., $$ T_1=\cdots=T_K,$$  and as \textit{decentralized} if each decision time and each selected hypothesis can depend on the already collected observations only from the corresponding data stream.
  That is, we say that 
$\chi\equiv(\bfT, \bfD)$    is  \textit{decentralized} if,  for each  $k\in [K]$,
\begin{itemize}
\item  $T_k$  is a stopping time with respect to $\{\cF_k(n), n\in \bN\}$, 
\item $D_k$   is an  $\cF_k(T_k)$-measurable Bernoulli  random variable, 
\end{itemize}
i.e., for each  $n \in \bN$ and $i \in \{0,1\}$, 
$$\{T_k=n \}, \; \{T_k=n, \, D_k=i\} \in \cF_k(n).
$$

We denote by $\Delta$ the family of all sequential multiple testing procedures,  by  $\Delta'$ the subfamily of \textit{decentralized}  procedures, 
and by $\Delta''$ the subfamily of \textit{synchronous}  procedures.  We  will further focus on sequential multiple testing procedures that satisfy user-specified bounds on certain metrics for the false positive and the false negative error rates. To simplify the presentation,  we consider first the  case of simple hypotheses and  classical familywise  error rates, and we extend the formulation and all subsequent results to general parametric composite hypotheses and to various other error metrics in Sections  \ref{section: generalization to composite hypotheses} and \ref{section: generalization to other GEM}, respectively.  Therefore, until then,  we assume that  the hypotheses in \eqref{Testing problem} are of the form
$$\cP_k^i =\{\Pro_k^i\},\quad i\in\{0,1\}, \quad k \in [K]$$   and, for any $\cA\in\Pi$, we denote by $\Pro_\cA$ the joint distribution of $(X_1,\ldots, X_k)$ when  the subset of signals is $\cA$, i.e.,
\begin{equation} \label{Product measure}
    \Pro_\cA\Def  \Pro_1\otimes \cdots \otimes  \Pro_K, \quad \text{where} \quad
    \Pro_k=\begin{cases}
    \begin{aligned}
        & \Pro_k^0, \; && \text{if } k\notin\cA\\
        & \Pro_k^1, \; && \text{if } k\in\cA,
    \end{aligned}
    \end{cases}
\end{equation}
and by $\Exp_\cA$ the corresponding expectation.  Moreover, for a sequential multiple testing procedure  
$\chi\equiv (\bfT, \bfD) \in\Delta$  that terminates almost surely  in every stream  when the true subset of signals is $\cA$, i.e.,
\begin{equation} \label{a.s. finite}
    \Pro_\cA(T_k<\infty)=1,\quad \forall\; k\in[K],    
\end{equation} 
we denote by  $\text{FWE}_\cA^1(\chi)$  its \emph{type-I  familywise error rate}, i.e., the probability of at least one noise  being identified as signal, when the true subset of signals is $\cA$, 
\begin{equation} \label{def of FWE1}
    \text{FWE}_\cA^1(\chi)  \Def \Pro_\cA(\bfD\backslash\cA\neq \emptyset)  = \Pro_\cA( \exists \;  j\notin\cA:  D_j=1), 
\end{equation}
and by $\text{FWE}_\cA^2(\chi)$  its \emph{type-II  familywise error rate},  i.e., the probability of at least one signal being identified as noise,   when the true subset of signals is $\cA$,
\begin{equation} \label{def of FWE2}
\text{FWE}_\cA^2(\chi)  \Def \Pro_\cA(\cA\backslash\bfD\neq\emptyset)= 
    \Pro_\cA( \exists \;  i \in \cA: D_i=0).
\end{equation}

Thus, for any class of prior information  $\Pi$ and any $\alpha,\beta\in (0,1)$,
\begin{equation} \label{Definition of Delta_alpha,beta(P)}
\begin{aligned}
    \Delta(\alpha,\beta,\Pi)\Def \bigg\{  \chi \in\Delta :  \text{FWE}_\cA^1(\chi) &\leq \alpha \; \; \& \; \;
    \text{FWE}_\cA^2(\chi)\leq \beta, \;   \forall\, \cA\in\Pi  \bigg\}
\end{aligned}
\end{equation}
is the family of sequential multiple testing  procedures that control, under every signal configuration consistent with $\Pi$,\
the type-I and type-II familywise error rates below  $\alpha$ and $\beta$, respectively,
 and 
\begin{equation} \label{def of AO}
  \cL_{k,\cA} (\alpha, \beta, \Pi) \equiv  \inf \left\{ \Exp_\cA [T_k]:  (\bfT, \bfD) \in \Delta(\alpha,\beta,\Pi) \right\}
\end{equation}
is the  optimal in $\Delta(\alpha, \beta, \Pi)$ expected decision time in stream $k \in [K]$  when the  subset of signals is $\cA \in \Pi$. 

The first goal of the present work is to introduce a  testing procedure that  achieves  \eqref{def of AO}  to a first-order asymptotic approximation  as  $\alpha$  and $\beta$ go to 0, \textit{simultaneously for every $\cA\in\Pi$ and $k\in[K]$},   when $\Pi$ is of the form $\Pi_{l,u}$.  The second goal  is to  quantify the asymptotic gains of the proposed procedure over decentralized and  synchronous procedures.  Specifically, for any $\cA\in \Pi$,  $k \in [K]$, $\alpha, \beta \in (0,1)$,  we consider the optimal   expected decision time in stream $k$ when the subset of signals is $\cA$ in the subfamily of \textit{decentralized}  procedures in
$\Delta(\alpha, \beta, \Pi)$, i.e., 
\begin{equation} \label{def of AO for decentralized}
  \cL'_{k,\cA} (\alpha, \beta, \Pi) \equiv  \inf \left\{ \Exp_\cA [T_k]:    (\bfT, \bfD)  \in   \Delta' \cap \Delta(\alpha,\beta,\Pi) \right\},
\end{equation}
and the optimal expected decision time when the subset of signals is $\cA$ in the subfamily of \textit{synchronous}  procedures in  $\Delta(\alpha, \beta, \Pi)$, i.e., 
\begin{equation} \label{def of AO for synchronous}
\begin{aligned}
    \cL''_{\cA}(\alpha, \beta, \Pi)   &\equiv   \inf \left\{\Exp_\cA[T_1]:
     (\bfT, \bfD)  \in  \Delta'' \cap \Delta(\alpha,\beta,\Pi)\right\}.
\end{aligned}
\end{equation} 
Our second goal is to  evaluate, for any  $\cA\in \Pi_{l,u}$ and $k \in [K]$,     the asymptotic relative efficiencies
\begin{align}
\Ae'_{k,\cA}(l,u) &\equiv   \lim_{\alpha,\beta\to 0} \frac{\cL_{k,\cA}(\alpha, \beta, \Pi_{ l,u} ) }
{ \cL'_{k,\cA}(\alpha, \beta, \Pi_{ l,u} ) }, \label{AE'} \\
\Ae''_{k,\cA}(l,u)   &\equiv   \lim_{\alpha,\beta\to 0} \frac{\cL_{k,\cA}(\alpha, \beta, \Pi_{ l,u} ) }
{ \cL''_{\cA}(\alpha, \beta, \Pi_{ l,u} ) }, \label{AE''}
\end{align}
where the inverse of \eqref{AE'} (resp. \eqref{AE''}) represents the limiting value as $\alpha, \beta \to 0$ of the  factor by which the best possible expected decision time  in stream $k$ when  the subset of signals is $ \cA \in \Pi_{l,u}$  increases  \textit{when allowing only for decentralized  (resp.  synchronous) procedures} in   $\Delta(\alpha, \beta, \Pi)$.

\subsection{Distributional assumptions} \label{subsec: distributional assumptions}
For each $k\in [K]$ and $n \in \bN$, we assume that the probability measures $\Pro_k^1$ and $\Pro_k^0$ are mutually absolutely continuous when restricted to $\cF_k(n)$, and denote by $\lambda_k(n)$ the corresponding \emph{log-likelihood ratio} (LLR), i.e.,
\begin{equation} \label{lambdak(n)}
    \lambda_k(n) \Def \log \frac{d\Pro_k^1}{d\Pro_k^0} \left(\cF_k(n)\right).
\end{equation}

For the proposed test to terminate with probability 1 and to achieve the prescribed error control,  it suffices to assume that,  for each $k\in [K]$,
\begin{align} \label{orthogonal}
\begin{split}
   \Pro_k^1\left( \lim_{n\to \infty} \lambda_k(n) = \infty \right) &=1\\
  \Pro_k^0\left( \lim_{n\to \infty} \lambda_k(n) =-\infty \right) &=1.
  \end{split}
\end{align}

To  establish our asymptotic optimality theory, we  need to make some further distributional assumptions. Specifically, we  assume that for each $k\in [K]$ there are positive numbers, $I_k$ and $ J_k$, so that 
\begin{equation} \label{SLLN, assumption for ALB}
\begin{aligned}
    \Pro_k^1\left( \limsup_{n\to\infty} \frac{\lambda_k(n)}{n} \leq  I_k \right) & = 1 \\
   \Pro_k^0\left( \limsup_{n\to\infty} \frac{-\lambda_k(n)}{n} \leq  J_k \right) & = 1, 
\end{aligned}
\end{equation}
and
\begin{align} \label{Complete convergence, assumption for AUB}
\begin{split}
  \forall \; \epsilon>0 \quad   \sum_{n=1}^\infty \Pro_k^1\left( \frac{\lambda_k(n)}{n} \leq I_k -\epsilon \right) & <\infty \\
  \forall \; \epsilon>0 \quad   \sum_{n=1}^\infty \Pro_k^0\left( \frac{-\lambda_k(n)}{n} \leq J_k - \epsilon \right) & <\infty.
\end{split}
\end{align}

\begin{remark}
By the Borel-Cantellil lemma, condition  \eqref{Complete convergence, assumption for AUB} implies 
\begin{equation*} 
\begin{aligned}
    \Pro_k^1\left( \liminf_{n\to\infty} \frac{\lambda_k(n)}{n} \geq I_k \right) & = 1 \\
    \Pro_k^0\left( \liminf_{n\to\infty} \frac{-\lambda_k(n)}{n} \geq  J_k \right) & = 1.
\end{aligned}
\end{equation*}
Therefore, conditions  \eqref{SLLN, assumption for ALB}-\eqref{Complete convergence, assumption for AUB}  imply
\begin{equation} \label{SLLN}
\begin{aligned}
    \Pro_k^1\left(\lim_{n\to\infty} \frac{\lambda_k(n)}{n}=I_k \right) & = 1 \\
    \Pro_k^0\left(\lim_{n\to\infty} \frac{\lambda_k(n)}{n}=-J_k \right) & = 1,
\end{aligned}
\end{equation}
and, as a result, they characterize $I_k$ and $J_k$.  
\end{remark}

\begin{remark}
In the special case that $ \{\lambda_k(n), \,n\in\bN\}$ has i.i.d. increments with mean $I_k>0$ under $\Pro_1^k$ and $-J_k<0$ under  $\Pro_0^k$,   by  Kolmogorov's Strong Law of Large Numbers   and  Chernoff's bound 
it follows that  conditions \eqref{SLLN, assumption for ALB}-\eqref{Complete convergence, assumption for AUB} are equivalent to \eqref{SLLN} and are automatically satisfied. 
\end{remark}

\subsection{Notations}
For any $\cA\subseteq [K]$, we denote by $\cI_\cA$
and   $\cJ_\cA$ the minimum of the numbers in $\{I_i: i \in \cA\}$  and $\{J_j: j\notin\cA\}$, respectively, i.e., 
\begin{equation} \label{KL_min}
 \cI_\cA\equiv \min_{i\in \cA} I_i \qquad \text{and} \qquad \cJ_\cA\equiv \min_{j\notin\cA} J_j.
 \end{equation}

We denote the   LLRs  at time  $n\in\bN$ in non-increasing order as $\lambda_{(k)}(n)$, $k\in [K]$, i.e.,
$$\lambda_{(1)}(n) \geq \lambda_{(2)}(n) \geq \cdots \geq \lambda_{(K)}(n),$$
where ties are ordered arbitrarily. Moreover, we denote by $p(n)$ the number of positive LLRs at time $n$, i.e., 
$$ p(n) \Def \big\vert \left\{ k\in[K]: \lambda_k(n)>0 \right\} \big\vert. $$

\section{The proposed sequential multiple testing procedure}  \label{section: the proposed test}
In this section we introduce the sequential multiple testing procedure that we propose in this work,   contrasting it with two existing procedures in the literature, a decentralized one and a synchronous one.

\subsection{A  decentralized procedure (parallel SPRT)}
When considering the  binary sequential testing problem in stream $k$ \textit{locally},  the most  natural  sequential test is arguably Wald's Sequential Probability Ratio Test  (SPRT),
\begin{equation}  \label{def: SPRT}
    \begin{aligned}
    \Tilde{T}_k& \equiv \inf\left\{ n\in\bN: \lambda_k(n)\notin(-b,a) \right\}, \\
    \Tilde{D}_k & \equiv 1\{   \lambda_k(\Tilde{T}_k)\geq a \big\},
    \end{aligned}
\end{equation}
where  $a,b>0$ are thresholds to be determined.
Thus, a natural sequential \emph{multiple} testing procedure is to simply apply an SPRT to each  binary testing problem, i.e., 
$\Tilde{\chi}\equiv (\Tilde{\bfT},\Tilde{\bfD} )$,  where 
\begin{equation}  \label{def: decentralized procedure}
    \begin{aligned}
    \tilde{\mathbf{T}} & \Def (\tilde{T}_1,\ldots,\tilde{T}_K),  \\
  \tilde{\mathbf{D}} & \Def (\tilde{D}_1,\ldots,\tilde{D}_K).
    \end{aligned}
\end{equation} 

This   procedure  has been considered, at least as a competitor,  in various works (e.g., \cite{De_Baron_Seq_Bonf,Malloy_Nowak_2014,Song_prior, PaperII}). It is clearly decentralized and, as we will see in Section \ref{sec: AO},  it achieves asymptotic optimality in the  subfamily of \text{decentralized} procedures for any given class of prior information,  $\Pi$. However, it 
does \textit{not} preserve this asymptotic optimality property  in  the general family of sequential multiple testing procedures considered in this work  
when $\Pi$ is of the form $\Pi_{l,u}$  \textit{apart from the case of no prior information} $(l=0,u=K)$.

 \subsection{The proposed procedure}
Suppose that we are given bounds on the number of signals, i.e.,  $\Pi$ is of the form $\Pi_{l,u}$.  Then, the proposed  procedure, $\hat{\chi}\equiv (\hat{\bfT},\hat{\bfD} )$, is defined as follows:
\begin{equation}  \label{def: general}
    \begin{aligned}
    \bfTs & \Def (\Ts_1,\ldots,\Ts_K), \text{ where  } \Ts_k  \Def \Ts_{k,1}\wedge \Ts_{k,2}, \;\; k \in [K], \\
    \bfDs & \Def \big\{k\in [K]: \Ts_k=\Ts_{k,1} \big\},
    \end{aligned}
\end{equation}
and the form of $\Ts_{k,1}$ and $\Ts_{k,2}$ depends on whether   $l=u$ or $l<u$.   When the number of signals is a priori known to be equal to  some $1\leq m\leq K-1$, i.e.,  $l=u\equiv m$, we set 
\begin{equation}  \label{def: gap rule}
    \begin{aligned}
    \Ts_{k,1} & \Def \inf\left\{ n\in\bN: \lambda_k(n)\geq \lambda_{(m+1)}(n) + c \right\}, \\
    \Ts_{k, 2} & \Def \inf\left\{ n\in\bN: \lambda_k(n)\leq \lambda_{(m)}(n) - d \right\}, 
    \end{aligned}
\end{equation}
where  $c,d>0$ are thresholds to be determined.  When the number of signals is not a priori known, i.e.,  $l<u$,  we combine the stopping rules in \eqref{def: SPRT} and  \eqref{def: gap rule}  and set
\begin{equation}  \label{def: gap-intersection rule}
    \begin{aligned}
    \Ts_{k,1} & \Def \inf\left\{ n\in\bN: \lambda_k(n)\geq \min\{ a, \; \lambda_{(l+1)}(n) + c\} \right\}, \\
    \Ts_{k, 2} & \Def \inf\left\{ n\in\bN: \lambda_k(n)\leq \max\{ -b, \; \lambda_{(u)}(n) - d \} \right\},
    \end{aligned}
\end{equation}
where  $a,b,c,d>0$ are again thresholds to be determined.\\

In the case of no prior information ($l=0$ and $u=K$),  thresholds $c$ and $d$ are inactive and   $\hat{\chi}$ reduces to the parallel SPRT,  introduced in \eqref{def: SPRT}-\eqref{def: decentralized procedure}, as  it is not possible,  for any $c,d>0$, $k\in[K]$, $n\in\bN$, to have 
  $$\lambda_k(n)\geq \lambda_{(1)}(n)+c \qquad   \text{or} \qquad  \lambda_k(n)\leq \lambda_{(K)}(n)-d. $$
    
    On the other hand,  in the presence of non-trivial  a priori bounds on the number of signals, i.e., when either $l>0$ or $u<K$,   $\hat{\chi}$  is not decentralized, as it requires comparisons between statistics of different streams. 
Indeed,  according to this scheme,    when $l>0$ (resp. $u<K$), 
a decision can be made in a stream once its  LLR statistic is  sufficiently larger  (resp. smaller) than the $(l+1)$-th (resp. $u$-th) largest LLR.

\subsection{A  synchronous procedure}
Assuming that $\Pi$ is of the form\ $\Pi_{l,u}$, 
the synchronous procedure, $\check{\chi}\equiv (\Tc, \bfDc),$   in   \cite{Song_prior} is defined as follows:
\begin{itemize}
\item When $l=u\equiv m$, 
\begin{equation} \label{def: complete gap rule}
\begin{aligned}
    \Tc \Def \inf\left\{ n\in\bN: \lambda_{(m)}(n) - \lambda_{(m+1)}(n) \geq c\vee d \right\},
\end{aligned}
\end{equation}
where  $c,d>0$ are thresholds to be determined, and $\bfDc$ is the subset of  the  $m$ streams with the largest LLRs at time $\Tc$.

\item When $l<u$, 
\begin{equation} \label{def: complete gap-intersection rule}
\begin{split}
  \Tc & \equiv \tau_1\wedge\tau_2\wedge\tau_3, \quad \text{where} \\
 \tau_1 & \equiv \inf\left\{ n\in\bN: \lambda_{(l+1)}(n) \leq \min \{-b, -c+ \lambda_{(l)}(n) \}  \right\} \\
    \tau_2 & \equiv \inf\left\{n\in\bN: \lambda_k(n) \notin (-b,a) \;\; \forall \,  k\in[K], \text{ and } l\leq p(n)\leq u \right\} \\
    \tau_3 & \equiv \inf\left\{ n\in\bN: \lambda_{(u)}(n)\geq  \max \{ a,  \, d+ \lambda_{(u+1)}(n)\} \right\},
\end{split}
\end{equation}
 where  $a,b,c,d>0$ are thresholds to be determined, 
and $\bfDc$ is the subset of the  $\left(p(\Tc)\vee l\right)\wedge u$ streams with the largest LLRs at time $\Tc$. \\
\end{itemize}

In the case of no prior information ($l=0$ and $u=K$),    
\begin{align*} 
  \Tc &= \inf\left\{n\in\bN: \lambda_k(n) \notin (-b,a) \;\; \forall\, k\in[K] \right\}, 
\end{align*}
and $\check\bfD$ is the subset of streams with positive LLRs at time $\Tc$. Therefore, in this case,  $\check{\chi}$   reduces to  the \textit{Intersection rule} introduced in \cite{De_Baron_Seq_Bonf} and  makes its  decisions once all LLRs are \textit{simultaneously} outside the interval  $(-b,a)$. 
Clearly, this is never sooner than the last decision time of the parallel SPRT with the same thresholds, $a$ and $b$.
In general, for any $l$ and $u$,  the last decision time of the proposed procedure occurs no later than the decision time of this synchronous procedure, i.e., 
\begin{equation} \label{max hat Tk leq check T}
    \max_{k \in [K]} \hat{T}_k \leq \check{T}, \quad \forall\;k\in[K],
\end{equation}
when the two procedures use the same thresholds $a,b,c, d$.

In  \cite{Song_prior} it was shown, in the case that each data stream consists of i.i.d. observations, that $\check{\chi}$ is asymptotically optimal in the subfamily of  \textit{synchronous} procedures,  i.e.,  it achieves  \eqref{def of AO for synchronous} asymptotically as $\alpha, \beta \to 0$  for every $ \cA \in \Pi_{l,u}$.  In Section \ref{sec: AO} we show that  this asymptotic optimality property  holds under the general distributional assumptions of Subection \ref{subsec: distributional assumptions}, but does not extend,  apart from a very specific setup, to the general  family of sequential multiple testing procedures that we consider in this work.

\section{Analysis} \label{sec: analysis}
In this section,  we show how the thresholds of the proposed procedure can be selected in order to control the two types of familywise error rates below arbitrary prescribed levels. Moreover, we establish asymptotic upper bounds
on its expected decision times 
in each stream as its thresholds go to infinity. 
For comparison purposes, we  also include the corresponding  results for the parallel SPRT  and the synchronous procedure introduced in the previous section. 

\subsection{Error control}
\begin{proposition}  \label{thm: a.s. finite and error control}
Suppose that \eqref{orthogonal} holds for every $k \in [K]$ and let   $\cA\subseteq [K]$, $a,b,c,d>0$. Then
    \begin{align}
        \text{FWE}_\cA^1(\Tilde{\chi}) & \leq |\cA^c|\, e^{-a}, \label{decentralized procedure, type I} \\
        \text{FWE}_\cA^2(\Tilde{\chi}) & \leq |\cA|\, e^{-b}. \label{decentralized procedure, type II}
    \end{align}
    If, also,  $\Pi=\Pi_{l,u}$ and $\cA\in \Pi_{l,u}$, then we have the following:
    \begin{itemize}
     \item If $l=u\equiv m$, then
        \begin{align}
        \text{FWE}_\cA^1(\hat{\chi}) & \leq m(K-m)\, e^{-c}, \label{gap rule, type I} \\
        \text{FWE}_\cA^2(\hat{\chi}) & \leq m(K-m)\, e^{-d}, \label{gap rule, type II} \\
        \text{FWE}_\cA^1(\check{\chi}) = \text{FWE}_\cA^2(\check{\chi}) &\leq m(K-m)\, e^{-c\vee d}. \label{synchronous gap rule, type-I and -II}
        \end{align}
        \item  If $l<u$, then
    \begin{align}
        & \text{FWE}_\cA^1(\hat{\chi}), \; \text{FWE}_\cA^1(\check{\chi}) \leq |\cA^c| \, (e^{-a} + |\cA|\,  e^{-c}), \label{gap-intersection rule, type I} \\
        & \text{FWE}_\cA^2(\hat{\chi}), \; \text{FWE}_\cA^2(\check{\chi}) \leq |\cA| \, (e^{-b} + |\cA^c| \, e^{-d}).  \label{gap-intersection rule, type II} 
    \end{align}
    \end{itemize}
\end{proposition} 
  
\begin{proof}
    See Appendix \ref{Proof related to AUB}. \\   
\end{proof}

The previous proposition provides  a concrete selection for  thresholds that guarantee prescribed familywise error rates. Indeed, let  $\alpha,\beta\in(0,1)$. Then,  for any given $\Pi\subseteq 2^{[K]}$, we have          $\Tilde{\chi} \in\Delta(\alpha,\beta,\Pi)$ when 
    \begin{align} \label{intersection rule, a, b_general}
    \begin{split}
        a &= |\log\alpha| + \max_{\cA\in \Pi}\, \log |\cA^c| \\        
      b &= |\log\beta| +  \max_{\cA\in \Pi}\, \log |\cA|.
      \end{split}
    \end{align}    
Similarly, for any given $l,u$, 
we have the following:
\begin{itemize}
   \item If $l=u\equiv m$, then $\hat{\chi}$ and $\check{\chi}$ both belong to $\Delta(\alpha,\beta,\Pi_{l,u})$ when
        \begin{align} \label{gap rule, c, d}
            \begin{split}
            c &= |\log\alpha| + \log(m(K-m)) \\
             d &= |\log\beta| + \log(m(K-m)).
            \end{split}
        \end{align}
    \item If $l<u$,
    then $\hat{\chi}$ and $\check{\chi}$ both belong to $\Delta(\alpha,\beta,\Pi_{l,u})$ when
    \begin{align} \label{gap-inter rule, a, b, c, d}
    \begin{split}
        a &= |\log\alpha| + \log K \\
        b &= |\log\beta| + \log K \\
        c &=|\log\alpha|+\log((K-l) \, K) \\
        d &=|\log\beta|+\log(u\, K).
    \end{split}
    \end{align} 
    \end{itemize}
    
These thresholds suffice for the asymptotic optimality theory  we develop in the next section, however they can be quite conservative for practical use.  Alternatively, 
one may use Monte Carlo simulation to determine the thresholds that equate, at least approximately, the maximum familywise error rates to their target levels. When the target levels are very small, this can be done efficiently 
using importance sampling, as we discuss in Subsection \ref{remark: importance sampling}.\\

\begin{proposition} \label{prop: AUB}
Suppose that \eqref{Complete convergence, assumption for AUB} holds and  let $\cA \subseteq [K]$, $i \in \cA$,  $j\notin\cA$.  
Then, as $a, b \to \infty$, 
\begin{align} \label{AUB, decentralized}
    & \Exp_\cA[\tilde{T}_i] \lesssim \frac{a}{I_i} 
    \qquad \text{and} \qquad
    \Exp_\cA[\tilde{T}_j] \lesssim \frac{b}{J_j}. 
\end{align}      
If, also, $\Pi=\Pi_{l,u}$ and $\cA\in \Pi_{l,u}$, then we have the following:
\begin{itemize}
    \item If $l=u$, then, as $c,d\to \infty$,
    \begin{gather}
    \Exp_\cA[\Ts_i] \lesssim \frac{c}{I_i + \cJ_\cA}, \qquad
    \Exp_\cA[\Ts_j] \lesssim \frac{d}{J_j+\cI_\cA}, \label{AUB, proposed, gap} \\
    \Exp_\cA[\check T]\lesssim \frac{c\vee d}{\cI_\cA+\cJ_\cA}. \label{AUB, synchronous, gap}
    \end{gather}

    \item If $l<u$, then, as $a,b,c,d\to\infty$,
    \begin{equation} \label{AUB, proposed, gap-inter}
    \begin{aligned}
        \Exp_\cA[\Ts_i] & \lesssim \begin{cases}
        \begin{aligned}
            & \frac{a}{I_i}\wedge\frac{c}{I_i+\cJ_\cA}, && \text{if } |\cA|=l \\
            & \;\; \frac{a}{I_i}, && \text{if } l<|\cA|\leq u,
        \end{aligned}
        \end{cases} \\
        \Exp_\cA[\Ts_j] & \lesssim \begin{cases}
        \begin{aligned}
            & \;\; \frac{b}{J_j}, && \text{if } l\leq |\cA|<u \\
            & \frac{b}{J_j}\wedge\frac{d}{J_j+\cI_\cA}, && \text{if } |\cA|=u,
        \end{aligned}
        \end{cases}
    \end{aligned}
    \end{equation}
    and 
    \begin{equation} \label{AUB, synchronous, gap-inter}
        \Exp_\cA[\Tc] \lesssim \begin{cases}
        \begin{aligned}
            & \max\left\{ \frac{b}{\cJ_\cA},\; \frac{a}{\cI_\cA}\wedge\frac{c}{\cI_\cA+\cJ_\cA} \right\}, && \text{if } |\cA|=l \\
            & \max\left\{ \frac{a}{\cI_\cA}, \; \frac{b}{\cJ_\cA} \right\}, && \text{if } l<|\cA|<u \\
            & \max\left\{ \frac{a}{\cI_\cA},\; \frac{b}{\cJ_\cA}\wedge\frac{d}{\cJ_\cA+\cI_\cA} \right\}, && \text{if } |\cA|=u.
        \end{aligned}
        \end{cases}
    \end{equation}
    
    In particular, as $a,b,c,d\to\infty$ so that $a\sim c$ and $b\sim d$, 
    \begin{equation} \label{AUB, proposed, gap-inter, a sim c}
    \begin{aligned}
        \Exp_\cA[\Ts_i] & \lesssim \frac{a}{I_i+\cJ_\cA\cdot\bfone\{|\cA|=l\}}, \\
        \Exp_\cA[\Ts_j] & \lesssim \frac{b}{J_j+\cJ_\cA\cdot\bfone\{|\cA|=u\}}, 
    \end{aligned}
    \end{equation}
    and 
    \begin{equation} \label{AUB, synchronous, gap-inter, a sim c}
        \Exp_\cA[\Tc] \lesssim 
        \max\left\{ \frac{a}{\cI_\cA+\cJ_\cA\cdot\bfone\{|\cA|=l\}},\;\frac{b}{\cJ_\cA+\cI_\cA\cdot\bfone\{|\cA|=u\}} \right\}. \\
    \end{equation}
\end{itemize}
\end{proposition}
\begin{proof}
    See Appendix \ref{Proof related to AUB}. \\
\end{proof}

The above \textit{first-order} asymptotic upper bounds are free of  $l$ and $u$.  In the next proposition we consider an i.i.d. setup and obtain  second-order terms that  depend on $l$ (resp.  $u$) when the true number of signals  is $l$ (resp. $u$).

\begin{proposition} \label{prop: higher-order AUB}
Suppose  that  the increments of $\{\lambda_k(n), n \in \bN\}$ are i.i.d. with finite variance and mean $I_k>0$ under $\Pro_k^1$ and $-J_k<0$ under $\Pro_k^0$  for every $k\in[K]$.  Let $\cA\subseteq [K]$, $i\in\cA$, $j\notin\cA$. Then, as $a,b \to \infty$, 
       \begin{align} \label{AUB with higher-order term, decentralized}
       \Exp_\cA[\tilde{T}_i] \leq \frac{a}{I_i} +O(1) \qquad \text{and} \qquad \Exp_\cA[\tilde{T}_j] \leq \frac{b}{J_j}+O(1).
    \end{align}      
    If, also,  $\Pi=\Pi_{l,u}$ and $\cA\in \Pi_{l,u}$, then we have the following:
\begin{itemize}       
     \item  If $l=u\equiv m$, then,
        as $c,d\to\infty$,        
        \begin{align}
           \Exp_\cA[\Ts_i] &
            \leq \frac{c}{I_i+\cJ_\cA} + O\left( (K-m) \sqrt{c} \right) , \label{AUB with higher-order term, proposed, gap, i} \\
            \Exp_\cA[\Ts_j] & \leq \frac{d}{J_j+\cI_\cA} + O( m\sqrt{d} ),  \label{AUB with higher-order term, proposed, gap, j} \\
            \Exp_\cA[\Tc] & \leq \frac{c\vee d}{\cI_\cA+\cJ_\cA} + O\left( m(K-m)\sqrt{c\vee d} \right). \label{AUB with higher-order term, synchronous, gap}
        \end{align}
    \item If $l<u$, then, as $a,b,c,d\to\infty$ so that $|c-a| = O(1)$ and $|d-b|=O(1)$,
    \begin{align} 
        \Exp_\cA[\Ts_i] & \lesssim \begin{cases}
        \begin{aligned}
            & \frac{a}{I_i+\cJ_\cA}+O\left( (K-l)\sqrt a \right), && \text{if } |\cA|=l,  \\
            & \;\; \frac{a}{I_i} + O(1), && \text{if } l<|\cA|\leq u.
        \end{aligned}
        \end{cases} \label{AUB with higher-order term, proposed, gap-inter, i} \\
        \Exp_\cA[\Ts_j] & \lesssim \begin{cases}
        \begin{aligned}
            & \;\; \frac{b}{J_j}+O(1), && \text{if } l\leq|\cA|<u, \\
            & \frac{b}{J_j+\cI_\cA}+O(u\sqrt b ), && \text{if } |\cA|=u.
        \end{aligned}
        \end{cases} \label{AUB with higher-order term, proposed, gap-inter, j}
    \end{align}
      \end{itemize}      
\end{proposition}

\begin{proof}
    See Appendix \ref{Proof related to AUB}. \\
\end{proof}

\begin{remark} \label{remark: intuition from higher-order terms} 
The upper bounds in \eqref{AUB with higher-order term, proposed, gap, i}  and \eqref{AUB with higher-order term, proposed, gap-inter, i}  (resp.  \eqref{AUB with higher-order term, proposed, gap, j} and \eqref{AUB with higher-order term, proposed, gap-inter, j})   suggest that the expected decision time of the proposed test in a signal (resp. noise)  stream should be  decreasing in $l$ (resp. increasing in $u$) when the true number of signals is $l$ (resp. $u$), and independent of $l$ and  $u$  when the true number of signals is larger than $l$ (resp. smaller than $u$).  On the other hand, when the number of signals is a priori known i.e., $l=u\equiv m$, the  upper bound in \eqref{AUB with higher-order term, synchronous, gap}  suggests that the expected decision time of the synchronous test  increases  as $m$ approaches $K/2$. The intuition from these bounds  will be corroborated in the simulation studies of Section \ref{section: simulation studies}. 
\end{remark}

\section{Asymptotic optimality} \label{sec: AO}
In this section we establish the asymptotic optimality theory of the paper.  First,  we state a universal,  asymptotic  (as $\alpha,\beta\to 0$) lower bound on $\cL_{k,\cA} (\alpha, \beta, \Pi)$, defined in \eqref{def of AO}, for any  $k\in[K]$,  $A\in\Pi$, and any  class of prior information $\Pi$,  
and then we  show  that   it  is attained  by the proposed procedure simultaneously for every $A \in \Pi$ and $k\in[K]$ when $\Pi$ is of the form $\Pi_{l,u}$.  \\

\begin{lemma} \label{Lemma, ALB}
Let $\Pi$  satisfy \eqref{no trivial}, $\cA\in \Pi$,  $i\in \cA$, $j \notin\cA$. 
If condition \eqref{SLLN, assumption for ALB} holds for every $k \in [K]$, then,      as $\alpha,\beta\to 0$, 
    \begin{align} 
           & 
          \cL_{i,\cA} (\alpha, \beta, \Pi)  \gtrsim \frac{|\log\alpha|}{\min \{I_{\cA,\cC}:  \cC\in \Pi, \,  i\notin \cC \}  }, \label{ALB, i, general} \\
           & 
           \cL_{j, A} (\alpha, \beta, \Pi)  \gtrsim
            \frac{|\log\beta|}{\min \{ I_{\cA,\cC}: \cC\in \Pi, \,  j\in \cC \} }, \label{ALB, j, general}
    \end{align} 
where 
\begin{equation} \label{Definition of D^A,C}
    I_{\cA,\cC} \Def \sum_{k\in \cA\backslash\cC} I_k + \sum_{k\in \cC\backslash\cA} J_k.
\end{equation}

\end{lemma}
\begin{proof}
    See Appendix \ref{Proof related to ALB}. \\
\end{proof}

\begin{theorem} \label{thm: AO} 
Let $l,u$ satisfy \eqref{lu}. Suppose that the thresholds $a, b,c,d$ of $\hat\chi$ are selected so that $\hat\chi \in \Delta(\alpha, \beta, \Pi_{l,u})$ for any $\alpha,\beta\in(0,1)$, and  also
\begin{equation} \label{thres_proposed}
a,c \sim |\log \alpha| \quad \text{and}  \quad b,d \sim |\log \beta| \quad \text{as} \quad \alpha, \beta \to 0,
\end{equation}
 e.g.,       according to \eqref{gap rule, c, d} when $l=u$ and   \eqref{gap-inter rule, a, b, c, d} when $l<u$.
 
 If conditions \eqref{SLLN, assumption for ALB}-\eqref{Complete convergence, assumption for AUB} hold for every $k \in [K]$, then,   as $\alpha,\beta\to 0$, 
    \begin{align} 
       \Exp_\cA[\Ts_i] & \sim   \cL_{i,\cA} (\alpha, \beta, \Pi_{l,u})  \sim \frac{|\log\alpha|}{I_i+\cJ_\cA \cdot \bfone\left\{ |\cA|=l \right\}},  \label{AO, i} \\
     \Exp_\cA[\Ts_j] & \sim   \cL_{j,\cA} (\alpha, \beta, \Pi_{l,u})  \sim \frac{|\log\beta|}{J_j+\cI_\cA \cdot \bfone\left\{ |\cA|=u \right\}}, \label{AO, j}
    \end{align}
simultaneously for every   $i\in\cA$,  $j\notin\cA$, $\cA\in\Pi_{l,u}$.
\end{theorem}
\begin{proof} 
    See Appendix \ref{Proof related to ALB}. \\
\end{proof}

\subsection{Comparison with  decentralized procedures} 
We next establish the asymptotic optimality of  the parallel SPRT, $\tilde\chi$, defined in \eqref{def: SPRT}-\eqref{def: decentralized procedure}, in the subfamily of  \textit{decentralized} procedures,  for any given class of prior information.

\begin{theorem} \label{thm: AA for decentralized procedures}
Let $\Pi$ satisfy \eqref{no trivial}. Suppose that the thresholds $a, b$ of $\Tilde{\chi}$ are selected so that $\Tilde\chi \in \Delta(\alpha, \beta, \Pi)$  for any  $\alpha,\beta\in(0,1)$, and also
 \begin{equation} \label{thres_decentralized}
a \sim |\log \alpha| \quad \text{and} \quad b \sim |\log \beta| \quad \text{as} \quad \alpha, \beta \to 0,
\end{equation}
e.g.,       according to \eqref{intersection rule, a, b_general}. 

If conditions \eqref{SLLN, assumption for ALB}-\eqref{Complete convergence, assumption for AUB} hold for every $k \in [K]$, then, as $\alpha,\beta\to 0$,
\begin{equation} \label{AA for decentralized}
\begin{aligned}
    \Exp_\cA[\tilde{T}_i] & \sim \cL'_{i,\cA}(\alpha, \beta, \Pi)   \sim \frac{|\log\alpha|}{I_i},\\
    \Exp_\cA[\tilde{T}_j] & \sim \cL'_{j,\cA}(\alpha, \beta, \Pi)  \sim \;  \frac{|\log\beta|}{J_j},
\end{aligned}
\end{equation}
simultaneously for every $i\in\cA$, $j\notin\cA$, $\cA\in\Pi$.
\end{theorem}

\begin{proof}
    See Appendix \ref{Proof related to ALB}. \\
\end{proof}

Using Theorems  \ref{thm: AO} and  \ref{thm: AA for decentralized procedures}  we can  now evaluate the asymptotic relative  efficiencies    in \eqref{AE'}.\\

\begin{corollary} \label{coro1}
Let $\Pi$ satisfy \eqref{no trivial}. Suppose that conditions \eqref{SLLN, assumption for ALB}-\eqref{Complete convergence, assumption for AUB} hold for every $k \in [K]$.  Then, for every  $i\in\cA$, $j\notin\cA$, $\cA\in\Pi_{l,u}$,
\begin{equation} \label{AE' relative to decentralized procedures}
\begin{split}    
    \Ae'_{i,\cA}(l,u) &= 
    \begin{cases}
    \begin{aligned}
    & I_i/(I_i+\cJ_\cA), && \text{when} \quad |\cA|=l\\
    & 1,   && \text{otherwise},
    \end{aligned}
    \end{cases}
    \\
    \Ae'_{j,\cA}(l,u) &= 
    \begin{cases}
    \begin{aligned}    
    &  J_j/(J_j+\cI_\cA), && \text{when} \quad |\cA|=u\\
    & 1,   && \text{otherwise}.
    \end{aligned}
     \end{cases}
\end{split}
\end{equation}
When, in particular, $l=0$, 
$$  \Ae'_{i,\cA}(l,u) =1, \quad \forall \; i \in A,  $$
 and when $u=K$, 
 $$  \Ae'_{j,\cA}(l,u) =1, \quad \forall \; j \notin A.$$
\end{corollary}

\begin{proof}
It suffices to  compare \eqref{AO, i}-\eqref{AO, j} and \eqref{AA for decentralized}. \\
\end{proof}

If the multiple testing problem is \textit{homogeneous} in the sense that 
\begin{equation} \label{homogeneous}
    I_k=I \quad \text{and} \quad J_k=J, \quad \forall\; k\in [K],
\end{equation}
then \eqref{AE' relative to decentralized procedures} reduces to 
\begin{equation} \label{AE' in homo setup}
\begin{split}    
    \Ae'_{i,\cA}(l,u) &= 
    \begin{cases}
    \begin{aligned}
    & I/(I+J), && \text{when} \quad |\cA|=l\\
    & 1,   && \text{otherwise},
    \end{aligned}
    \end{cases}
    \\
    \Ae'_{j,\cA}(l,u) &= 
    \begin{cases}
    \begin{aligned}    
    & J/(I+J), && \text{when} \quad |\cA|=u\\
    & 1,   && \text{otherwise}.
    \end{aligned}
     \end{cases}
\end{split}
\end{equation} 
If the multiple testing problem is also \textit{symmetric} in the sense that 
\begin{equation} \label{symmetric}
    I_k=J_k, \quad \forall\; k\in [K],
\end{equation}
then 
\begin{equation} \label{AE' in homo and symm setup}
\begin{split}    
    \Ae'_{i,\cA}(l,u) &= 
    \begin{cases}
    \begin{aligned}
    & 1/2, && \text{when} \quad |\cA|=l\\
    & 1,   && \text{otherwise},
    \end{aligned}
    \end{cases}
    \\
    \Ae'_{j,\cA}(l,u) &= 
    \begin{cases}
    \begin{aligned}    
    & 1/2, && \text{when} \quad |\cA|=u\\
    & 1,   && \text{otherwise}.
    \end{aligned}
     \end{cases}
\end{split}
\end{equation}

Theorems  \ref{thm: AO} and  \ref{thm: AA for decentralized procedures} and  Corollary \ref{coro1} imply that 
when the  thresholds of the  parallel SPRT and the proposed test  are selected so that \eqref{thres_proposed} and \eqref{thres_decentralized} hold, the two tests induce, asymptotically as $\alpha, \beta \to 0$,  the same expected decision time  in every signal (resp. noise) stream apart from when the true number of signals is equal to its a priori lower (resp. upper) bound. In the latter case,  the expected decision time of the  parallel SPRT is larger in every signal (resp. noise) stream and,  in particular, \textit{twice as large}  when 
\eqref{homogeneous} and \eqref{symmetric} hold.\\

Finally, we stress that the above comparisons are  only valid to a first-order asymptotic approximation as $\alpha, \beta \to 0$.  The actual, i.e., non-asymptotic, relative efficiencies of the parallel SPRT over the proposed test are computed  in various simulation studies in Section \ref{section: simulation studies}, where they are compared with their limiting values in  \eqref{AE' relative to decentralized procedures}.

\subsection{Comparison with  synchronous procedures} 
We next establish the  asymptotic optimality of the  synchronous procedure, defined  in \eqref{def: complete gap rule}-\eqref{def: complete gap-intersection rule}, in the subfamily of \textit{synchronous} procedures,  extending and generalizing the corresponding result in \cite{Song_prior} that applies only to the i.i.d. setup and under a second moment assumption on the log-likelihood ratios.  
    
\begin{theorem} \label{thm: AA for synchronous stopping}
Let $l,u$ satisfy \eqref{lu}. Suppose that the thresholds $a, b,c,d$ of $\check\chi$ are selected so that $\check{\chi} \in \Delta(\alpha, \beta, \Pi_{l,u})$ for any  $\alpha,\beta\in(0,1)$ and also so that  \eqref{thres_proposed} holds. 

If conditions \eqref{SLLN, assumption for ALB}-\eqref{Complete convergence, assumption for AUB} hold for every $k \in [K]$, then, as $\alpha,\beta\to 0$, 
    \begin{equation} \label{AA for the complete rule}
    \begin{aligned}
        \Exp_\cA[\Tc] & \sim \cL''_\cA(\alpha, \beta, \Pi_{ l,u} )  \\
        & \sim \max\left\{ \frac{|\log\alpha|}{\cI_\cA+\cJ_\cA \cdot \bfone\left\{ |\cA|=l \right\}},\;  \frac{|\log\beta|}{\cJ_\cA+\cI_\cA\cdot \bfone\left\{ |\cA|=u \right\}}\right\}
    \end{aligned}
    \end{equation}
simultaneously    for every $\cA\in\Pi_{l,u}$.
\end{theorem}
\begin{proof}
    See Appendix \ref{Proof related to ALB}. \\
\end{proof}

From  Theorems \ref{thm: AO} and \ref{thm: AA for synchronous stopping} it follows that  the optimal expected decision time in the subfamily of synchronous procedures agrees, to a first-order asymptotic approximation as $\alpha, \beta \to 0$,  with the maximum (with respect to the streams)
optimal expected  decision time  in  the general family of sequential multiple testing procedures. This is the content of  the following corollary. 

\begin{corollary}
Let  $l,u$ satisfy \eqref{lu} and suppose that 
conditions \eqref{SLLN, assumption for ALB}-\eqref{Complete convergence, assumption for AUB} hold for every $k\in[K]$.  Then,  as $\alpha,\beta\to 0$,
\begin{equation} \label{wait for the slowest one}
    \cL''_\cA(\alpha, \beta, \Pi_{ l,u} )  \sim \max_{k\in[K]} \cL_{k,\cA}(\alpha, \beta, \Pi_{ l,u})
\end{equation}
simultaneously    for every $\cA\in\Pi_{l,u}$.
\end{corollary}

\begin{proof}
Compare \eqref{AO, i}-\eqref{AO, j} and \eqref{AA for the complete rule}. \\
\end{proof}

Using  Theorems \ref{thm: AO} and \ref{thm: AA for synchronous stopping}, we next  evaluate the asymptotic relative  efficiencies  in \eqref{AE''}.

\begin{corollary} \label{coro: AE''}
  Let  $l,u$ satisfy \eqref{lu},  $A\in\Pi_{l,u}$, $i\in A$, $j\notin A$.  Suppose that  conditions \eqref{SLLN, assumption for ALB}-\eqref{Complete convergence, assumption for AUB} hold for every $k \in [K]$.
\begin{enumerate}
\item[(i)] If   $\alpha,\beta\to 0$ so that  $|\log \alpha| \ll |\log \beta|$, then
\begin{align*} 
\Ae''_{i,\cA}(l,u) &=0, \\
 \Ae''_{j,\cA}(l,u) 
 & = \frac{\cJ_\cA+\cI_\cA \cdot\bfone\{|\cA|=u\}}{J_j+\cI_\cA\cdot\bfone\{|\cA|=u\}}.
 \end{align*}
\item[(ii)] If   $\alpha,\beta\to 0$ so that $|\log\alpha|\gg|\log\beta|$, then
\begin{align*}  \Ae''_{i,\cA}(l,u) 
& =\frac{\cI_\cA+\cJ_\cA\cdot\bfone\{|\cA|=l\}}{I_i+\cJ_\cA\cdot\bfone\{|\cA|=l\}} ,\\
\Ae''_{j,\cA}(l,u) &=0. 
 \end{align*}
\item[(iii)] If $\alpha,\beta\to 0$ so that  
$|\log\alpha|\sim r|\log\beta|$ for some $r>0$, then 
    \begin{equation} \label{AE''i or j, cA, with r}
    \begin{aligned}
        \Ae''_{i,\cA}(l,u) 
        & = \frac{r/(I_i+\cJ_\cA\cdot\bfone\{|\cA|=l\})}{\max\big\{ r/(\cI_\cA+\cJ_\cA\cdot\bfone\{|\cA|=l\}),\, 1/(\cJ_\cA+\cI_\cA\cdot\bfone\{|\cA|=u\}) \big\}}, \\
        \Ae''_{j,\cA}(l,u) 
        & = \frac{1/(J_j+\cI_\cA\cdot\bfone\{|\cA|=u\})}{\max\big\{ r/(\cI_\cA+\cJ_\cA\cdot\bfone\{|\cA|=l\}),\, 1/(\cJ_\cA+\cI_\cA\cdot\bfone\{|\cA|=u\}) \big\}}.
    \end{aligned}
    \end{equation}
\end{enumerate}
\end{corollary}
\begin{proof}
Compare \eqref{AO, i}-\eqref{AO, j} and \eqref{AA for the complete rule}. \\
\end{proof}

\begin{remark}
Suppose that the multiple testing problem is \textit{homogeneous}, i.e.,  \eqref{homogeneous} holds, and also  
\begin{equation} \label{balanced}
    |\log\alpha|\sim |\log\beta| \quad \text{as} \quad \alpha,\beta\to 0.
\end{equation}
In this case,  if the  number of signals is a priori known, i.e., $l=u$, then 
\begin{equation} \label{asy_opt}
  \Ae''_{k,\cA}(l,u)=1, \quad \forall\; k \in [K].
  \end{equation}
On the other hand,  if the  number of signals is not a priori known, i.e., $l<u$, then
\begin{center}
    \begin{tabular}{c|c|c|c}
    & $|\cA|=l$ & $l<|\cA|<u$ & $|\cA|=u$ \\ \hline 
    $\Ae''_{i,\cA}(l,u)$ &  $J/(I+J)$ & $(I\wedge J)/I$ & 1 \\ \hline
    $\Ae''_{j,\cA}(l,u)$ &  1 & $(I\wedge J)/J$ & $I/(I+J)$ \\ \hline
    \end{tabular}
\end{center}

Suppose, also, that the multiple testing problem is \textit{symmetric}, i.e., \eqref{symmetric} holds, then 
\begin{equation} \label{AE'' in homo and symm setup}
\begin{split}
    \Ae''_{i,\cA}(l,u) &= 
    \begin{cases}
    \begin{aligned}
    & 1/2, && \text{when} \quad |\cA|=l<u\\
    & 1,   && \text{otherwise},
    \end{aligned}
    \end{cases}
    \\
    \Ae''_{j,\cA}(l,u) &= 
    \begin{cases}
    \begin{aligned}    
    & 1/2, && \text{when} \quad |\cA|=u>l\\
    & 1,   && \text{otherwise}. \\
    \end{aligned}
     \end{cases}
\end{split}
\end{equation}\\
\end{remark}

Theorems  \ref{thm: AO} and \ref{thm: AA for synchronous stopping} and  Corollary \ref{coro: AE''} imply that when the thresholds of  synchronous and the proposed test  are selected so that \eqref{thres_proposed} holds, then the following hold:
\begin{itemize}
\item When $|\log\alpha|$ diverges at a slower  (resp. faster) rate than  $|\log\beta|$, the expected decision time of the synchronous test  is  asymptotically   much larger than that of the proposed test in a signal (resp. noise) stream and asymptotically larger  in noise (resp. signal) stream $j$ (resp. $i$) unless $J_j$ (resp. $I_i$) is equal to the minimum $\cJ_\cA$ (resp. $\cI_\cA$), defined in \eqref{KL_min}. 
\item  When $|\log\alpha|$ and $|\log\beta|$ are of the same order of magnitude and the multiple testing problem is homogeneous,  then  the expected decision time of the synchronous test is asymptotically the same as that of the proposed test  in every signal (resp. noise) stream when $|\cA|=u$ (resp. $|\cA|=l$).
When the multiple testing problem is also symmetric, the expected decision time of the synchronous test is asymptotically the same as that of the proposed test in every signal (resp. noise) stream also when 
$l<|\cA|<u$, and  twice as large when 
$l=|\cA|<u$ (resp. $l<|\cA|=u$).
\end{itemize}

We reiterate that the above comparisons are valid to a first-order asymptotic approximation as $\alpha, \beta \to 0$.  The 
actual, i.e., non-asymptotic, relative efficiencies of the synchronous test over the proposed  test are computed  in various simulation studies in Section \ref{section: simulation studies}, where they are compared with their limiting values in  \eqref{AE''i or j, cA, with r}.

\section{Simulation studies} \label{section: simulation studies}
In  this section we compare  the proposed test, $\hat\chi$,   the parallel SPRT, $\Tilde{\chi}$,   and the  synchronous test, $\check\chi$,  in various simulation studies  where  each  $X_k$ is a sequence of i.i.d. Gaussian random variables with variance 1 and mean $0$ (resp. $\mu_k$) under $\Pro_k^0$ (resp. $\Pro_k^1$). Thus, for each $k \in [K]$,  
$$\lambda_k(n)=\mu_k \sum_{i=1}^n \big(X_k(i)-\mu_k/2\big), \quad n \in \bN, $$
and conditions  \eqref{SLLN, assumption for ALB}-\eqref{Complete convergence, assumption for AUB} are satisfied  with $I_k=J_k=  \mu_k^2/2$. Therefore, the multiple testing problem is symmetric, i.e., \eqref{symmetric} holds, but not necessarily homogeneous. Indeed, we consider two cases for the  means under the alternative hypotheses, 
\begin{itemize}
\item a homogeneous one, where
\begin{align} \label{mu_homogeneous}
\mu_k=\mu, \quad k \in [K], 
\end{align}
\item and  a non-homogeneous one, where there is some $\phi \in (0,1)$ so that 
\begin{equation} \label{mu_nonhomogeneous}
\begin{aligned}
    \mu_k &= \phi \, \mu, \quad 1 \leq k\leq K/2, \\
    \mu_k &=\mu, \quad K/2< k \leq K.
\end{aligned}
\end{equation}
\end{itemize}

Moreover, we consider two cases  regarding prior information on the number of signals: 
\begin{itemize}
\item  the number of signals is a priori known, i.e., $l=u$,
\item  the number of signals is at least $l$ and at most $u$ where $l<u$ and  $l+u=K$.
\end{itemize}

\subsection{Design}
For each of the three tests under consideration,  we let  only  one free parameter for its design. Specifically,   for the parallel SPRT we set $a=b$, whereas for the proposed and the synchronous test,  we set $ c=d$ when  $l=u$, and 
\begin{equation} \label{thresholds when l<u}
    a = b, \quad c = a+\log(K-l), \quad d= b+\log u, \quad \text{when} \quad l<u.
\end{equation}
 For a wide range of values of the free parameter, 
we compute,  for each $\cA\in\Pi_{l,u}$, the expected decision times in each stream, i.e.,
\begin{align} \label{quantities of E[T]}
       \Exp_\cA[\Ts_k], \quad  \Exp_\cA[\Tilde{T}_k], \quad \Exp_\cA[\Tc], \quad \forall \; k\in[K], 
\end{align}
and the corresponding familywise error rates, i.e.,
\begin{equation} \label{quantities of FWER}
     \text{FWE}_\cA^i(\hat\chi),\quad 
     \text{FWE}_\cA^i(\Tilde\chi),  \quad \text{FWE}_\cA^i(\check\chi), \quad  i\in\{0,1\}.
\end{equation}
From the latter  we obtain the maximum 
familywise error rates
\begin{align*} 
\begin{split}
    \ha &\Def\max_{\cA\in\Pi_{l,u}} \text{FWE}_\cA^1(\hat\chi) \qquad \text{and} \qquad 
    \hb  \Def\max_{\cA\in\Pi_{l,u}} \text{FWE}_\cA^2(\hat\chi),\\
    \ta &\Def\max_{\cA\in\Pi_{l,u}} \text{FWE}_\cA^1(\Tilde\chi)  \qquad \text{and} \qquad
    \tb \Def\max_{\cA\in\Pi_{l,u}} \text{FWE}_\cA^2(\Tilde\chi) ,\\
    \ca & \Def\max_{\cA\in\Pi_{l,u}} \text{FWE}_\cA^1(\check\chi)  \qquad \text{and} \qquad
     \cb \Def\max_{\cA\in\Pi_{l,u}} \text{FWE}_\cA^2(\check\chi).
     \end{split}
\end{align*}

\subsection{Computation}
\label{remark: importance sampling}
For the computation of  the expected decision times in \eqref{quantities of E[T]}  we used plain Monte Carlo, but for the estimation of the familywise error rates in  \eqref{quantities of FWER} we used importance sampling.  To be  specific,  for each $A \in \cP_{l,u}$, the familywise error rates of the parallel SPRT are given by 
\begin{equation*}
\begin{aligned}
    \text{FWE}_\cA^1(\tilde\chi) & = 1-\prod_{j\notin\cA} \big(1-\Pro_j^0(D_j=1)\big) \\
    \text{FWE}_\cA^2(\tilde\chi) & = 1-\prod_{i\in\cA} \big(1-\Pro_i^1(D_i=0)\big),
\end{aligned}
\end{equation*}
and it suffices to estimate the error probabilities of the SPRT in each individual testing problem. This can be done using the importance sampling approach described  in \cite{Siegmumd_IS}.

A generalization of this importance sampling approach is proposed in  \cite[Section 4]{Song_prior} for the estimation of the familywise error rates of the synchronous test, and we follow a similar approach for the estimation of the familywise error rates of the proposed test.   To be specific, let us consider the estimation of the type-I familywise error rate under some $\cA\subseteq [K]$. Then, 
 we select the importance sampling distribution $\Pro_\cA^*$ as  in \cite[Section 4]{Song_prior}
and we evaluate, via plain Monte Carlo, the right-hand side of the following identity:
\begin{align*}
  \text{FWE}_\cA^1(\hat{\chi})&\equiv  \Pro_\cA(\hat{\bfD}\backslash\cA\neq \emptyset)  =
  \Exp_\cA^*\left[ \left( \frac{d\Pro_\cA^*}{d\Pro_\cA}(\cF(\tau_A^1)) \right)^{-1} ; \; \hat{\bfD}\backslash\cA\neq \emptyset  \right], 
  \end{align*} 
where  $\tau_A^1$ is  the first time the proposed test commits a type-I error when the true subset of signals is $\cA$, i.e., 
$$ \tau_A^1 \equiv \inf\left\{ n\in \bN: \exists\,j\notin\cA, \; \Ts_j=n \text{ and } \Ds_j=1 
\right\}.$$
We note that $\tau_A^1$ can be replaced in the above identity by $\max_{k\in[K]}\hat T_k$, but in that case the variance of the resulting estimator  increases considerably. 

Overall, for the computation of each quantity in 
\eqref{quantities of E[T]}-\eqref{quantities of FWER}   we used   $10^4$ Monte Carlo replications and obtained  relative errors below  $0.5\%$   in all cases.

\subsection{Comparisons}
For  each setup under consideration and each  $\cA\in\Pi_{l,u}$, we plot the expected decision time of each of the three tests in every signal (resp. noise) stream against the absolute value of the base-10 logarithm of the maximum 
 type-I (resp. -II) familywise error rate. That is,  for each $i\in\cA$ we plot 
\begin{align} \label{E[T] against |lg alpha|} 
\begin{split}
    \Exp_\cA[\Ts_i] \quad \text{against} \quad & |\lg\ha|, \\
      \Exp_\cA[\Tilde{T}_i] \quad \text{against} \quad & |\lg\ta|, \\
         \Exp_\cA[\Tc] \quad \text{against} \quad & |\lg\ca|,
         \end{split}
         \end{align} 
and for each $j\not\in\cA$ we plot
\begin{equation} \label{E[T] against |lg beta|}
\begin{aligned}
    \Exp_\cA[\Ts_j] \quad \text{against} \quad & |\lg\hb|, \\
    \Exp_\cA[\Tilde{T}_j] \quad \text{against} \quad & |\lg\tb|, \\
    \Exp_\cA[\Tc] \quad \text{against} \quad & |\lg\cb|.
\end{aligned}
\end{equation}

Moreover, we  plot the  ratios of the expected decision times  of the parallel SPRT and of the synchronous test over  those of   the proposed test in every signal (resp. noise) stream, also against the common absolute value of the base-10 logarithm of the maximum type-I (resp. -II) familywise error rate. That is,  for each $i\in\cA$ we plot
\begin{align} \label{ARE against |lg alpha|} 
\begin{split}
    \frac{\Exp_\cA[\Ts_i]}{\Exp_\cA[\Tilde{T}_i]} \;\quad \text{against} \quad & |\lg\ha|=|\lg\ta|, \\
  \frac{\Exp_\cA[\Ts_i]}{\Exp_\cA[\Tc]} \;\quad \text{against} \quad & |\lg\ha|=|\lg\ca|,
  \end{split}
\end{align}
and for each $j\notin\cA$ we plot
\begin{equation} \label{ARE against |lg beta|}
\begin{aligned}
    \frac{\Exp_\cA[\Ts_j]}{\Exp_\cA[\Tilde{T}_j]}  \quad \text{against} \quad & |\lg\hb|=|\lg\tb|, \\
    \frac{\Exp_\cA[\Ts_j]}{\Exp_\cA[\Tc]}  \quad \text{against} \quad & |\lg\hb|=|\lg\cb|.
\end{aligned}
\end{equation}
Meanwhile, we compare the above  relative efficiencies with their limiting values. For the parallel SPRT, they are given by  \eqref{AE' relative to decentralized procedures}.  For the synchronous test,  they are given by   \eqref{AE''i or j, cA, with r} with $r=1$, since in all cases we consider we have 
 $$ \ca=\cb.$$ 
Indeed, when $l=u=m$, the decision of the synchronous test,  $\check\bfD$,  is always of size $m$ and it commits a type-I error  if and only if it commits a type-II error. 
When $l<u$, $ \ca=\cb$ holds  because the thresholds are selected according to   \eqref{thresholds when l<u},
$l+u=K$ and 
\begin{equation} \label{same dist}
\begin{aligned}
    \text{the distribution of $\lambda_k$ under} &\text{ $\Pro_k^1$ is the same as} \\
    \text{that of $-\lambda_k$ under} & \text{ $\Pro_k^0$, \; for every }   k\in[K]. 
\end{aligned}
\end{equation}

\subsection{Results in the homogeneous setup}
We  first consider  the  homogeneous setup, \eqref{mu_homogeneous}, with $K=10$, $\mu=0.5$.
In this case,  the expected decision time is the same in every signal (resp. noise) stream, i.e., \eqref{E[T] against |lg alpha|}  and \eqref{ARE against |lg alpha|} (resp. \eqref{E[T] against |lg beta|} and  \eqref{ARE against |lg beta|}) do not depend on  $i\in\cA$ (resp. $j\notin\cA$).   Moreover, \eqref{E[T] against |lg alpha|} and  \eqref{ARE against |lg alpha|}   coincide with  \eqref{E[T] against |lg beta|} and  \eqref{ARE against |lg beta|} when $l,u,A$ are replaced by  $K-u,K-l,A^c$. When $l=u$,  this is the case because    $c=d$ and  \eqref{same dist} holds.  When $l<u$,  this is the case because the thresholds are selected according to \eqref{thresholds when l<u}, $l+u=K$, and  \eqref{same dist} holds. 

In view of these observations, we  only plot   \eqref{E[T] against |lg alpha|} and \eqref{ARE against |lg alpha|} for an  arbitrary  $A$  and an arbitrary $i\in\cA$ when $l=u$
in Figure \ref{Figure, gap, ESS against actual error rates},  
and for  every  $|\cA|\in\{l,\ldots,u\}$  and an arbitrary $i\in\cA$ when $l<u$ 
in Figure \ref{Figure, gap-inter, ESS against actual error rates}.
In each figure, the first row depicts 
the expected decision time in an arbitrary signal stream,  \eqref{E[T] against |lg alpha|}, and the second row
the relative efficiency,   \eqref{ARE against |lg alpha|},
both  against  the absolute value of the base-10 logarithm of the  maximum type-I familywise error rate. The first column refers to the proposed test, the second to the parallel SPRT, and the third to the  synchronous test.

\subsubsection{Known number of signals}
\begin{figure} [ht]
    \centering
    \hfill
    \begin{subfigure}[b]{0.325\textwidth} \label{homo, gap, 1}
    \centering
    \includegraphics[width=\textwidth]{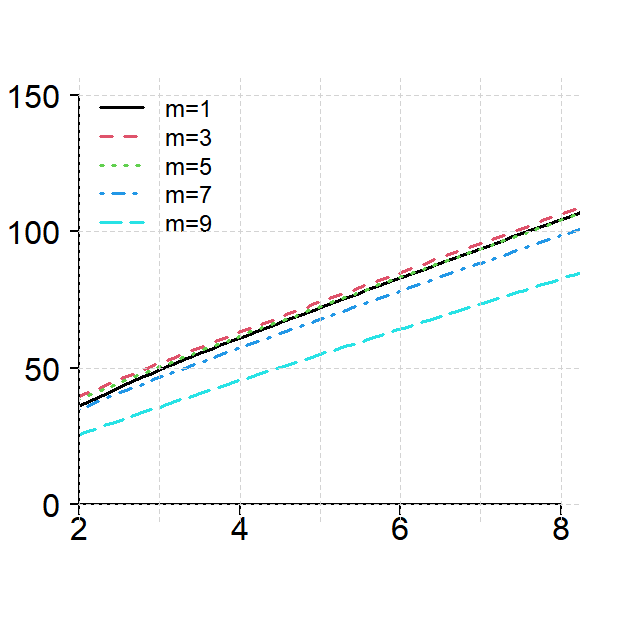}
    \caption{$\hat\chi$}
    \end{subfigure}
    \hfill
    \begin{subfigure}[b]{0.325\textwidth} \label{homo, gap, 2}
    \centering
    \includegraphics[width=\textwidth]{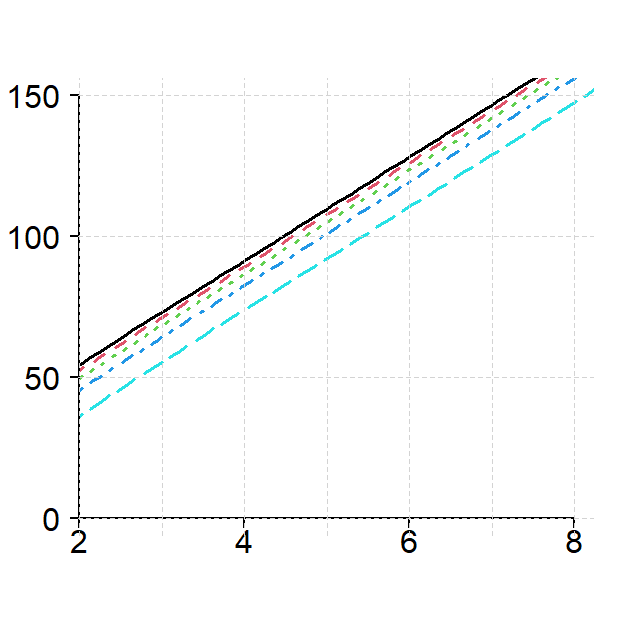}
    \caption{$\tilde\chi$}
    \end{subfigure}
    \hfill
    \begin{subfigure}[b]{0.325\textwidth} \label{homo, gap, 3}
    \centering
    \includegraphics[width=\textwidth]{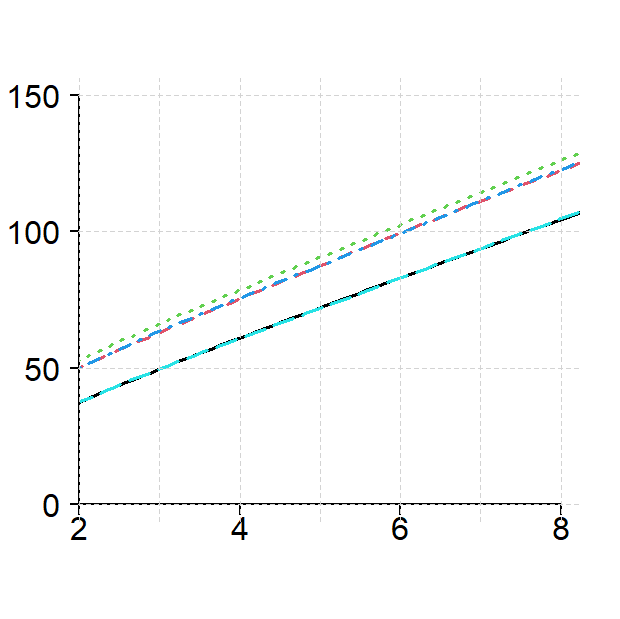}
    \caption{$\check\chi$}
    \end{subfigure} \\
    \hfill
    \begin{subfigure}[b]{0.325\textwidth} \label{homo, gap, 4}
    \centering
    \includegraphics[width=\textwidth]{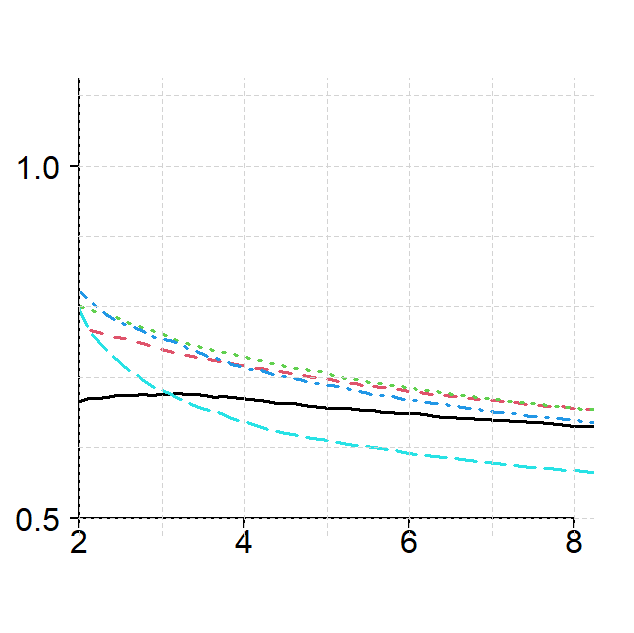}
    \caption{$\hat\chi/\tilde\chi$}
    \end{subfigure}
    \begin{subfigure}[b]{0.325\textwidth} \label{homo, gap, 5}
    \centering
    \includegraphics[width=\textwidth]{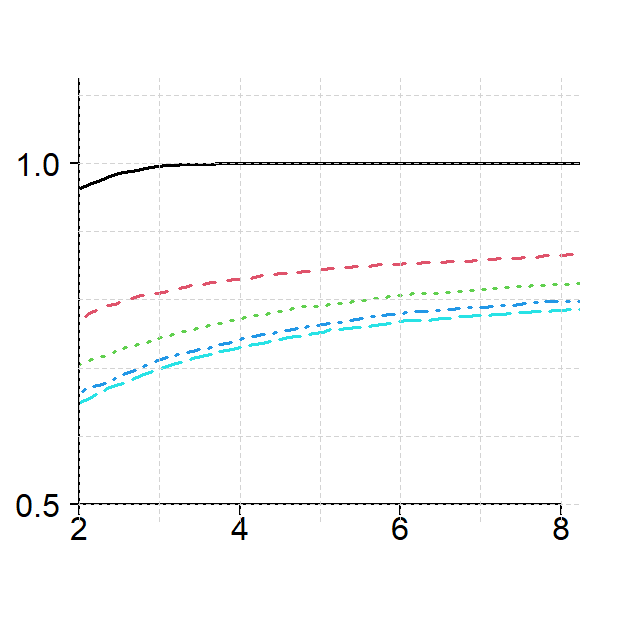}
    \caption{$\hat\chi/\check{\chi}$}
    \end{subfigure}
    \caption{Homogeneous setup with known number of signals, $l=u\equiv m\in\{1,3,5,7,9\}$.     The first row depicts the expected decision times in \eqref{E[T] against |lg alpha|} and the second   the relative efficiencies in \eqref{ARE against |lg alpha|}. The first column refers to the proposed test, the second to the parallel SPRT, and the third to the  synchronous test.}
    \label{Figure, gap, ESS against actual error rates}
\end{figure}

We consider  first  the homogeneous setup when the number of signals is a priori known to be equal to some  $m \in \{1, \ldots, K-1\}$. We start with the first row in Figure \ref{Figure, gap, ESS against actual error rates},  where the expected decision times  of the three tests are presented.

\begin{itemize}
\item From Figure \ref{Figure, gap, ESS against actual error rates}.(a) we can see that, for the proposed test, the expected decision time in a signal  stream  is the smallest  when $m=9$, and it is overall decreasing in $m$, as  expected in view of Remark \ref{remark: intuition from higher-order terms}.

\item From Figure \ref{Figure, gap, ESS against actual error rates}.(b) we can see that, for  the parallel SPRT, the curves are horizontal translations of each other, moving to the right as $m$ increases. This is  because  we set $a=b$, in which case the expected decision time is the same in every stream and for every $m$, whereas the type-I familywise error rate decreases as $m$ increases. 

\item From Figure \ref{Figure, gap, ESS against actual error rates}.(c) we can see that, for  the synchronous test, the  expected decision time 
is the largest when $m=K/2$, as expected in view of Remark \ref{remark: intuition from higher-order terms}. We also note that  the expected decision time and the maximum familywise error rate coincide when $l=u$ is equal to $m$ and when it is equal to $K-m$, as expected by the symmetry of the setup.  
\end{itemize}

We continue with the second row in Figure \ref{Figure, gap, ESS against actual error rates},  where the relative efficiencies of the  parallel SPRT and of the synchronous test  against the proposed test are compared with their limits, which are  given by 
\eqref{AE' in homo and symm setup}  and \eqref{AE'' in homo and symm setup}, respectively. 
Specifically, in a homogeneous and symmetric setup with known number of signals, as the current one, the latter  are  in all streams equal to $1/2$ for the parallel SPRT and to  $1$ for the synchronous test.

\begin{itemize}
    \item From Figure \ref{Figure, gap, ESS against actual error rates}.(d) we can see that, for the parallel SPRT, the convergence rate to $1/2$ is about the same for all values of $m$.
    Specifically, the relative efficiencies of the parallel SPRT compared with the proposed test are at most $82\%$ when $\alpha=0.01$ and at most $65\%$ when $\alpha=10^{-8}$.
    
    \item  From Figure \ref{Figure, gap, ESS against actual error rates}.(e) we can see that, for the synchronous test,  the convergence to $1$ is  very fast when $m=1$, but much slower when $m>1$.
    Indeed, when $m=1$, the synchronous test performs similarly to the proposed one,  as expected  by the second-order term in \eqref{AUB with higher-order term, proposed, gap, i} and \eqref{AUB with higher-order term, synchronous, gap}. 
    On the other hand, when $m>1$, the relative efficiencies of the synchronous test compared with the proposed test are at most $78\%$ when $\alpha=0.01$ and at most $88\%$ when $\alpha=10^{-8}$. Therefore, when $m>1$, the proposed test performs significantly better than the synchronous test  in this setup, even though the  \textit{asymptotic} relative efficiency is equal to 1 in every stream. 
\end{itemize}

\subsubsection{Lower and upper bounds on the number of signals}
\begin{figure} [ht]
    \centering
    \hfill
    \begin{subfigure}[b]{0.325\textwidth} \label{homo, gap-inter, 1}
    \centering
    \includegraphics[width=\textwidth]{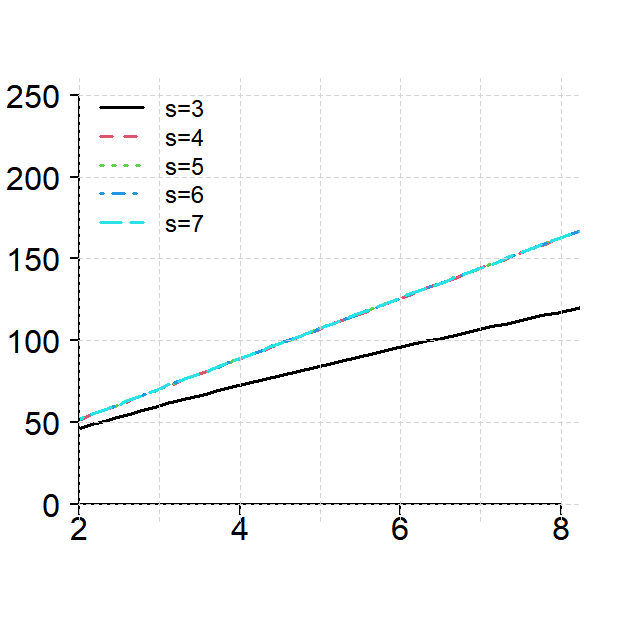}
    \caption{$\hat\chi$}
    \end{subfigure}
    \hfill
    \begin{subfigure}[b]{0.325\textwidth} \label{homo, gap-inter, 2}
    \centering
    \includegraphics[width=\textwidth]{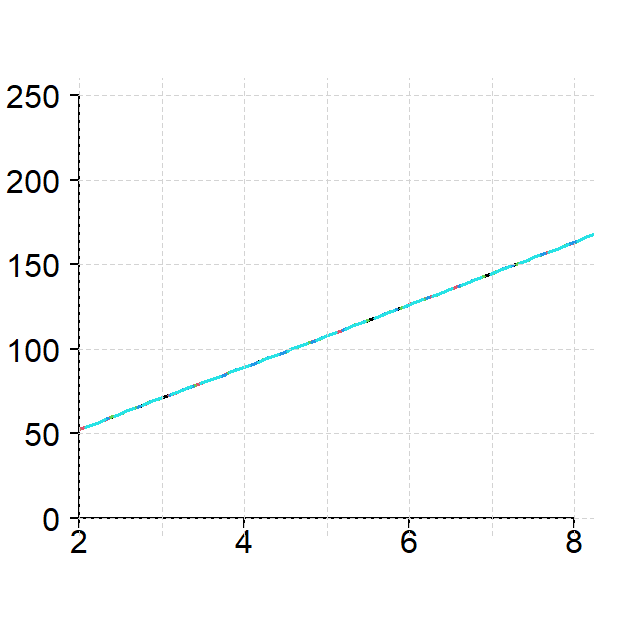}
    \caption{$\tilde\chi$}
    \end{subfigure}
    \hfill
    \begin{subfigure}[b]{0.325\textwidth} \label{homo, gap-inter, 3}
    \centering
    \includegraphics[width=\textwidth]{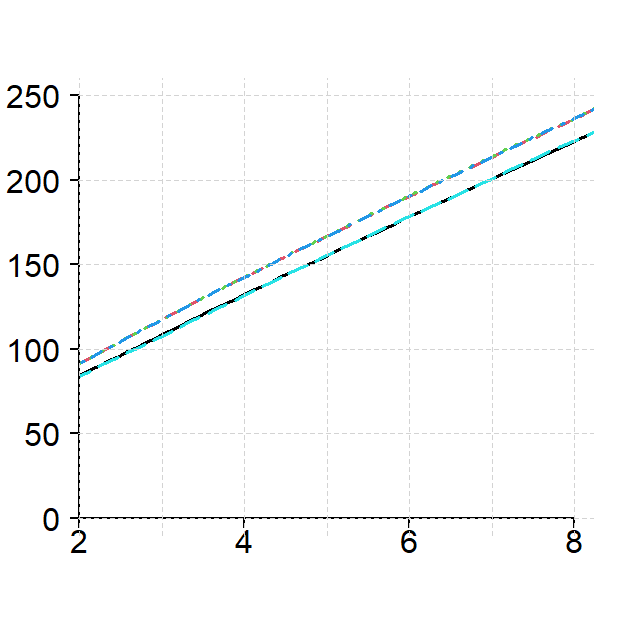}
    \caption{$\check\chi$}
    \end{subfigure} \\
    \hfill
    \begin{subfigure}[b]{0.325\textwidth} \label{homo, gap-inter, 4}
    \centering
    \includegraphics[width=\textwidth]{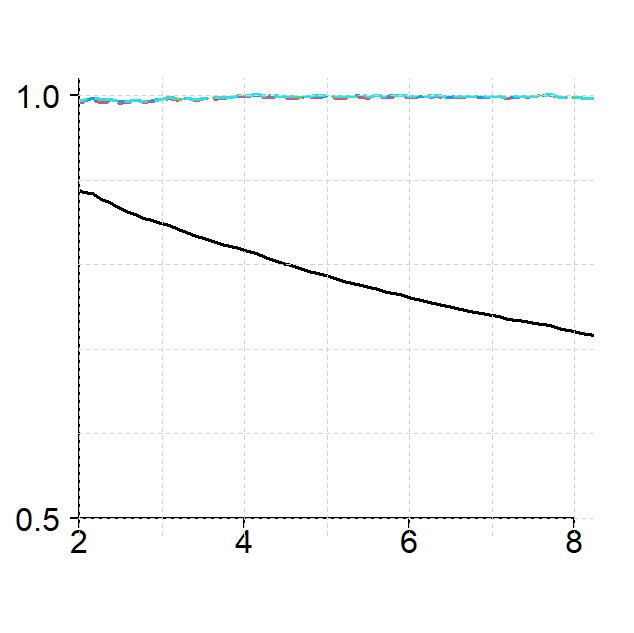}
    \caption{$\hat\chi/\tilde\chi$}
    \end{subfigure}
    \begin{subfigure}[b]{0.325\textwidth} \label{homo, gap-inter, 5}
    \centering
    \includegraphics[width=\textwidth]{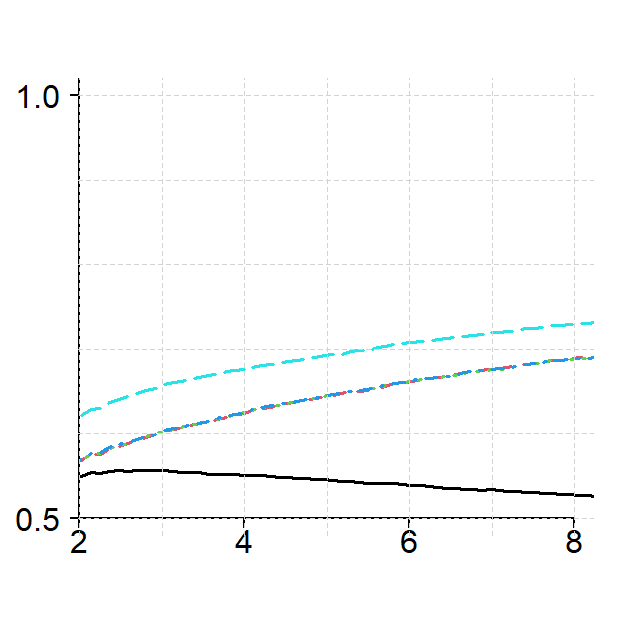}
    \caption{$\hat\chi/\check\chi$}
    \end{subfigure}
    \caption{Homogeneous setup with 
    $l=3$, $u=7$ and $|\cA|\equiv s\in\{3,4,5,6,7\}$.
    The first row depicts the expected decision times in \eqref{E[T] against |lg alpha|} and the second  the relative efficiencies in  \eqref{ARE against |lg alpha|}. The first column refers to the proposed test, the second to the parallel SPRT, and the third to the  synchronous test.} 
    \label{Figure, gap-inter, ESS against actual error rates}
\end{figure}

We next consider the homogeneous setup when the number of signals is a priori known to be at least $l=3$ and at most $u=7$ and the true number of signals, $|A|$, is equal to some $s \in \{l, \ldots, u\}$.  We start with the  first row of Figure \ref{Figure, gap-inter, ESS against actual error rates}, where the expected decision times of the three tests are presented. 

\begin{itemize}
\item   From Figure \ref{Figure, gap-inter, ESS against actual error rates}.(a) we can see that the  expected decision time  of the proposed test  in a signal stream is essentially  the same for all values of $s$ with $s>l$, but  much smaller when $s=l$, as expected  by  \eqref{AO, i}.

\item From Figure \ref{Figure, gap-inter, ESS against actual error rates}.(b) we can see that the curves for the parallel SPRT  coincide for all values of $s$. This is expected since its expected  decision time  in each stream does not depend on the true subset of signals when $a=b$. 

\item From Figure \ref{Figure, gap-inter, ESS against actual error rates}.(c) we can see that the curves for the synchronous test coincide when $s=l$ and  $s=u=K-l$, and that they  also coincide for all values of $s$ between $l$ and $u$. 
\end{itemize}

We continue with the second row in Figure  \ref{Figure, gap-inter, ESS against actual error rates},  where the relative efficiencies of the  parallel SPRT and of the synchronous test  against the proposed test are compared with their limits. From  \eqref{AE' in homo and symm setup} and \eqref{AE'' in homo and symm setup} it follows that the latter are equal,  for both tests,  to $1$ when $s>l$ and to $1/2$ when $s=l$.

From Figure \ref{Figure, gap-inter, ESS against actual error rates}.(d) and  Figure \ref{Figure, gap-inter, ESS against actual error rates}.(e)  we can see that the parallel SPRT performs substantially better than the synchronous test in all cases. Specifically, when $s>l$, its relative efficiencies are essentially equal to 1 even for large error probabilities, whereas those of the synchronous test never surpass $75\%$.  When 
$s=l$, the relative efficiency of the parallel SPRT decreases to $1/2$ in a relatively slow rate, whereas that of the synchronous test is always below $60\%$

\subsection{Results in the  non-homogeneous setup}
We next consider the non-homogeneous setup, \eqref{mu_nonhomogeneous}, with  $K=4$, $\mu=0.5$, $\phi=0.5$. In the case where the number of signals is not a priori known, we set $l=1$ and $u=3$.   Unlike the homogeneous setup of the previous subsection,  the expected decision times in  \eqref{quantities of E[T]} now do not depend only on  the true number of signals, but also on the true subset of signals itself. We  next focus on the following cases for the true subsets of signals:
\begin{itemize}
    \item [(i)] $\cA=\{1\}$, where
    $\cI_\cA=\cJ_\cA=1/32$,
    \item [(ii)] $\cA=\{3\}$, where
    $\cI_\cA=1/8$ and $\cJ_\cA=1/32$,
    \item [(iii)] $\cA=\{1,2\}$, where
     $\cI_\cA=1/32$ and $\cJ_\cA=1/8$,
    \item [(iv)] $\cA=\{1,3\}$, where
    $\cI_\cA=\cJ_\cA=1/32$.
\end{itemize}
As before, we consider two cases regarding the prior information. One where the number of signals is a priori known, i.e., $|A|=l=u$, and one where it is a priori known to  be at least $l=1$ and at most  $u=3$.  The asymptotic relative efficiencies, given by   \eqref{AE' relative to decentralized procedures} and \eqref{AE''i or j, cA, with r} with $r=1$,  are presented in Table \ref{Table, AE, ell=u} and \ref{Table, AE, ell<u}.

\begin{table}[ht]
    \centering
    \begin{tabular}{|c|c|c|c|c|}
        \hline
        \backslashbox{$\cA$}{k} & 1 & 2 & 3 & 4 \\
        \hline
         $\{1\}$ &  1/2 & 1/2 & 4/5 & 4/5 \\
        \hline
          $\{3\}$ & 1/5 & 1/5 & 4/5 & 1/2 \\
        \hline
         $\{1,2\}$ & 1/5 & 1/5 & 4/5 & 4/5 \\
        \hline
         $\{1,3\}$ & 1/2 & 1/2 & 4/5 & 4/5 \\
        \hline
    \end{tabular}
\qquad \qquad 
    \begin{tabular}{|c|c|c|c|c|}
        \hline
        \backslashbox{$\cA$}{k} & 1 & 2 & 3 & 4 \\
        \hline
         $\{1\}$ &  1 & 1 & 2/5 & 2/5 \\
        \hline
          $\{3\}$ & 1 & 1 & 1 & 5/8 \\
        \hline
         $\{1,2\}$ & 1 & 1 & 1 & 1 \\
        \hline
         $\{1,3\}$ & 1 & 1 & 2/5 & 2/5 \\
        \hline
    \end{tabular}
    \caption{$\Ae'_{k,\cA}(l,u)$ and $\Ae''_{k,\cA}(l,u)$ when $l=u=|\cA|$.} 
    \label{Table, AE, ell=u}
\end{table}

\begin{table}[ht]
    \centering
    \begin{tabular}{|c|c|c|c|c|}
        \hline
        \backslashbox{$\cA$}{k} & 1 & 2 & 3 & 4 \\
        \hline
         $\{1\}$  & 1/2 & 1 & 1 & 1 \\
        \hline
         $\{3\}$  & 1 & 1 & 4/5 & 1 \\
        \hline
         $\{1,2\}$  & 1 & 1 & 1 & 1 \\
        \hline
      $\{1,3\}$  & 1 & 1 & 1 & 1 \\
        \hline
    \end{tabular}
\qquad \qquad 
    \begin{tabular}{|c|c|c|c|c|}
        \hline
        \backslashbox{$\cA$}{k} & 1 & 2 & 3 & 4 \\
        \hline
         $\{1\}$  & 1/2 & 1 & 1/4 & 1/4 \\
        \hline
         $\{3\}$  & 1 & 1 & 1/5 & 1/4 \\
        \hline
         $\{1,2\}$  & 1 & 1 & 1/4 & 1/4 \\
        \hline
      $\{1,3\}$  & 1 & 1 & 1/4 & 1/4 \\
        \hline
    \end{tabular}
    \caption{$\Ae'_{k,\cA}(l,u)$ and $\Ae''_{k,\cA}(l,u)$ when $l=1, u=3$.}
    \label{Table, AE, ell<u}
\end{table}

We present the plots for the case of $l=u$ in Table \ref{Figure, non-homo, gap} and for the case of $l=1$, $u=3$ in Table \ref{Figure, non-homo, gap-inter}.
In each of these two tables of plots, each column corresponds to one of the four cases for $\cA$ in (i)-(iv). The  first row depicts the expected decision times  of the parallel SPRT in each of the four streams.  The second row depicts the corresponding plots  for the proposed test and the synchronous test. The third (resp. fourth) row depicts the relative efficiencies of the parallel SPRT (resp. synchronous test) against the proposed test in each of the four streams.  The corresponding limiting values from Tables \ref{Table, AE, ell=u} and \ref{Table, AE, ell<u} are marked on the vertical axis. In all plots, the horizontal axis represents the absolute value of the base-10 logarithm of the  maximum type-I (resp. -II) familywise error rate if that stream is a signal (resp. noise).

From the third and the fourth row in each table of plots we can see that    in all cases
the relative efficiencies converge to their  limiting values very quickly. Thus, in this setup, the latter provide  a very accurate approximation of the  actual relative efficiencies.

\begin{table}[]
    \begin{tabular}{ m{1.745em} m{9em} m{9em} m{9em} m{9em} } 
         & \centering (i)  & \centering (ii) & \centering (iii) & \begin{center} (iv) \end{center} \\
    $\tilde\chi$
    & \includegraphics[width=0.25\textwidth]{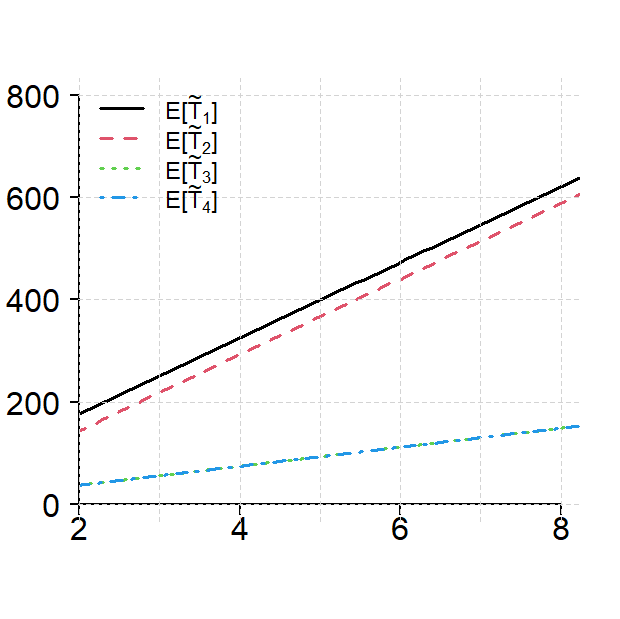} & \includegraphics[width=0.25\textwidth]{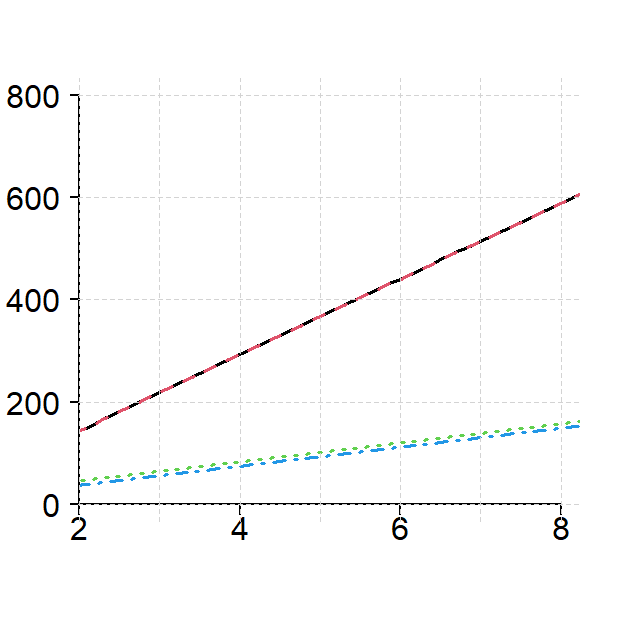} & \includegraphics[width=0.25\textwidth]{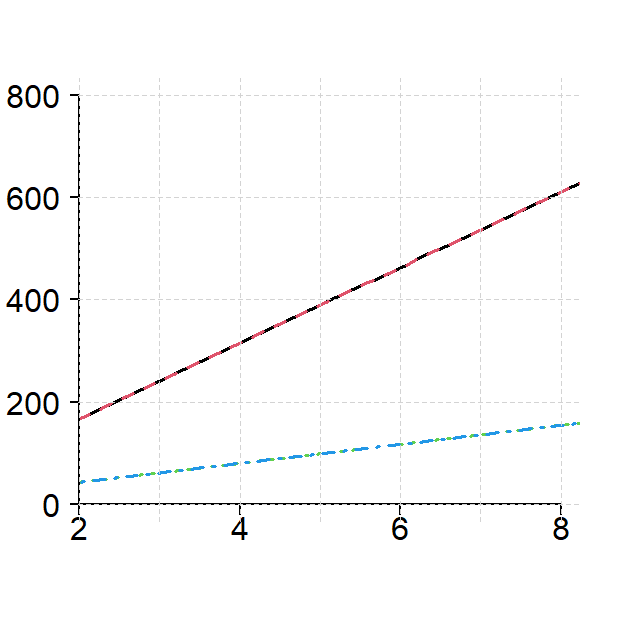} & \includegraphics[width=0.25\textwidth]{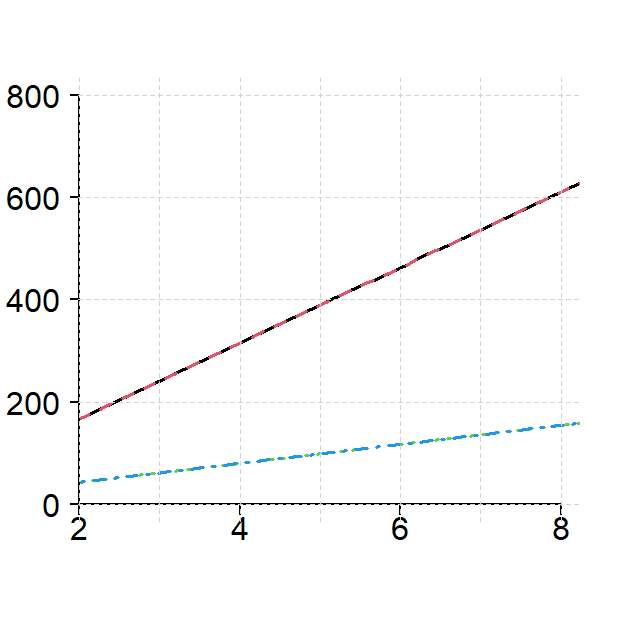} \\
    $\hat\chi$, $\check\chi$
    & \includegraphics[width=0.25\textwidth]{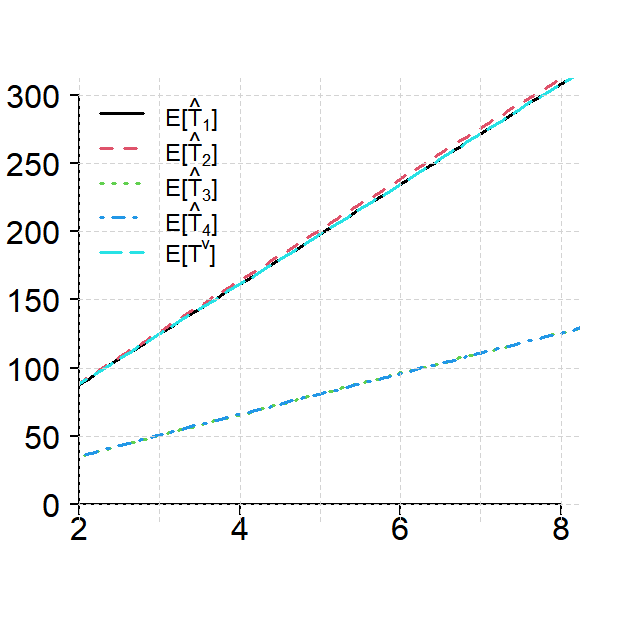} & \includegraphics[width=0.25\textwidth]{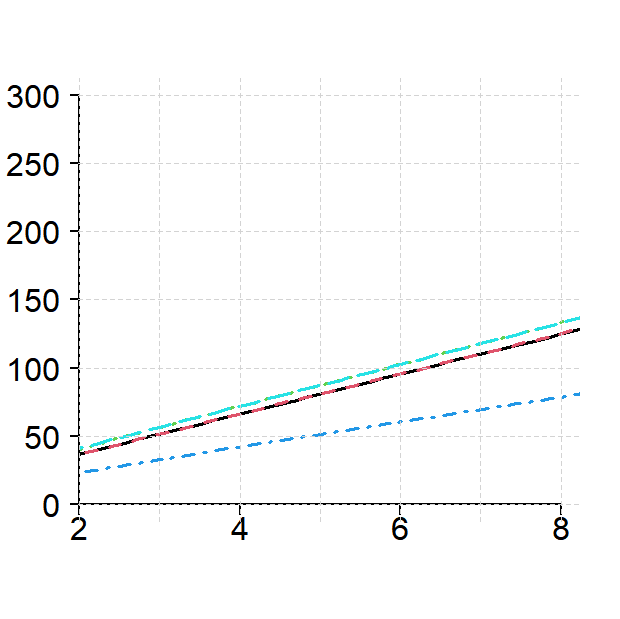} & \includegraphics[width=0.25\textwidth]{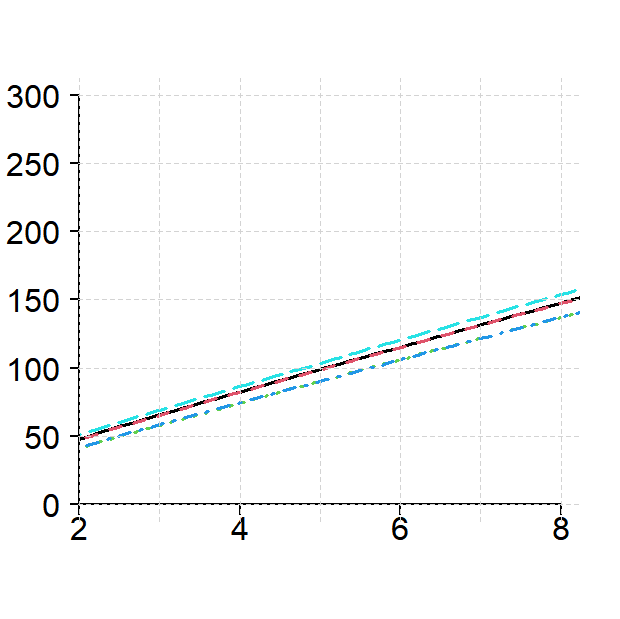} & \includegraphics[width=0.25\textwidth]{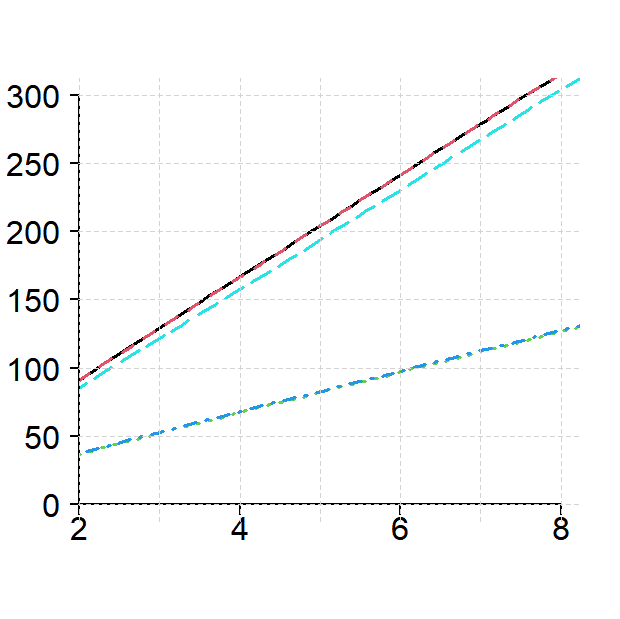}  \\
    $\hat\chi/\tilde\chi$
    & \includegraphics[width=0.25\textwidth]{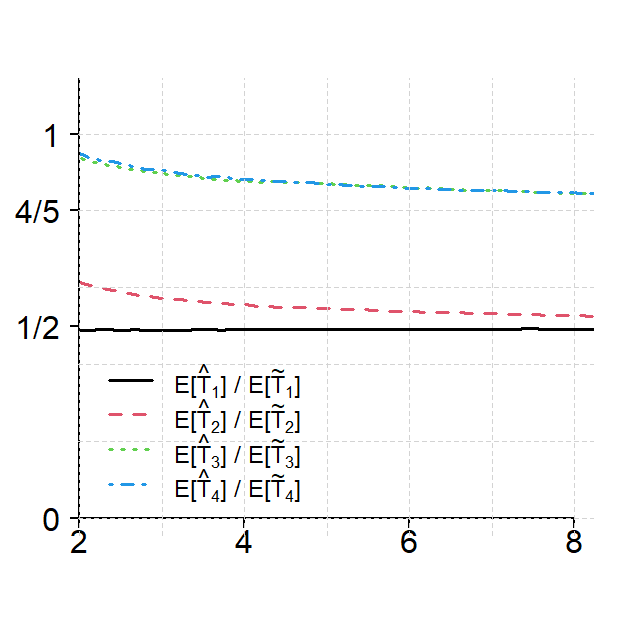} & \includegraphics[width=0.25\textwidth]{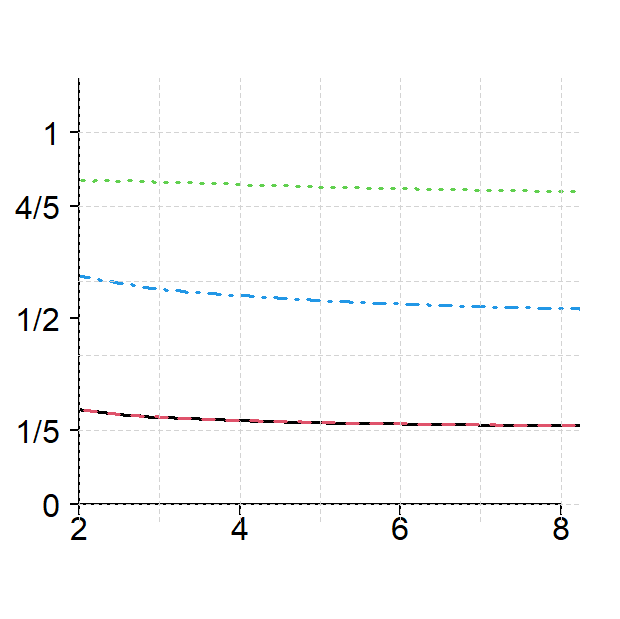} & \includegraphics[width=0.25\textwidth]{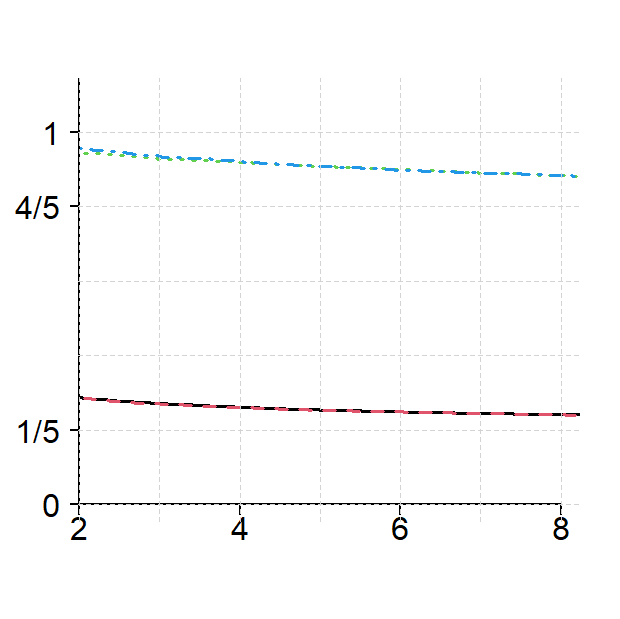} & \includegraphics[width=0.25\textwidth]{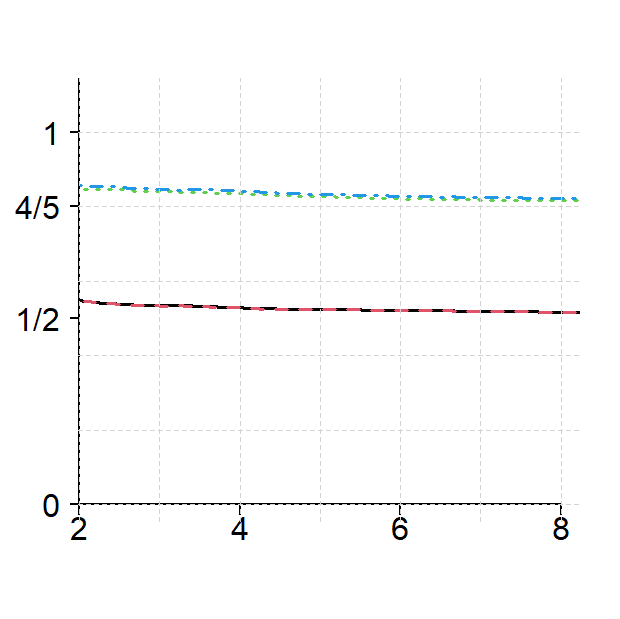} \\
    $\hat\chi/\check\chi$
    & \includegraphics[width=0.25\textwidth]{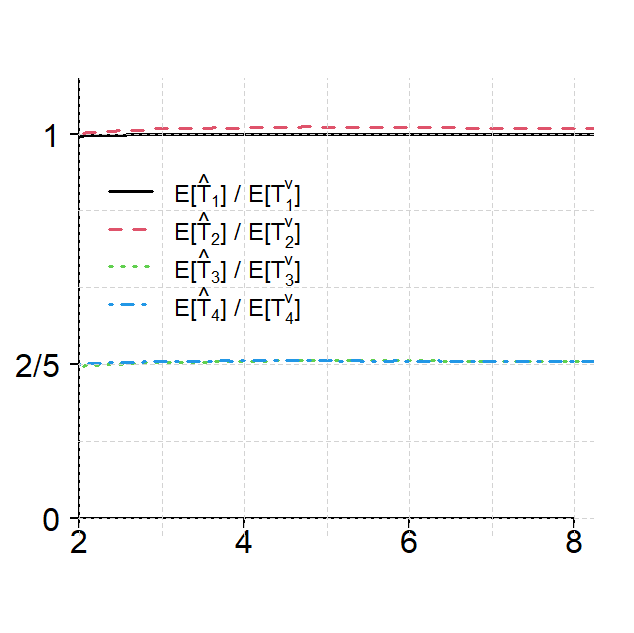} & \includegraphics[width=0.25\textwidth]{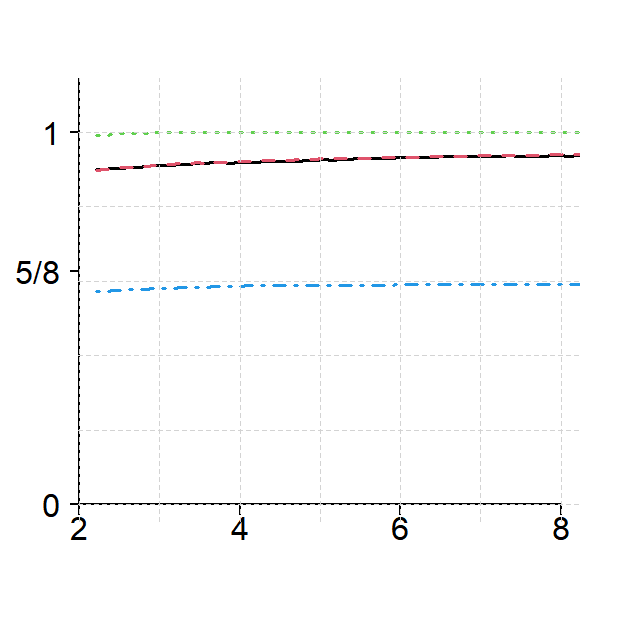} & \includegraphics[width=0.25\textwidth]{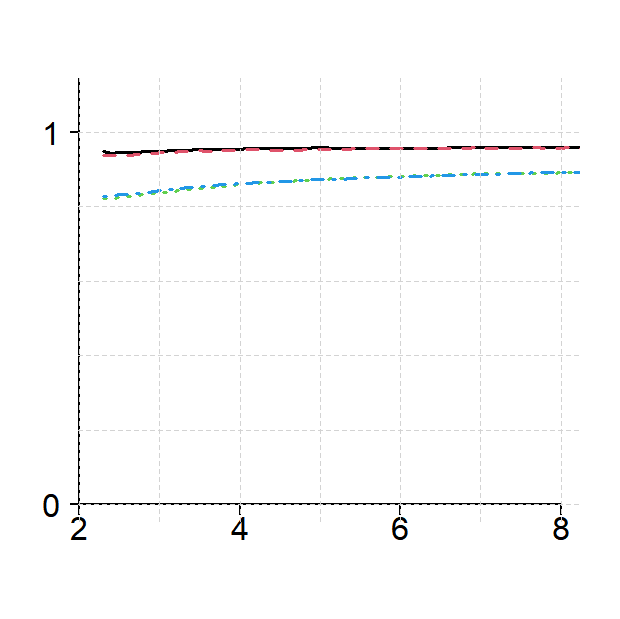} & \includegraphics[width=0.25\textwidth]{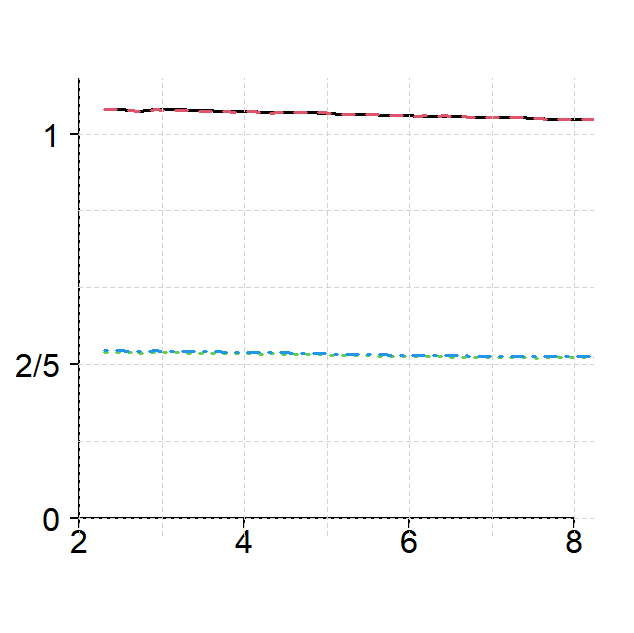}
    \end{tabular}
    \caption{Non-homogeneous setup with known number of signals.    Each column corresponds to a case for the true subset of signals.    The horizontal axis represents the absolute value of the base-10 logarithm  of the maximum familywise type-I or -II error rate, depending on whether the stream is a signal or a noise.    The vertical axis in the  first (resp. second) row represents the expected decision times of the parallel SPRT, $\tilde\chi$ (resp. the proposed test, $\hat\chi$, and the synchronous test, $\check\chi$).  The vertical axis in the third (resp. fourth) row represents the relative efficiencies of $\tilde\chi$ (resp. $\check\chi$) against $\hat\chi$.   The asymptotic relative efficiencies in Table \ref{Table, AE, ell=u} are marked on the vertical axis in the third and fourth rows.}
    \label{Figure, non-homo, gap}
\end{table}

\begin{table}[]
    \begin{tabular}{ m{1.745em} m{9em} m{9em} m{9em} m{9em} } 
         & \centering (i)  & \centering (ii) & \centering (iii) & \begin{center} (iv) \end{center} \\
    $\tilde\chi$
    & \includegraphics[width=0.25\textwidth]{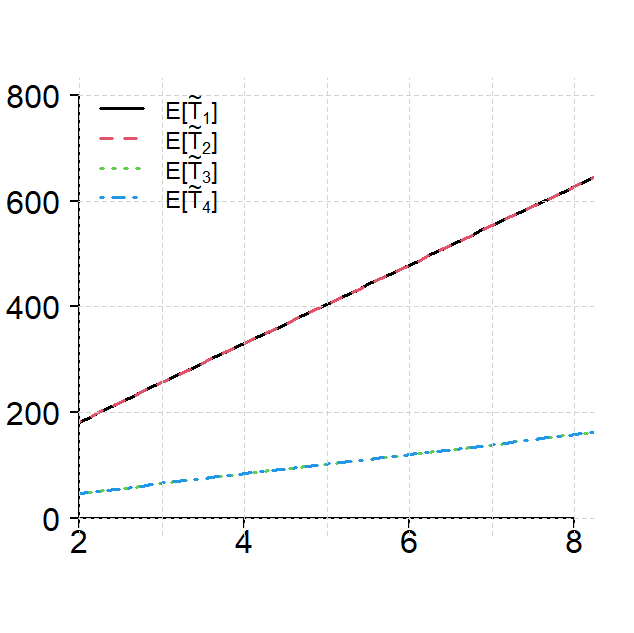} & & & \\
    $\hat\chi$, $\check\chi$
    & \includegraphics[width=0.25\textwidth]{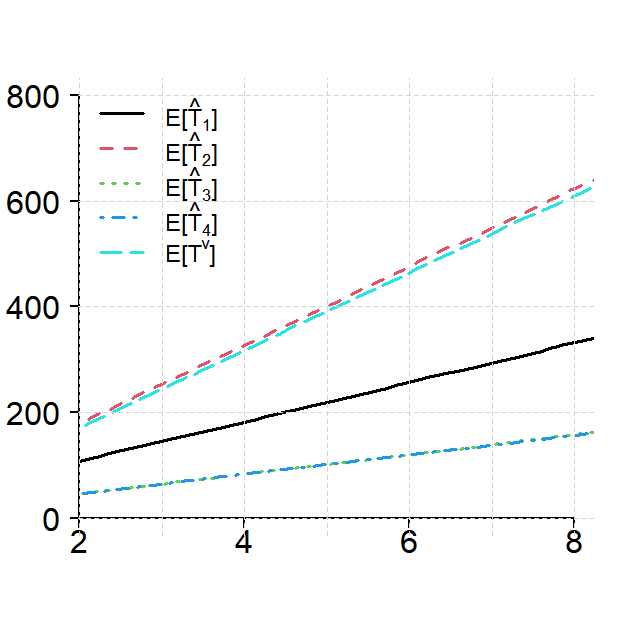} & \includegraphics[width=0.25\textwidth]{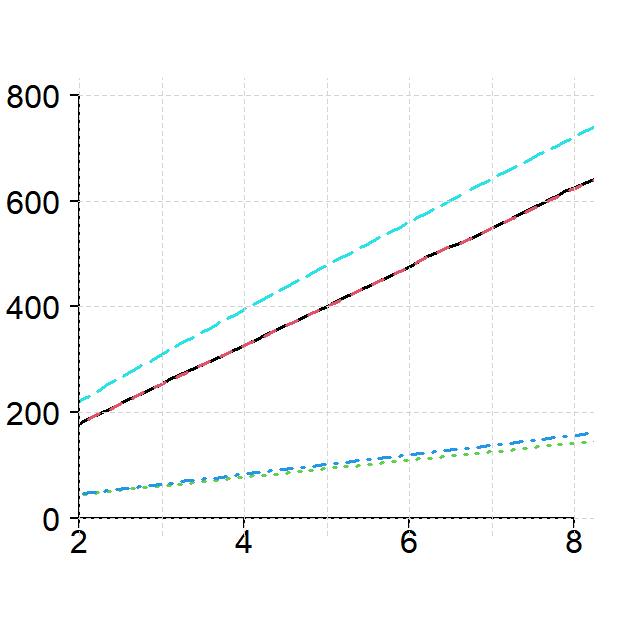} & \includegraphics[width=0.25\textwidth]{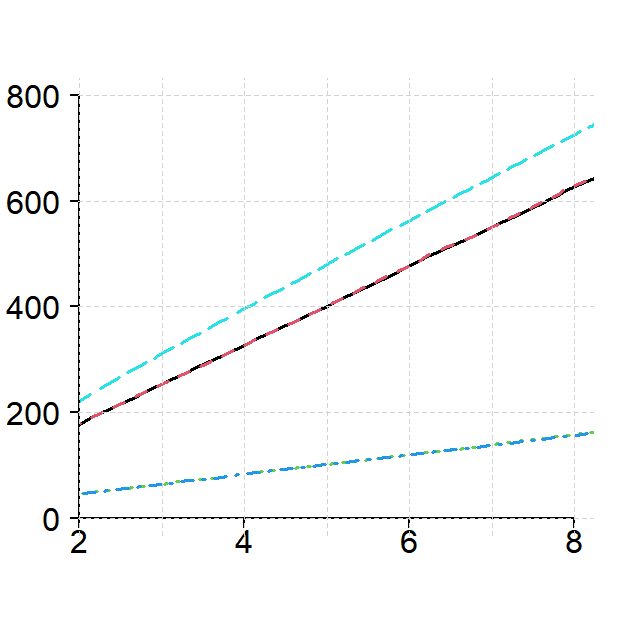} & \includegraphics[width=0.25\textwidth]{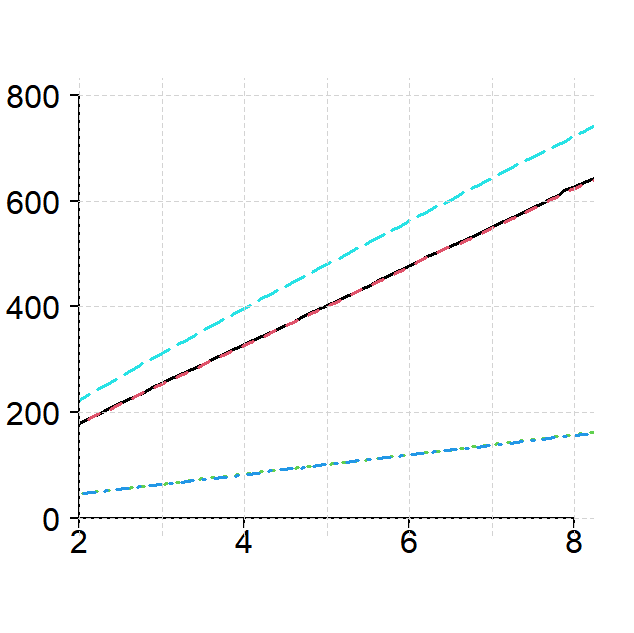}  \\
    $\hat\chi/\tilde\chi$
    & \includegraphics[width=0.25\textwidth]{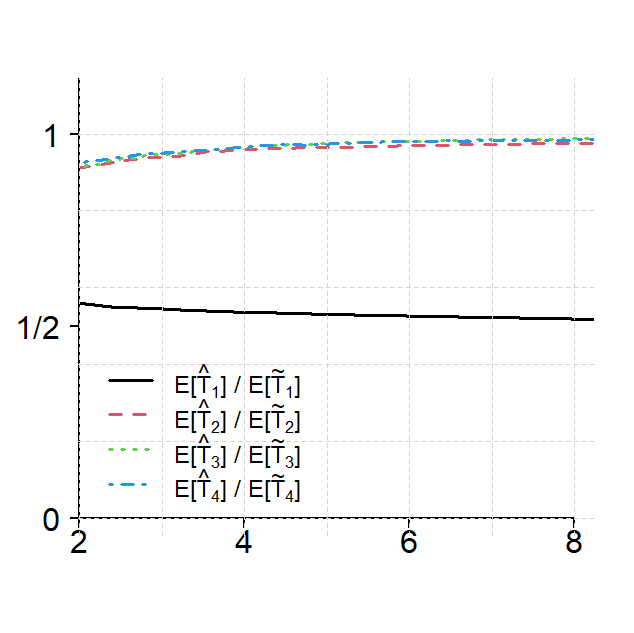} & \includegraphics[width=0.25\textwidth]{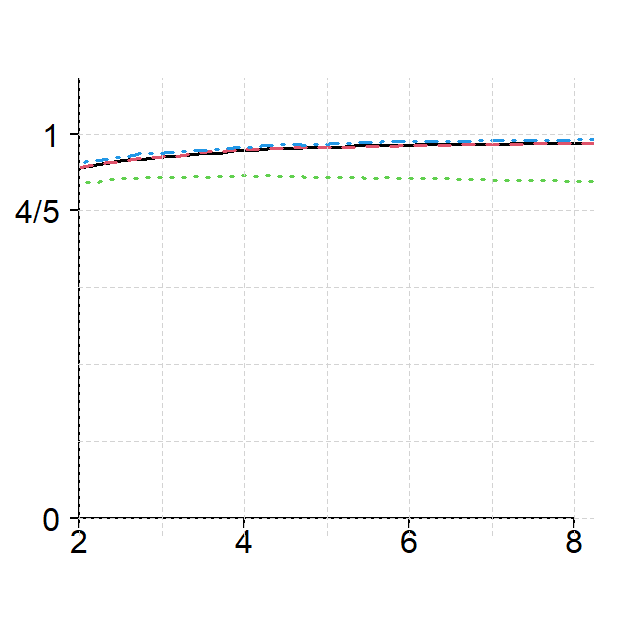} & \includegraphics[width=0.25\textwidth]{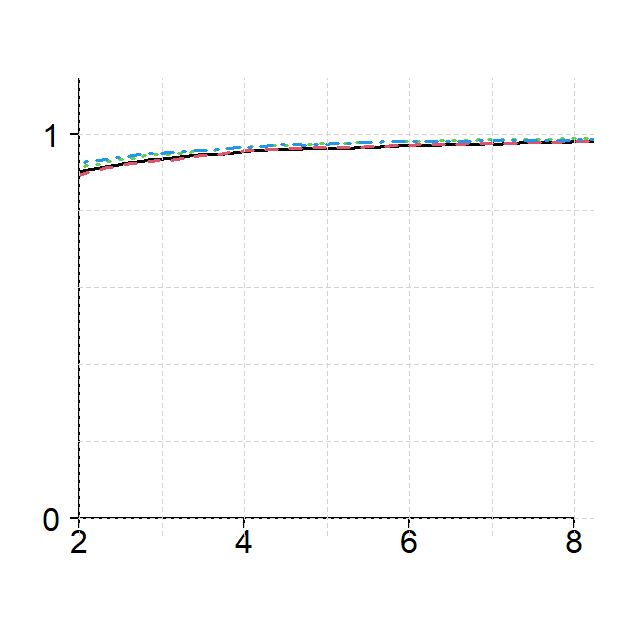} & \includegraphics[width=0.25\textwidth]{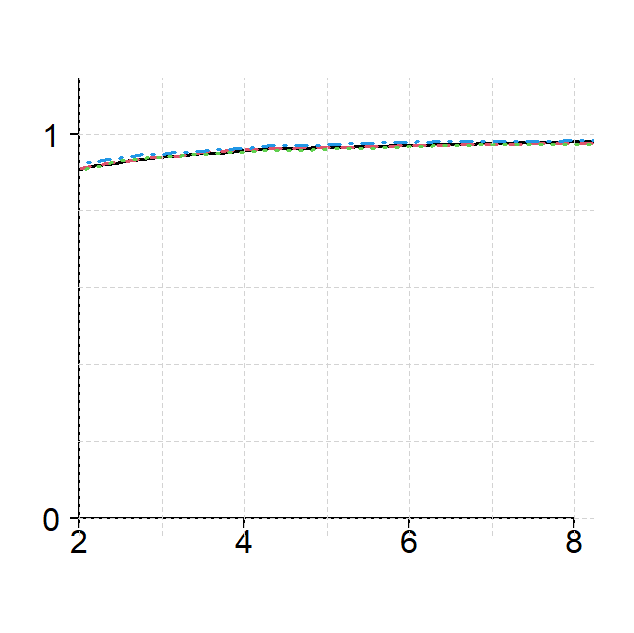} \\
    $\hat\chi/\check\chi$
    & \includegraphics[width=0.25\textwidth]{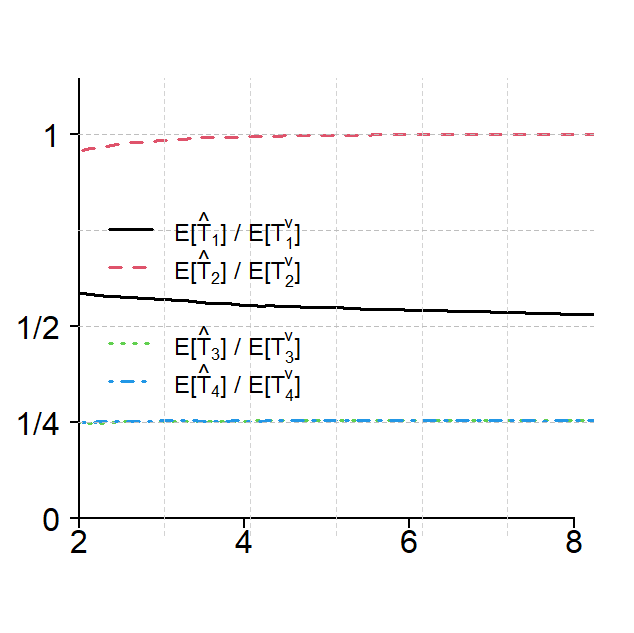} & \includegraphics[width=0.25\textwidth]{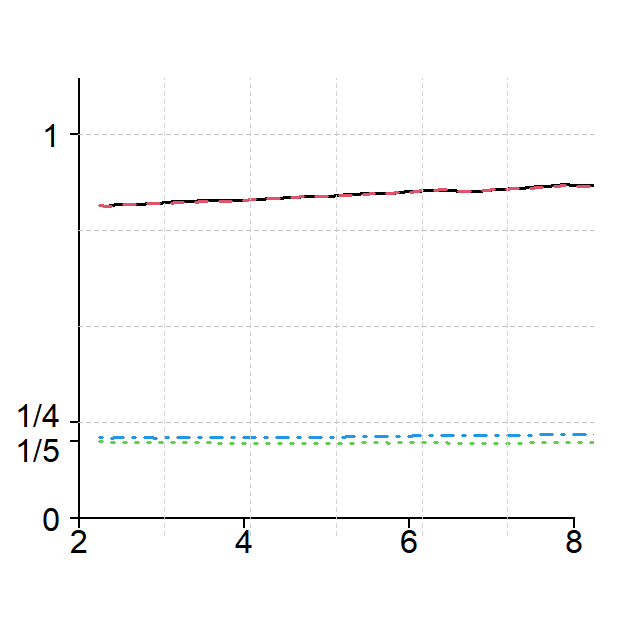} & \includegraphics[width=0.25\textwidth]{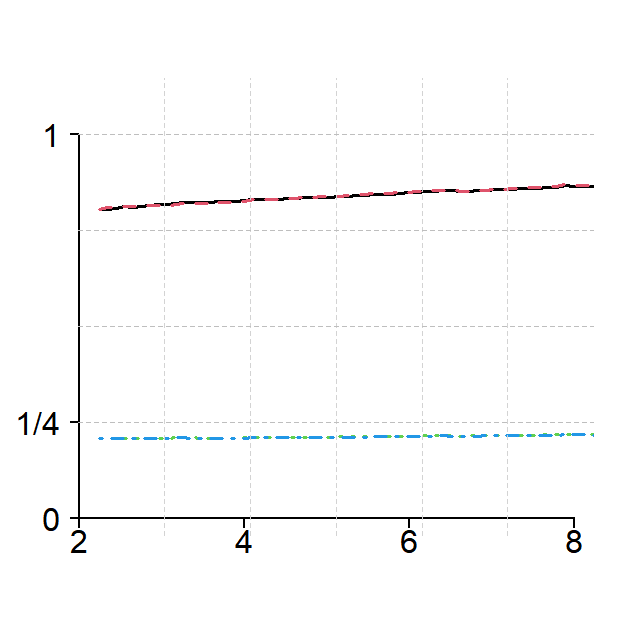} & \includegraphics[width=0.25\textwidth]{vii,ratios.png}
    \end{tabular}
    \caption{Non-homogeneous setup with lower bound $l=1$ and upper bound $u=3$ on the number signals. The meaning of each column, each row and each axis is the same as Table \ref{Figure, non-homo, gap}. There is only one plot in the first row because the corresponding plots  for the parallel SPRT are exactly the same in all cases regarding the true subset of signals.      The limits  for the relative efficiencies, which are marked on the vertical axis in the third and fourth rows, are give by Table \ref{Table, AE, ell<u}.}
    \label{Figure, non-homo, gap-inter}
\end{table}

\section{Generalization to parametric composite hypotheses} \label{section: generalization to composite hypotheses}
In this section we extend the methodology and results of the previous sections to  general, parametric composite hypotheses, following the adaptive likelihood ratio approach in 
\cite{ Robbins_Siegmund, pavlov1991sequential} and   \cite[Section 6]{Song_AoS} (see, also,  \cite[Chapter 5]{Tartakovsky_Book}). 

Thus,  for each $k \in [K]$, we assume that the distribution  of $X_k$ is specified up to an unknown parameter that belongs to a subset $\Theta_k$ of some Euclidean space, and denote it by $\Pro_{k,\theta}$ when the value of this parameter is $\theta$.
 In this context, 
 the  hypotheses in \eqref{Testing problem} are of the form 
  $$\cP_k^i= \{\Pro_{k,\theta}, \, \theta \in \Theta^i_k\}, \quad i\in\{0,1\},$$
where $\Theta_k^0$ and $\Theta_k^1$ are two  disjoint subsets of $\Theta_k$, and the goal is to  minimize the expected decision time simultaneously \emph{in every stream and for every possible value of the unknown parameter}, while controlling the \emph{worst-case} type-I and type-II  familywise error rates.  To be specific, we denote by $\bfTheta$ the global parameter space, i.e., 
$$
\bfTheta \equiv\Theta_1 \times \cdots \times \Theta_K,
$$
and for any $\cA\subseteq [K]$, denote by $\bfTheta_\cA$ its subset that is consistent with  the subset of signals being  $\cA$, i.e., 
$$ \bfTheta_\cA\equiv \left\{  (\theta_1,\ldots,\theta_K)\in\bfTheta: \quad \theta_i\in\Theta_i^1,  \quad
\theta_j\in\Theta_j^0, \quad  i\in\cA, \quad j\notin\cA  \right\}.$$
For any  $\bftheta=(\theta_1,\ldots,\theta_K) \in \Theta_\cA$ we denote by $\Pro_{\cA,\bftheta}$ the joint distribution of all streams when the subset of signals is $A$ and the parameter is $\bftheta$,  by $\Exp_{\cA,\bftheta}$ the corresponding expectation, and we note that by the assumption of  independence across streams we have 
$$ \Pro_{\cA,\bftheta}= \Pro_{1,\theta_1}\otimes \cdots \otimes  \Pro_{K,\theta_K}. $$

Then,  for any $\alpha,\beta\in (0,1)$ and $\Pi\subseteq 2^{[K]}$, the class of admissible tests takes the following form:
\begin{equation*}
    \begin{aligned}
        \Delta^{*}(\alpha,\beta,\Pi)\equiv \bigg\{ \chi\in\Delta:\, 
        \sup_{\bftheta\in\bfTheta_\cA} \Pro_{\cA,\bftheta}(\bfD\backslash\cA\neq\emptyset) & \leq\alpha \\
        \text{and} \quad
        \sup_{\bftheta\in\bfTheta_\cA} \Pro_{\cA,\bftheta}(\cA\backslash\bfD\neq\emptyset) & \leq\beta,  \quad \forall\, \cA
        \in \Pi \bigg\},
    \end{aligned}
\end{equation*}
and the goal is to achieve 
\begin{equation*}
    \inf\big\{ \Exp_{\cA,\bftheta}[T_k]: \chi\in\Delta^*(\alpha,\beta,\Pi) \big\}
\end{equation*}
simultaneously for every $k\in[K]$, 
$\cA\in\Pi$, $\bftheta\in\bfTheta_\cA$,  as $\alpha,\beta\to 0$.

\subsection{The sequential tests} \label{subsec: three tests in composite testing} 
We next generalize the testing  procedures in Section \ref{section: the proposed test} 
to the composite testing setup.  
To do this, for each $k \in [K]$, $\theta \in \Theta_k$, and $n\in\bN$, we assume that $\Pro_{k,\theta}$ is dominated by a $\sigma$-finite measure $\nu_k$ when both measures are restricted to $\cF_k(n)$, and we  denote by $\ell_k(n,\theta)$ the corresponding log-likelihood function, i.e., 
\begin{equation*}
    \ell_k(n,\theta)\equiv\log\frac{d\Pro_{k,\theta}}{d\nu_k}(\cF_k(n)), \quad \theta \in \Theta_k,
\end{equation*}
and  by $(\Delta\ell)_k(n,\theta)$ its increment at time $n$, i.e., 
\begin{equation*}
    (\Delta\ell)_k(n,\theta)\equiv \ell_k\left(n,\theta\right)-\ell_k\left(n-1,\theta\right), \quad \text{where} \quad 
    \ell_k(0)\equiv 0.
\end{equation*}

We assume that, for each 
$k \in [K]$ and $n \in \bN$, we have an  $\cF(n)$-measurable estimator of $\theta_k$,   $\hat\theta_k(n)$, and we  define the \emph{adaptive log-likelihood} in stream $k$ at time $n$ as 
\begin{equation*}
    \ella_k(n)\equiv \ella_k(n-1)+(\Delta\ell)_k\big(n,\hat\theta_k(n-1)\big),  \quad \text{where} \quad  \ella_k(0) \equiv 0,
\end{equation*}
with $\hat\theta_k(0)$ being an arbitrary deterministic value  in $\Theta_k$.   The tests we propose in the composite testing setup are essentially the same as the ones in Section \ref{section: the proposed test} with the difference that  the LLR statistic in stream $k$ at time $n$,  $\lambda_k(n)$, is now replaced by  the following adaptive log-likelihood ratio:
\begin{equation} \label{def: lambdas}
    \lambdas_k(n)=
    \begin{cases}
    \begin{aligned}
        & \ella_k(n)-\ellg_k^0(n), && \text{if } \ellg_k^0(n)<\min\left\{ \ellg_k^1(n),\ella_k(n) \right\} \\
        & \ellg_k^1(n)- \ella_k(n), && \text{if } \ellg_k^1(n)<\min\left\{ \ellg_k^0(n),\ella_k(n) \right\} \\
        & 0, && \text{otherwise},
    \end{aligned}
    \end{cases}
\end{equation}
where, for each $i \in \{0,1\}$, $\ellg_k^i(n)$ is  the 
\emph{maximum log-likelihood} under $\Theta_k^i$, i.e.,
\begin{equation} \label{generalized ll statistic}
    \ellg_k^i(n)\equiv \sup\left\{\ell_k(n,\theta): \theta\in\Theta_k^i\right\}, \quad i\in\{0,1\}.
\end{equation}

\subsubsection{The decentralized test}
The decentralized test is defined  as in \eqref{def: SPRT}-\eqref{def: decentralized procedure} with  $\lambda_k(n)$ replaced by $\lambdas_k(n)$, as long as   $\hat\theta_k(n)$ is an estimator of the unknown parameter in stream $k$ that depends only on the observations from that stream up to time $n$, i.e.,
$\hat\theta_k(n)$ is   $\cF_k(n)$-measurable for every $k\in[K]$ and $n\in\bN$.

\subsubsection{The proposed test}
The proposed test is given by  \eqref{def: general}, where  $\Tss_{k,1}$  and  $\Tss_{k,2}$ now take the following form:
\begin{itemize}
    \item When $l=u\equiv m$,
    \begin{equation*} \label{def: composite testing, gap}
    \begin{aligned}
        \Tss_{k,1} & \Def \inf\left\{n\in\bN: \lambdas_k(n)> \max\big\{ \lambdas_{(m+1)}(n)+c, \, 0 \big\}, \; \lambdas_{(m+1)}(n)< 0 \right\}, \\
        \Tss_{k,2} & \Def \inf\left\{n\in\bN: \lambdas_k(n) < \min\big\{ \lambdas_{(m)}(n)-d,\,0 \big\}, \; \lambdas_{(m)}(n) > 0 \right\},
    \end{aligned}
    \end{equation*}
    where $c,d>0$ are thresholds to be determined. \\

    \item When $l<u$,
    \begin{align}
        \Tss_{k,1} \equiv & \inf\left\{n\in\bN: \lambdas_k(n)> a\right\} \bigwedge \nonumber \\ 
        & \inf\left\{ n\in\bN: \lambdas_{k}(n) > \max\big\{ \lambdas_{(l+1)}(n)+c, \,0 \big\}, \; \lambdas_{(l+1)}(n) < 0 \right\}, \nonumber \\
        \Tss_{k,2} \equiv & \inf\left\{ n\in\bN: \lambdas_{k}(n)< -b \right\} \bigwedge \label{composite testing, gap-inter} \\
        & \inf\left\{ n\in\bN: \lambdas_{k}(n)< \min\big\{ \lambdas_{(u)}(n)-d,\,0 \big\}, \; \lambdas_{(u)}(n)> 0 \right\}, \nonumber
    \end{align}
    where $a,b,c,d>0$ are thresholds to be determined. 
\end{itemize}

\subsubsection{The synchronous test}
The synchronous test  takes the following form:

\begin{itemize}
\item When $l=u\equiv m$, 
    \begin{equation*} \label{composite, synchronous, known number}
    \begin{aligned}
        \Tc\equiv \inf\left\{n\in\bN: \lambdas_{(m)}(n)-\lambdas_{(m+1)}(n) > c\vee d,\; \lambdas_{(m+1)}(n) < 0 < \lambdas_{(m)}(n) \right\},
    \end{aligned}
    \end{equation*}
    and $\bfDc$ is the subset of the $m$ streams with the largest values of 
    $\lambdas_k(\Tc)$,  $k\in[K]$.
    
\item     When  $l<u$, 
    \begin{equation*} \label{composite, synchronous, lower and upper bounds}
    \begin{aligned}
        \Tc & \equiv \tau_1\wedge\tau_2\wedge\tau_3, \quad \text{where} \\
        \tau_1 & \equiv \inf\left\{ n\in\bN: \lambdas_{(l+1)}(n) < \min 
        \big\{-b, -c+\lambdas_{(l)}(n)\big\}, \; \lambdas_{(l)}(n) > 0 \right\}, \\
        \tau_2 & \equiv \inf\left\{n\in\bN: \lambdas_k(n) \notin (-b,a) \;\forall\, k\in[K], \text{ and } l\leq p^*(n)\leq u \right\}, \\
        \tau_3 & \equiv \inf\left\{ n\in\bN: \lambdas_{(u)}(n)>  \max \big\{ a,  \, d+ \lambdas_{(u+1)}(n)\big\}, \; \lambdas_{(u+1)}(n)< 0 \right\},
    \end{aligned}
    \end{equation*} 
and $\bfDc$ is the subset of the $\left(p^*(\Tc)\vee l\right)\wedge u$ streams with the largest values of $\lambdas_k(\Tc)$, $k\in[K]$, where $p^*(n)\equiv \left| 
\{k\in[K]:\lambdas_k(n)>0\} \right|$.
\end{itemize}

\subsection{Distributional assumptions}
We next state the distributional assumptions that we  need to make in order to  generalize the asymptotic optimality theory of Section \ref{sec: AO}. 

 First of all,  we assume that for every $k\in[K]$ and $\theta, u\in\Theta_k$ there exists a positive number $I_k(\theta, u)$ such that
\begin{gather}
    \Pro_{k,\theta} \left( \limsup_{n\to\infty} \frac{\ell_k(n,\theta)-\ell_k(n,u)}{n} \leq I_k(\theta, u) \right)=1. \label{assump1.1}
\end{gather}

Second, we assume that the null and alternative hypotheses in each stream are separated, in the sense that, for each $k\in [K]$, 
\begin{align} \label{assump2}
\begin{split}
 I_k(\theta) &\equiv \inf_{u \in\Theta_k^0} I_k(\theta, u) >0, \quad \forall \, \theta \in\Theta_k^1, \\
  J_k(\theta) &\equiv \inf_{u\in\Theta_k^1} I_k(\theta,u)>0,  \quad  \, \forall \,  \theta \in\Theta_k^0.
  \end{split}
  \end{align}

Finally, we assume that for each $k\in [K]$ and $\epsilon>0$,
\begin{equation} \label{assump3}
\begin{gathered}
    \sum_{n=1}^\infty \Pro_{k,\theta}\left( \frac{\ella_k(n)-\ell_k^0(n)}{n}\leq I_k(\theta)-\epsilon \right)<\infty, \quad \forall\,\theta \in\Theta_k^1, \\
    \sum_{n=1}^\infty \Pro_{k,\theta}\left( \frac{\ella_k(n)-\ell_k^1(n)}{n}\leq J_k(\theta) -\epsilon \right)<\infty, \quad \forall\,\theta \in\Theta_k^0.
\end{gathered}
\end{equation}

\begin{remark} Assumptions   \eqref{assump1.1}-\eqref{assump3} are satisfied when,  for example, $X_k$ is an i.i.d. sequence whose distribution belongs to some multi-parameter exponential family, both the null and the  alternative parameter spaces, $\Theta_0^k$ and $\Theta_1^k$,  are compact, and  $\hat{\theta}_k(n)$  is the maximal likelihood estimator based on the observations from stream $k$ up to time $n$, i.e.,
    $$ \hat\theta_k(n)\equiv \operatorname{argmax}\left\{\ell_k(n,\theta_k): \theta_k\in \Theta_k^0\cup\Theta_k^1\right\}
    $$
 (see, e.g.,   \cite[Appendix E]{Song_AoS}). 
\end{remark}

\subsection{Asymptotic optimality}
We are now ready to generalize the asymptotic optimality theorems of Section \ref{sec: AO}  to the case of general composite hypotheses.  To state these results, for any $\cA\subseteq [K]$ and $\bftheta=(\theta_1, \ldots, \theta_K)\in\bfTheta_\cA$ we set 
$$ \cI_\cA(\bftheta)\equiv \min_{i\in\cA} I_i(\theta_i) \qquad \text{ and } \qquad \cJ_\cA(\bftheta)\equiv \min_{j\notin\cA} J_j(\theta_j). $$

\begin{theorem} \label{thm: composite hypotheses}
Let $l,u$ satisfy \eqref{lu}.
Suppose that the thresholds of $\hat\chi$  are  selected  so that $\hat\chi\in\Delta^{*}(\alpha,\beta,\Pi_{l,u})$ for any  $\alpha,\beta\in(0,1)$ and also \eqref{thres_proposed} holds, e.g., according to \eqref{gap rule, c, d} when $l=u$ and \eqref{gap-inter rule, a, b, c, d} when $l<u$.

If conditions \eqref{assump1.1}-\eqref{assump3} hold for every $k\in[K]$, then, as $\alpha,\beta\to 0$,
        \begin{equation} \label{AO, i, composite}
            \begin{aligned}
                \Exp_{\cA,\bftheta}[\Tss_i] & \sim \inf\big\{ \Exp_{\cA,\bftheta}[T_i]: \chi\in\Delta^*(\alpha,\beta,\Pi_{l,u}) \big\} \\
                & \sim \frac{|\log\alpha|}{I_i(\theta_i)+\cJ_\cA(\bftheta)\cdot \bfone\left\{ |\cA|=\ell \right\}}
            \end{aligned}
        \end{equation}
        \begin{equation} \label{AO, j, composite}
            \begin{aligned}
               \Exp_{\cA,\bftheta}[\Tss_j] & \sim \inf\big\{ \Exp_{\cA,\bftheta}[T_j]: \chi\in\Delta^*(\alpha,\beta,\Pi_{l,u}) \big\} \\
               & \sim \frac{|\log\beta|}{J_j(\theta_j)+\cI_\cA(\bftheta)\cdot \bfone\left\{ |\cA|=u \right\}}
            \end{aligned}
        \end{equation}
simultaneously for every $i\in\cA$, $j\notin\cA$, $\cA\in\Pi_{l,u}$, $\bftheta\in\bfTheta_\cA$. 
\end{theorem}
\begin{proof}
Appendix \ref{Proof related to composite hypotheses}. \\
\end{proof}

\begin{theorem} \label{thm: composite hypotheses, decentralized procedure}
Let  $\Pi$ satisfy \eqref{no trivial}. 
Suppose that the thresholds of $\tilde\chi$ are selected so that $\Tilde{\chi}\in\Delta^*(\alpha,\beta,\Pi)$ for any  $\alpha,\beta\in(0,1)$ and also \eqref{thres_decentralized} holds, e.g., according to \eqref{intersection rule, a, b_general}.

If conditions \eqref{assump1.1}-\eqref{assump3} hold for every $k\in[K]$, then, as $\alpha,\beta\to 0$,
   \begin{align*}
        \begin{split}
            \Exp_{\cA,\bftheta}[\Tilde{T}_i] &\sim 
            \inf \left\{ \Exp_{\cA,\bftheta}[T_i]: \chi\in\Delta'\cap\Delta^*(\alpha,\beta,\Pi) \right\} \sim \frac{|\log\alpha|}{I_i(\theta_i)} \\
            \Exp_{\cA,\bftheta}[\Tilde{T}_j] &\sim \inf\big\{ \Exp_{\cA,\bftheta}[T_j]: \chi\in\Delta'\cap\Delta^*(\alpha,\beta,\Pi) \big\} \sim \frac{|\log\beta|}{J_j(\theta_j)}
        \end{split}
        \end{align*}
simultaneously   for every $i\in\cA$, $j\notin\cA$, $\cA\in\Pi$, 
    $\bftheta\in\bfTheta_\cA$.
\end{theorem}

\begin{proof}
Appendix \ref{Proof related to composite hypotheses}. \\
\end{proof}

\begin{theorem} \label{thm: composite hypotheses, synchronous decision making}
Let $l,u$ satisfy \eqref{lu}. Suppose that the thresholds of $\check\chi$ are selected so that $\check{\chi}\in\Delta^*(\alpha,\beta,\Pi_{l,u})$ for any $\alpha,\beta\in(0,1)$ and also \eqref{thres_proposed} holds, e.g., according to \eqref{gap rule, c, d} when $l=u$ and \eqref{gap-inter rule, a, b, c, d} when $l<u$.

If conditions \eqref{assump1.1}-\eqref{assump3} hold for every $k\in[K]$, then, as $\alpha,\beta\to 0$,
        \begin{equation*}
        \begin{split}
            \Exp_{\cA,\bftheta}[\Tc] & \sim \inf\big\{ \Exp_{\cA,\bftheta}[T_1]: \chi\in\Delta''\cap\Delta^*(\alpha,\beta,\Pi_{l,u}) \big\} \\
            & \sim \max\left\{ \frac{|\log\alpha|}{\cI_\cA(\bftheta)+\cJ_\cA(\bftheta)\cdot \bfone\{|\cA|=l\}}, \; \frac{|\log\beta|}{\cJ_\cA(\bftheta)+\cI_\cA(\bftheta)\cdot\bfone\{|\cA|=u\}}
        \right\}
        \end{split}
        \end{equation*}
simultaneously    for  every  $\cA\in\Pi_{l,u}$, $\bftheta\in\bfTheta_\cA$.
\end{theorem}
\begin{proof}
Appendix \ref{Proof related to composite hypotheses}.\\
\end{proof}

\section{Generalization to other global error metrics} \label{section: generalization to other GEM}
In this section, we discuss the extension of the asymptotic optimality theory of Section \ref{sec: AO} to various error metrics beyond the classical familywise error rates. For this, we follow the approach in \cite{Bartroff2021}, where the asymptotic optimality theory for synchronous tests in \cite{Song_prior} was
extended to general error metrics. 

For simplicity, we focus on the case of simple hypotheses 
and in the place of the familywise error rates, $\text{FWE}_\cA^1(\chi)$ and $\text{FWE}_\cA^2(\chi)$, of a test $\chi\in\Delta$, we consider  the type-I and type-II versions of a generic error metric
 $$\text{GEM}_\cA^1(\chi) \qquad \text{ and } \qquad \text{GEM}_\cA^2(\chi).$$ 

Then,  $\hat\chi$, $\Tilde{\chi}$, $\check\chi$ can be designed to control these error metrics and preserve their asymptotic optimality properties  as long as there exist $C_1, C_2>0$ so that,  for both $i\in\{1,2\}$ and every  $\cA\in\Pi$, 
\begin{align} 
    \text{GEM}_\cA^i(\chi) &\leq C_1\cdot \text{FWE}_\cA^i(\chi) \quad \text{for all} \quad \chi \in \{\hat{\chi}, \tilde{\chi}, \check{\chi}\}, \label{C1} \\
      \text{and} \quad       \text{GEM}_\cA^i(\chi) &\geq C_2\cdot \text{FWE}_\cA^i(\chi) \quad \text{for all} \quad  \chi \in\Delta. \label{C2}
        \end{align} 
To see this, for any $\alpha, \beta \in (0,1)$ and  $\Pi$ we set
\begin{equation*} 
\begin{aligned}
    \Delta^\text{GEM}(\alpha,\beta,\Pi)\Def \bigg\{  \chi\in\Delta : \; 
    \text{GEM}_\cA^1(\chi) \leq \alpha \; \text{ and } \; \text{GEM}_\cA^2(\chi) \leq \beta,  \; \forall \; \cA\in\Pi  \bigg\}.
\end{aligned}
\end{equation*}
Then,  by \eqref{C1} we have 
$$\chi \in \{\hat{\chi}, \tilde\chi,  \check\chi\} \quad \text{and} \quad \chi\in\Delta(\alpha/C_1,\beta/C_1,\Pi) \quad \Rightarrow \quad \chi\in\Delta^\text{GEM}(\alpha,\beta,\Pi),$$ whereas \eqref{C2} implies that, for every $A\in\Pi$ and $k\in[K]$, as $\alpha,\beta\to 0$, 
\begin{equation*}
\begin{split}
    \inf\left\{ \Exp_\cA[T_k]: \chi\in\Delta^\text{GEM}(\alpha,\beta,\Pi) \right\} & \geq \cL_{k,\cA}(\alpha/C_2,\beta/C_2,\Pi) \sim \cL_{k,\cA}(\alpha,\beta,\Pi) \\
    \inf\left\{ \Exp_\cA[T_k]: \chi\in \Delta'\cap\Delta^\text{GEM}(\alpha,\beta,\Pi) \right\} & \geq \cL'_{k,\cA}(\alpha/C_2,\beta/C_2,\Pi) \sim \cL'_{k,\cA}(\alpha,\beta,\Pi) \\
    \inf\left\{ \Exp_\cA[T_1]: \chi\in \Delta''\cap\Delta^\text{GEM}(\alpha,\beta,\Pi) \right\} & \geq \cL''_{\cA}(\alpha/C_2,\beta/C_2,\Pi) \sim \cL''_{\cA}(\alpha,\beta,\Pi). 
\end{split}
\end{equation*} \\

We next provide some concrete examples of error metrics that satisfy the above conditions when $\Pi$ is of the form $\Pi_{l,u}$. 

\begin{itemize}
    \item [(i)] \textit{Per-comparison error rates (PCE)}: the expected proportion of type-I or type-II errors among all streams:
    \begin{equation*}
        \text{PCE}_\cA^1(\chi) \equiv \frac{\Exp_\cA[|\bfD\backslash\cA|]}{K} \qquad \text{and} \qquad \text{PCE}_\cA^2(\chi) \equiv \frac{\Exp_\cA[|\cA\backslash\bfD|]}{K}.
    \end{equation*}
    \item [(ii)] \textit{False discovery rate (FDR) and false non-discovery rate}: the expected proportion of type-I errors among all streams identified as signals and the expected proportion of type-II errors among all streams identified as noises:
    \begin{equation*}
        \text{FDR}_\cA^1(\chi) \equiv \Exp_\cA\left[\frac{|\bfD\backslash\cA|}{|\bfD|}\right] \qquad \text{and} \qquad \text{FDR}_\cA^2(\chi) \equiv \Exp_\cA\left[\frac{|\cA\backslash\bfD|}{K-|\bfD|}\right],
    \end{equation*}
    with the convention that $0/0=0$.
    \item [(iii)] \textit{Positive false discovery rate (pFDR) and positive false non-discovery rate}: the expected proportion of type-I (resp. -II) errors among all streams identified as signals (resp. noises) given that there are streams identified as signals (resp. noises):
    \begin{align*}
        \text{pFDR}_\cA^1(\chi)  &\equiv \Exp_\cA\left[ \left.\frac{|\bfD\backslash\cA|}{|\bfD|}  \;   \right\vert  \; |\bfD|\geq 1 \right] \\
         \text{pFDR}_\cA^2(\chi)  &\equiv \Exp_\cA\left[\left. \frac{|\cA\backslash\bfD|}{K-|\bfD|} \; \right   \vert  \;   K-|\bfD|\geq 1 \right]. 
    \end{align*}
\end{itemize} 

Indeed,  by \cite{Bartroff2021} it follows that the inequalities  \eqref{C1}-\eqref{C2} hold when $\Pi=\Pi_{l,u}$ and 
    \begin{itemize}
        \item [(i)] $\text{GEM}=\text{PCE}$, with  $C_1=(u\vee(K-l))/K$, $C_2=1/K$.
        \item [(ii)] $\text{GEM}=\text{FDR}$, with  $C_1=1$, $C_2=1/K$.
        \end{itemize}
The  next Proposition deals with the error metric in (iii).

\begin{proposition} \label{prop: GEM}
Let  $\Pi=\Pi_{l,u}$, where    $0<l\leq u<K$, and suppose that $\text{FWE}_A^i(\chi)\leq 1/2$, where  $i\in\{0,1\}$ and $\chi\in\{\hat\chi,\tilde\chi,\check\chi\}$. 
If  $\text{GEM}=\text{pFDR}$, then  
the inequalities  \eqref{C1}-\eqref{C2} hold with  $C_1=2$, $C_2=1/K$.  
\end{proposition}
\begin{proof}
    See Appendix \ref{Proofs related to other global error metrics}.
\end{proof}

\section{Conclusion and open problems} \label{section: conclusion}
In this work  we consider the problem of simultaneously testing the marginal distributions of  multiple, sequentially monitored data streams. We  introduce a general formulation  in which the decisions for the various testing problems can be made at different times, using data from all  streams, 
and all streams can be monitored  until all decisions have been made.  Assuming  a priori bounds on the number of signals, we propose a  novel sequential multiple testing procedure that takes advantage of the flexibility of this formulation. 
Under general distributional assumptions and in the case of general parametric  composite hypotheses,  
 we show  that the proposed procedure  minimizes the expected decision time, simultaneously for every data stream and for every signal configuration,  asymptotically as the target  error rates go to zero.  In this asymptotic regime, we also  evaluate  the factor by which the expected decision time in every stream increases when one is limited to \emph{decentralized} procedures, where only local data can be used for each testing problem, or \textit{synchronous} procedures, where all decisions must be made at the same time.

There are various directions  for further study. First of all, it is interesting to develop an analogous asymptotic optimality theory  when  sampling from a data stream  must be terminated once the decision  for the  corresponding testing problem has been made, and data from all  streams can be utilized for  each testing problem. This setup, which is considered, for example, in  \cite{Bartroff_2014},  is  more restrictive  than the one we consider in the present paper, where all streams can be monitored until all decisions have been made. However, it is more flexible than a decentralized setup, where it is possible to use only local data for each testing problem. The results of the present paper indicate that this additional flexibility does not  provide any gains, as far as it concerns first-order asymptotic optimality, when there is no prior information regarding the subset of signals.  An asymptotic optimality theory in the case where non-trivial prior information is available  is an open problem.

Second,  it is interesting to extend the  formulation of the present paper to the case where it is possible to observe only a subset of  data streams at each time instant, which is determined by the practitioner based on the already collected observations. Such sampling constraints have been considered in various works, such as  \cite{Kobi_2015a, Kobi_2018, Kobi_2020_composite, Aris_IEEE, Prabhu2022}, but in all of them  a common decision time is assumed.

Finally, another  direction of interest is to consider error metrics that are not  bounded by the familywise error rates up to multiplicative constants and, as a result, the results of Section \ref{section: generalization to other GEM} do not apply. This is the case, for example, for the generalized familywise error rates 
\cite{Lehmann_2005}  or the generalized false discovery/non-discovery rate 
\cite{k-FDR}.   The former has been considered in the case of \emph{synchronous} procedures in \cite{Song_AoS}.  An asymptotic optimality theory with such error metrics in the general family of sequential multiple testing procedures of the present work is an open problem.

\appendix
\section{}\label{Supporting lemmas}
In this Appendix, we state and prove three  supporting lemmas. The first two are  used in the proof of the asymptotic upper bounds in Propositions \ref{prop: AUB} and \ref{prop: higher-order AUB}, whereas  the third is used in the proof of the asymptotic lower bound in Lemma \ref{Lemma, ALB}. 

\begin{lemma} \label{basic lemma for AUB}
    Let $\xi_l\Def\{ \xi_l(n),\, n\in \bN \}$, $l\in[L]$ be $L\in\bN$ stochastic processes on some probability space $(\Omega,\cF,\Pro)$. 
    For any $\bfa = \{a_1,\ldots, a_L\}$ where  $a_l \geq 0$ for every $l\in [L]$, set 
    $$ \nu(\bfa) \equiv \inf\{ n\in\bN: \, \xi_l(n) \geq a_l \text{ for every } l\in [L] \}. $$
    If for each $l\in [L]$ there exists $\mu_l>0$ so that
    \begin{equation*} \label{complete conv in the lemma}
    \forall \; \epsilon>0 \quad     \sum_{n=1}^\infty \Pro\left( \frac{1}{n} \xi_l(n) \leq \mu_l -\epsilon \right) < \infty,
    \end{equation*}
    then, as $\min_{l\in[L]} a_l \to \infty$, 
    \begin{equation*} \label{upper bound in the lemma}
        \Exp[\nu(\bfa)] \lesssim \max_{l\in[L]} \left\{ \frac{a_l}{\mu_l} \right\},
    \end{equation*} 
    where $\Exp$ is the expectation under $\Pro$.
\end{lemma}
\begin{proof}
    See \cite[Lemma F.2]{Song_AoS}. \\
\end{proof}

\begin{lemma} \label{lemma, for higher-order upper bound}
    Let $\xi_l,\;l\in[L]$ and $\nu(\bfa)$ be defined the same as in Lemma \ref{basic lemma for AUB}. If,  for each $l\in[L]$, $\{\xi_l(n)-\xi_l(n-1),\; n\in\bN\}$ where $\xi_l(0)=0$ are i.i.d.,
    $$ \Exp[\xi_l(1)]=\mu_l>0 \quad \text{and} \quad \operatorname{Var}(\xi_l(1))<\infty, $$ then, as $\min_{l\in [L]} a_l \to \infty$,
    $$ \Exp[\nu(\bfa)] \leq \max_{l\in[L]} \left\{ \frac{a_l}{\mu_l} \right\} + O\left( \sum_{l\in[L]}\sqrt{\frac{a_l}{\mu_l}} \right) \leq \max_{l\in[L]} \left\{ \frac{a_l}{\mu_l} \right\} + O\left( L\sqrt{\max_{l\in[L]} a_l} \right). $$
\end{lemma}
\begin{proof}
    See \cite[Lemma A.2.]{Song_prior}. \\
\end{proof}

\begin{lemma} \label{lemma used for proving the ALB}
    Let $\xi\Def\{\xi(n),\,n\in\bN\}$ be a stochastic process on some probability space $(\Omega,\cF,\Pro)$. Suppose that 
    $$ \Pro\left( \limsup_{n\to \infty} \frac{\xi(n)}{n} \leq  \mu \right)=1, $$
    then for any $\epsilon>0$ we have
    $$ \lim_{m\to \infty} \Pro\left( \frac{1}{m} \max_{1\leq n\leq m} \xi(n)>(1+\epsilon)\mu \right)=0. $$ 
\end{lemma}
\begin{proof} 
Fix $\epsilon>0$ and $1\leq N< m$. 
Denote 
\begin{align*}
    \Pro_{m,N}(\epsilon) & \Def \Pro\left( \frac{1}{m} \max_{1\leq n\leq N} \xi(n)>(1+\epsilon)\mu \right) \\
    U_N(\epsilon) & \Def \left\{ \sup_{n>N} \frac{\xi(n)}{n}> (1+\epsilon)\mu \right\}.
\end{align*}
By the union bound, we have 
\begin{align*}
 \Pro\left( \frac{1}{m} \max_{1\leq n\leq m} \xi(n)>(1+\epsilon)\mu \right)
    &\leq \Pro_{m,N}(\epsilon)+ \Pro\left( \frac{1}{m} \max_{N< n\leq m} \xi(n)>(1+\epsilon)\mu \right) \\
  &\leq  \Pro_{m,N}(\epsilon)+ \Pro\left( \max_{N< n\leq m} \frac{\xi(n)}{n}>(1+\epsilon)\mu \right) \\
 & \leq  \Pro_{m,N}(\epsilon)+  \Pro(U_N(\epsilon)).
\end{align*}
Since $\xi(n)$ is finite almost surely for every $n\in\bN$, 
we have  $\Pro_{m,N}(\epsilon)\to 0$ as $m\to\infty$ and, consequently, 
\begin{align*}
    \limsup_{m\to \infty} \Pro\left( \frac{1}{m} \max_{1\leq n\leq m} \xi(n)>(1+\epsilon)\mu \right)\leq \Pro(U_N(\epsilon)).
\end{align*}
Finally,  by the assumption of the Lemma it follows that,
as  $N\to \infty$,  we have 
\begin{align*}
    \limsup_{N \to \infty} \Pro(U_N(\epsilon)) & \leq \Pro\left(\inf_{N\geq 1} \sup_{n>N} \frac{\xi(n)}{n} \geq (1+\epsilon)\mu \right) \\
    &\leq \Pro\left( \limsup_n \frac{\xi(n)}{n}  >\mu \right) =0.
\end{align*} 
\end{proof}

\section{} \label{Proof related to AUB}
In this Appendix we prove Propositions
\ref{thm: a.s. finite and error control}, \ref{prop: AUB}, and \ref{prop: higher-order AUB}.   For any $\cA,\cC\subseteq [K]$ and $n\in\bN$, we denote by $\lambda_{\cA,\cC}(n)$ the log-likelihood ratio of $\Pro_\cA$ versus $\Pro_\cC$ when both measures are restricted to $\cF(n)$, i.e., 
\begin{equation} \label{Definition lambda^A,C(n)}
    \lambda_{\cA,\cC}(n) \Def \log \frac{d\Pro_\cA}{d\Pro_\cC} (\cF(n)) =\sum_{k\in \cA\backslash\cC} \lambda_k(n) - \sum_{k\in \cC\backslash\cA} \lambda_k(n),
\end{equation}    
where the equality follows by  \eqref{Product measure}.  We  set $\lambda_{\cA,\cC}\equiv \{\lambda_{\cA,\cC}(n), n\in\bN\}$.
Moreover,  if  $\Exp$ represents expectation with respect to  a probability measure  $\Pro$, for an  event $\Gamma$ and a  random variable $X$ we write
$$ \Exp[X;\Gamma]\Def \int_\Gamma X d\Pro. $$

\begin{proof}[Proof of Proposition \ref{thm: a.s. finite and error control}] The a.s. finiteness and the familywise error rate control  of the decentralized  and the synchronous test are established in  \cite{De_Baron_Seq_Bonf} and \cite{Song_prior}, respectively. Thus, it suffices to show the corresponding results for the  proposed test.\\

Fix $l,u$ that satisfy \eqref{lu}, $\cA\in\Pi_{l,u}$, and $a,b,c,d>0$.
We first show 
\begin{equation} \label{a.s. finite of the proposed test}
    \Pro_\cA(\Ts_i<\infty) =      \Pro_\cA(\Ts_j<\infty) =1, \quad \forall\; i \in A, \; j \notin A. 
\end{equation}     
Note that if $l<u$, 
\begin{align}
    \Ts_i\leq \Ts_{i,1} & \leq \inf\left\{ n\in\bN: \lambda_k(n) \geq a \right\}, \label{stopping time of passing a} \\
    \Ts_j\leq \Ts_{j,2} & \leq \inf\left\{ n\in\bN: \lambda_k(n) \leq -b \right\}, \label{stopping time of passing -b}
\end{align}
and if $l=u$,
\begin{align}
    \Ts_i\leq \Ts_{i,1} & \leq \inf\left\{ n\in\bN: \lambda_i(n) \geq \lambda_{k}(n)+c \text{ for every } k\notin\cA \right\}, \label{upper bound on hat Ti, the gap part}\\
    \Ts_j\leq \Ts_{j,2} & \leq \inf\left\{ n\in\bN: \lambda_j(n) \leq \lambda_{k}(n)-d \text{ for every } k\in\cA \right\}.
\end{align}
All upper bounds are  a.s. finite under $\Pro_\cA$ because of \eqref{orthogonal}. \\

Next, we only prove the upper bounds on the type-I familywise error rate, i.e., \eqref{gap rule, type I} and \eqref{gap-intersection rule, type I}, as  those on the type-II familywise error rate can be proved similarly.
By the definition of the type-I familywise error in  \eqref{def of FWE1} and Boole's inequality we have
\begin{equation} \label{write FWE as prob of a union}
    \text{FWE}_\cA^1(\hat\chi)=\Pro_\cA\left( \bigcup_{j\notin\cA} \{\Ds_j=1\} \right)\leq \sum_{j\notin\cA} \Pro_\cA(\Ds_j=1).
\end{equation}
Thus, in what follows, we assume that $A\neq [K]$, fix $j \notin A$, and we upper bound $\Pro_\cA(\Ds_j=1)$.\\

If $l=0$, $u=K$, then
\begin{equation} \label{hatGammaj}
  \{\hat D_j=1\} = \left\{ \lambda_j(\Ts_j)\geq a \right\}  \equiv \hat{\Gamma}_j.
\end{equation}
Let  $\cC=\cA\cup\{j\}$. Then,  by \eqref{Definition lambda^A,C(n)} we have  $\lambda_{\cA,\cC}=-\lambda_j$ and  by Wald's likelihood ratio identity we  obtain:
\begin{equation} \label{PcA(Gammaj)}
    \Pro_\cA(\hat\Gamma_j) = \Exp_\cC\left[ 
    \exp\{ \lambda_{\cA,\cC}(\Ts_j) \} ; \hat\Gamma_j \right] = \Exp_{\cC} \left[ \exp\{ -\lambda_j(\Ts_j) \}; \hat\Gamma_j \right] \leq e^{-a}. 
\end{equation}

If $l=u$, then 
\begin{equation} \label{type I implies Gamma j,i}
    \{\Ds_j=1\} = \bigcup_{k\in\cA} \hat\Gamma_{j,k},
\end{equation}
where 
\begin{equation} \label{hatGammaj,i}
    \hat\Gamma_{j,k} \Def \left\{ \lambda_j(\Ts_j) \geq \lambda_k(\Ts_j)+c \right\}, \quad k\in\cA.
\end{equation}
To see \eqref{type I implies Gamma j,i}, we observe that since $l=u\equiv m$, we know that there are exactly $K-m-1$ noises other than stream $j$ itself. Thus, when $\{\Ds_j=1\}$ occurs at time $\Ts_j$, among the $K-m$ other streams whose LLRs $\lambda_j$ is greater than by at least $c$  there must be  at least one signal.  Therefore, by Boole's inequality we have:
$$\Pro_\cA(\Ds_j=1) \leq  \sum_{k\in\cA} \Pro_\cA(\hat\Gamma_{j,k}). $$
Let $k \in A$ and set  $\cC=\cA\cup\{j\}\backslash\{k\}$. Hence,   $\lambda_{\cA,\cC}=\lambda_k-\lambda_j$.  By Wald's likelihood ratio identity we have
\begin{equation} \label{PcA(Gammaj,i)}
\begin{split}
    \Pro_\cA(\hat\Gamma_{j,k})& =\Exp_\cC\left[ \exp\{\lambda_{\cA,\cC}(\Ts_j)\};\hat\Gamma_{j,k} \right]\\
    &=\Exp_\cC\left[ \exp\{\lambda_k(\Ts_j)-\lambda_j(\Ts_j)\};\hat\Gamma_{j,k} \right] \leq e^{-c},
\end{split}
\end{equation}
which implies  \eqref{gap rule, type I}.  \\

If $l<u$ and either $l>0$ or $u<K$, then 
\begin{equation*}
    \{\Ds_j=1\} = \hat\Gamma_j \cup \left(\bigcup_{k\in\cA} \hat\Gamma_{j,k}\right),
\end{equation*}
where $\hat\Gamma_j$ is defined as in \eqref{hatGammaj} and $\hat\Gamma_{j,k}$  as in \eqref{hatGammaj,i}.
By Boole's inequality, we have 
$$ \Pro_\cA(\Ds_j=1)  \leq  \Pro_\cA(\hat\Gamma_j) + \sum_{k \in \cA} \Pro_\cA(\hat\Gamma_{j,k}), $$
and to obtain \eqref{gap-intersection rule, type I} it suffices to apply  \eqref{PcA(Gammaj)} and \eqref{PcA(Gammaj,i)}. \\
\end{proof}

\begin{proof} [Proof of Proposition \ref{prop: AUB}]
The bounds for the parallel SPRT are well known and can be found for example in  \cite[Chapter 3]{Tartakovsky_Book}). Thus, it suffices to show the  upper bounds related to the proposed and  the synchronous test. 

We first consider  the proposed test. We only prove the upper bounds for $i\in\cA$, as those for $j\notin\cA$ can be proved similarly. 

If $l=u$, in view  of condition \eqref{Complete convergence, assumption for AUB},  we apply Lemma \ref{basic lemma for AUB} to the stopping time  in the upper bound of \eqref{upper bound on hat Ti, the gap part} and obtain, as $c\to \infty$,
    \begin{equation} \label{AUB with gap for i}
        \Exp_\cA[\Ts_i] \lesssim \max_{k\notin\cA} \left\{ \frac{c}{I_i + J_k} \right\} = \frac{c}{I_i+\cJ_\cA}.        
    \end{equation}

If $l<u$, in view of condition \eqref{Complete convergence, assumption for AUB}, we apply Lemma \ref{basic lemma for AUB} to the stopping time in the upper bound of \eqref{stopping time of passing a} and obtain, as $a\to\infty$,
    \begin{equation} \label{AUB with intersection, for i}
        \Exp_\cA[\Ts_i]\lesssim \frac{a}{I_i}.
    \end{equation}
    When, in particular, $|\cA|=l$, \eqref{upper bound on hat Ti, the gap part} and  \eqref{stopping time of passing a}, thus, \eqref{AUB with gap for i} and \eqref{AUB with intersection, for i}, both hold. Thus, the  proof for the proposed test is complete. \\

Next,   we  focus on the synchronous test.

If $l=u$, then
    $$ \Tc\leq \inf\left\{ n\in\bN: \lambda_i(n)-\lambda_j(n)\geq c\vee d \;\text{ for every }\; i\in\cA, \; j\notin\cA \right\}. $$
 Applying Lemma \ref{basic lemma for AUB} to 
    the stopping time in the upper bound, in view of condition \eqref{Complete convergence, assumption for AUB},    we have
    $$ \Exp_\cA[\Tc] \lesssim \max_{i\in\cA, \, j\notin\cA} \left\{\frac{c\vee d}{I_i+J_j}\right\} = \frac{c\vee d}{\cI_\cA+\cJ_\cA}. $$
    
    If $l<u$, then 
    $$ \Tc\leq \tau_2 \leq \inf\{n\in\bN: \lambda_i(n)\geq a \text{ and } \lambda_j(n)\leq -b \text{ for every $i\in\cA$ and $j\notin\cA$}\},  $$
and    applying again Lemma \ref{basic lemma for AUB} to 
    the stopping time in the upper bound, in view  of \eqref{Complete convergence, assumption for AUB}, we have 
    \begin{equation} \label{tau'2}
        \Exp_\cA[\Tc] \lesssim
    \max\left\{ \max_{i\in\cA}\left\{\frac{a}{I_i}\right\},\;\max_{j\notin\cA}\left\{\frac{b}{J_j}\right\} \right\}=\max\left\{ \frac{a}{\cI_\cA},\;\frac{b}{\cJ_\cA} \right\}.        
    \end{equation}
    When, in particular, $|\cA|=l$, 
    $$ \Tc\leq \tau_1\leq \inf\{n\in\bN: \lambda_j(n)\leq -b \text{ and } \lambda_i(n)-\lambda_j(n)\geq c \text{ for every } i\in\cA,\; j\notin\cA \}, $$
    and  applying again  Lemma \ref{basic lemma for AUB} 
    we have
    \begin{equation} \label{tau'1}
    \begin{split}
        \Exp_\cA[\Tc] & \lesssim \max \left\{ \max_{j\notin\cA}\left\{\frac{b}{J_j}\right\},\; \max_{i\in\cA,\,j\notin\cA}\left\{\frac{c}{I_i+J_j}\right\} \right\} \\
        & = \max\left\{ \frac{b}{\cJ_\cA}, \; \frac{c}{\cI_\cA+\cJ_\cA} \right\}. 
    \end{split}        
    \end{equation}
    Thus, when $|\cA|=l<u$,   combining \eqref{tau'2} and \eqref{tau'1}  we have 
    $$ \Exp_\cA[\Tc]\lesssim \max\left\{ \frac{b}{\cJ_\cA},\;\frac{a}{\cI_\cA}\wedge\frac{c}{\cI_\cA+\cJ_\cA} \right\}. $$
    Similarly, when $|\cA|=u>l$, we have 
    \begin{equation*}
        \Exp_\cA[\Tc]\lesssim \max\left\{ \frac{a}{\cI_\cA},\;\frac{b}{\cJ_\cA}\wedge\frac{d}{\cJ_\cA+\cI_\cA} \right\}.
    \end{equation*}
 Thus,  the proof for the synchronous test is complete. \\
\end{proof}

\begin{proof} [Proof of Proposition \ref{prop: higher-order AUB}] The bounds for the parallel SPRT are well known and can be found for example in  \cite[Chapter 3]{Tartakovsky_Book}), whereas those for the synchronous test can be found in  \cite[Lemma 5.2]{Song_prior}.  Thus, it suffices to show the  upper bounds related to the proposed test. 
  We only prove those for  $i\in\cA$, i.e., \eqref{AUB with higher-order term, proposed, gap, i} and \eqref{AUB with higher-order term, proposed, gap-inter, i}, as those for $j\notin\cA$ can be proved similarly.   
    
    If $l=u$, applying Lemma \ref{lemma, for higher-order upper bound} to the upper bound in  \eqref{upper bound on hat Ti, the gap part} we have,  as $c\to\infty$,
    \begin{equation} \label{also holds}
        \Exp_\cA[\Ts_i]\leq \max_{k\notin\cA}\left\{ 
\frac{c}{I_i+J_k} \right\} + O\left( |\cA^c|\sqrt c \right) = \frac{c}{I_i+\cJ_\cA} + O\left( (K-m)\sqrt c \right).
    \end{equation}
    
If $l<u$, we have
    \begin{equation}
        \Exp_\cA[\Ts_i]\leq \Exp_\cA[\Tilde{T}_{i}]\leq \frac{a}{I_i}+O(1).
    \end{equation}
If, also,  $|\cA|=l$ and $|c-a|=O(1)$ as $a,c\to\infty$, then \eqref{also holds} also holds with $c$ replaced by $a$. 
\end{proof}

\section{} \label{Proof related to ALB}
In this Appendix, we prove all results in Section \ref{sec: AO}.

\begin{proof}[Proof of Lemma \ref{Lemma, ALB}]
    We only prove the asymptotic lower bound in  \eqref{ALB, i, general},  as  the one in \eqref{ALB, j, general} can be proved similarly.
    We fix $\cA\in\Pi$ and  $i\in\cA$.
By \eqref{no trivial}, there exists a $\cC\in\Pi$ so that $i\notin\cC$. It then suffices to show that for any such  $\cC$, as $\alpha,\beta\to 0$,
\begin{equation} \label{for an arbitrary cC}
\frac{\cL_{i,\cA}(\alpha,\beta,\Pi)}{|\log\alpha|/I_{\cA,\cC}}\gtrsim 1.
\end{equation}

Fix arbitrary such $\cC$.    
    To show \eqref{for an arbitrary cC}, we fix $\alpha,\beta\in (0,1)$ and $\chi\in\Delta(\alpha,\beta,\Pi)$.
    We further fix an $\epsilon>0$ and set 
    $$ N=N_{\cA,\cC}(\alpha,\epsilon) \Def \frac{|\log\alpha|}{I_{\cA,\cC}} (1-\epsilon). $$
    By Markov's inequality it follows that
    \begin{equation*}
        \frac{\Exp_\cA[T_i]}{N} \geq \Pro_\cA(T_i>N) 
        = 1- \Pro_\cA\left(T_i\leq N, \; D_i=1\right) - \Pro_\cA\left(T_i\leq N, \; D_i=0\right).
    \end{equation*}
  Since $\{D_i=0\}$ is a type-II error under $\Pro_\cA$,  the second probability in the lower bound is upper bounded by $\beta$. The first one is upper bounded by 
    \begin{equation*}
        \begin{split}
            \Pro_\cA\left( D_i=1, \; \lambda_{\cA,\cC}(T_i)
            \leq \log \delta  \right) + \Pro_\cA\left( T_i\leq N, \; \lambda_{\cA,\cC} (T_i)>\log \delta \right) \Def I + II,
        \end{split}
    \end{equation*}
    where 
      $$ \delta \Def \alpha^{-1+\epsilon^2} \quad \text{or equivalently} \quad \log \delta= (1-\epsilon^2) |\log\alpha|.$$
    By Wald's likelihood ratio identity we have
    \begin{equation*}
        \begin{split}
            I & = \Exp_\cC\big[\exp\left\{ \lambda_{\cA,\cC}(T_i) \right\}; D_i=1, \lambda_{\cA,\cC}(T_i) \leq \log \delta  \big] \\
            & \leq \delta \, \Pro_\cC(D_i=1) \leq \delta \, \alpha= \alpha^{\epsilon^2},
        \end{split}
    \end{equation*}
    since $\{D_i=1\}$ is a type-I error under $\Pro_\cC$.
    By the definition of $N$ and $\delta$ we have
    \begin{equation*}
        \begin{split}
            II & \leq \Pro_\cA\left( \max_{1\leq n\leq N} \lambda_{\cA,\cC} (n) > \log \delta \right) \\
            & = \Pro_\cA\left( \frac{1}{N} \max_{1\leq n\leq N} \lambda_{\cA,\cC} (n) > I_{\cA,\cC} (1+\epsilon) \right) \Def \omega_{\cA,\cC} (\alpha,\epsilon),
        \end{split}
    \end{equation*}
    which, by Lemma \ref{lemma used for proving the ALB} and \eqref{SLLN, assumption for ALB},
    converges to zero as $\alpha\to 0$. Combining the above we obtain
    \begin{equation}
        \frac{\Exp_\cA[T_i]}{|\log\alpha|/I_{\cA,\cC}} \geq \left( 1-\alpha^{\epsilon^2} - \omega_{\cA,\cC} (\alpha,\epsilon) -\beta \right)  (1-\epsilon).
    \end{equation}
    Taking infimum over $\chi\in\Delta(\alpha,\beta,\Pi)$, letting first  $\alpha,\beta\to 0$, and then  $\epsilon\to 0$ complete the proof. \\
\end{proof}

\begin{proof} [Proof of Theorem \ref{thm: AO}]
We only prove \eqref{ALB, i, general}, as the proof of 
 \eqref{ALB, j, general} is similar. For this, it suffices to show that, for any  $i \in A$,
\begin{align*}
\min \{ I_{\cA,\cC}:  \cC\in \Pi_{l,u}, \,  i\notin \cC \} \leq
 I_i+\cJ_\cA \cdot \bfone\left\{ |\cA|=l \right\}. 
\end{align*}
When $|\cA|=l>0$,  this is the case because  for $\cC=
\cA\backslash\{i\}\cup\{j\}$ with $j\notin\cA$ we have 
$I_{\cA,\cC}= I_i+J_j$.   When  $|\cA|>l\geq 0$, this is the case because   for $C=\cA\backslash\{i\}$ we have 
$I_{\cA,\cC}= I_i$. \\
\end{proof}

\begin{proof} [Proof of Theorem \ref{thm: AA for decentralized procedures}]
Based on Proposition \ref{prop: AUB}, it suffices to establish the asymptotic lower bounds. For this, we note that for any $\Pi\subseteq 2^{[K]}$ that satisfies \eqref{no trivial} and any decentralized testing procedure in  $\Delta'\cap\Delta(\alpha,\beta,\Pi)$, the corresponding  local test  in stream $k\in[K]$  solves the local testing problem,  $H_0^k$ versus $H_1^k$, controlling the local type-I and type-II error probabilities below $\alpha$ and $\beta$, respectively. From the  asymptotic  optimality theory for sequential binary testing (see, e.g., \cite[Lemma 3.4.1]{Tartakovsky_Book}) we know that, when  \eqref{SLLN, assumption for ALB} holds, the minimum expected sample size of such a test is, asymptotically as $\alpha,\beta\to 0$, equal to  $|\log\alpha|/I_k$ under $H_k^1$ and $|\log\beta|/J_k$ under $H_k^0$. \\
\end{proof}

\begin{proof} [Proof of Theorem \ref{thm: AA for synchronous stopping}]
Based on Proposition \ref{prop: AUB}, 
it suffices to establish the asymptotic lower bound. Indeed, since $\Delta'' \subseteq \Delta$, for  any  $\Pi\subseteq 2^{[K]}$, $A \in \cP$,  $\alpha,\beta\in(0,1)$, we have 
\begin{equation*}
    \begin{aligned}
        \cL''_\cA(\alpha,\beta,\Pi) & \geq   \cL_{k,\cA}(\alpha,\beta,\Pi), \quad \forall\; k \in [K],
            \end{aligned}
\end{equation*}
    and from \eqref{AO, i}-\eqref{AO, j} it  follows that,  as $\alpha,\beta\to 0$,
    \begin{equation*}
    \begin{aligned}        
       \max_{k\in[K]} \cL_{k,\cA}(\alpha,\beta,\Pi)
        & \sim \max\left\{ \frac{|\log\alpha|}{\cI_\cA+\cJ_\cA \cdot \bfone\left\{ |\cA|=l \right\}},\;  \frac{|\log\beta|}{\cJ_\cA+\cI_\cA\cdot \bfone\left\{ |\cA|=u \right\}}\right\}.
    \end{aligned}
    \end{equation*}
\end{proof}

\section{} \label{Proof related to composite hypotheses}
In this Appendix, we  establish the result in Section \ref{section: generalization to composite hypotheses}. We only prove Theorem \ref{thm: composite hypotheses}, as the proofs of Theorems
 \ref{thm: composite hypotheses, decentralized procedure} and \ref{thm: composite hypotheses, synchronous decision making} are similar and easier. 
 
\begin{lemma} \label{martingale}
Fix $\cA\subseteq [K]$ and $\bftheta\in\bfTheta_\cA$.
    For any $k\in [K]$, 
    $$ \exp\left\{ \ella_k(n)-\ell_k(n,\theta_k) \right\}, \quad n\in\bN $$
    is an $\{\cF(n),\,n\in\bN\}$-martingale with expectation 1 under $\Pro_{\cA,\bftheta}$.
\end{lemma}
\begin{proof}
    \cite[Lemma D.2]{Song_AoS}. 
\end{proof}

\begin{remark} \label{remark:new measure}
This lemma and the assumption of  independence across streams imply that, for any $\cA\subseteq[K]$, $\bftheta\in\bfTheta_\cA$ and $\cC\subseteq [K]$, we can define a probability measure $\Pros_{\cA,\bftheta,\cC}$ so that
$$ \frac{d\Pros_{\cA,\bftheta,\cC}}{d\Pro_{\cA,\bftheta}}(\cF(n))=\prod_{k\in\cC} \exp\left\{ \ella_k(n)-\ell_k(n,\theta_k) \right\}, \quad n\in\bN. $$ We denote by $\Exps_{\cA,\bftheta,\cC}$ its corresponding expectation. \\
\end{remark}

\begin{proof} [Proof of Theorem \ref{thm: composite hypotheses}]
First of all, we note that by the definition in \eqref{def: lambdas}, for any $k\in[K]$ and $n\in\bN$,
\begin{equation} \label{obs}
\begin{aligned}
    \left\{ \lambdas_k(n)>0 \right\} & \subseteq \left\{ \lambdas_k(n)=\ella_k(n)-\ell_k^0(n) \right\} \\
    \left\{ \lambdas_k(n)<0 \right\} & \subseteq  \left\{ \lambdas_k(n)=\ell_k^1(n)-\ella_k(n) \right\}.
\end{aligned}
\end{equation}

    (i) For the error control, we only consider the case when $l<u$,
    as the case when $l=u$ can be proved similarly.       Fix  $\cA\in\Pi_{l,u}$ and $\bftheta\in\bfTheta_\cA$.
As in the proof of Proposition \ref{thm: a.s. finite and error control} in Appendix \ref{Proof related to AUB}, we have 
    $$ \Pro_{\cA,\bftheta}(\bfDss\backslash\cA\neq\emptyset)\leq \sum_{j\notin\cA}\left( \Pro_{\cA,\bftheta}(\Gammas_j)+\sum_{k\in\cA} \Pro_{\cA,\bftheta}(\Gammas_{j,k}) \right), $$
    where 
\begin{align*}
 \Gammas_j &\equiv \left\{ \lambdas_j(\Tss_j)> a \right\}, \quad  \\
 \Gammas_{j,k} &\equiv \left\{ \lambdas_j(\Ts_j)> \max\big\{ \lambdas_k(\Ts_j)+c,\,0 \big\},\; \lambdas_k(\Ts_j)< 0 \right\}.
 \end{align*}
Thus, for the proof of uniform control of type-I familywise error rate it suffices to  show that 
         \begin{align} \label{ineq}
         \begin{split}
\Pro_{\cA,\bftheta}(\Gammas_j) &\leq e^{-a} \quad
 \forall \;  j\notin\cA \\
 \Pro_{\cA,\bftheta}(\Gammas_{j,k}) &\leq e^{-c} \quad 
\forall \; j\notin\cA, \; k\in\cA.
\end{split}
\end{align}
Fix  $j\notin\cA$ and $k\in\cA$. By  \eqref{def: lambdas},
 \eqref{generalized ll statistic}, Lemma \ref{martingale} and \eqref{obs}, we have
         \begin{align*}
     \Gammas_j  &= \left\{ \ella_j(\Tss_j)-\ell_j^0(\Tss_j)> a \right\}  \\
     &\subseteq \left\{ \ella_j(\Tss_j)-\ell_j(\Tss_j,\theta_j) =  \frac{d\Pros_{\cA,\bftheta,\{j\}}}{d\Pro_{\cA,\bftheta}}\big(\cF(\Tss_j)\big)> a \right\}\\
             \Gammas_{j,k} & \subseteq \left\{ \left( \ella_j(\Tss_j)-\ell_j^0(\Tss_j) \right) + \left( \ella_k(\Tss_j)-\ell_k^1(\Tss_j) \right) > c \right\}, \\
             & \subseteq \left\{ \left( \ella_j(\Tss_j)-\ell_j(\Tss_j,\theta_j) \right) + \left( \ella_k(\Tss_j)-\ell_k(\Tss_j,\theta_k) \right) =  \frac{d\Pros_{\cA,\bftheta,\{j,k\}}}{d\Pro_{\cA,\bftheta}}\big(\cF(\Tss_j)\big) > c \right\},
     \end{align*}
   where the probability measures $\Pro_{A,\bftheta,\{j\}}$ and  $\Pro_{A,\bftheta,\{j,k\}}$ are  defined in Remark \ref{remark:new measure}.   Thus,  the first (resp. second) inequality in \eqref{ineq} follows by a  change of measure  from $\Pro_{A,\bftheta}$ to $\Pro_{A,\bftheta,\{j\}}$ (resp.  $\Pro_{A,\bftheta,\{j,k\}}$). \\

(ii) We only prove \eqref{AO, i, composite} as \eqref{AO, j, composite} can be proved similarly.
Under assumption \eqref{assump2}, for any $\epsilon>0$, there exists $\bftheta^\epsilon\equiv \left( \theta^\epsilon_1,\ldots,\theta^\epsilon_K \right) \in\bfTheta_{\cA^c}$, such that 
\begin{equation*}
    \begin{aligned}
        I_i(\theta_i,\theta_i^\epsilon) & \leq (1+\epsilon) \, I_i(\theta_i), \quad \forall\,i\in\cA \\
        I_j(\theta_j,\theta_j^\epsilon) & \leq (1+\epsilon) \, J_j(\theta_j), \quad \forall\,j\notin\cA.
    \end{aligned}
\end{equation*}
Consider the following problem of testing multiple pairs of simple hypotheses:
\begin{equation*}
    \begin{aligned}
        H_i^0: \gamma_i=\theta_i^\epsilon
        \quad & \text{ versus } \quad H_i^1: \gamma_i=\theta_i, \quad \forall\; i\in\cA, \\
        H_j^0: \gamma_j=\theta_j \quad & \text{ versus } \quad H_j^1: \gamma_j=\theta_j^\epsilon, \quad \forall\; j\notin\cA,
    \end{aligned}
\end{equation*}
where for each $k\in[K]$ we write $\gamma_k$ for the generic local parameter in stream $k$ to distinguish it from the $k^{th}$ component of $\bftheta$.
Since any test in $\Delta^{*}(\alpha,\beta,\Pi_{l,u})$ also solves this problem and controls the two types of familywise error rates below $\alpha,\beta$ respectively, by Theorem \ref{thm: AO}, under assumption \eqref{assump1.1}, we have that for every $i\in\cA$, as $\alpha,\beta\to 0$,
\begin{equation*}
\begin{aligned}
    & \;\inf\left\{ \Exp_{\cA,\bftheta}[T_i]: \chi\in\Delta^{*}(\alpha,\beta,\Pi_{l,u})  \right\} \\
    \gtrsim & \; \frac{|\log\alpha|}{I_i(\theta_i,\theta^\epsilon_i)+ \min_{j\notin\cA} I_j(\theta_j,\theta_j^\epsilon)\cdot \bfone\left\{ |\cA|=\ell \right\} } \\
    \geq & \; \left.\left(\frac{|\log\alpha|}{I_i(\theta_i)+\cJ_\cA(\bftheta) \cdot \bfone\left\{ |\cA|=\ell \right\}}  \right)\right/(1+\epsilon).
\end{aligned}
\end{equation*}
Since $\epsilon>0$ is arbitrary, the asymptotic lower bound is proved.\\

We next  show that  the proposed test attains this lower bound. We show this only  when $l<u$, as the proof when  $l=u$  is similar (and easier). Fix $\cA\in\Pi_{l,u}$ and   $i\in\cA$. Then, by \eqref{def: lambdas} and \eqref{composite testing, gap-inter} we have 
\begin{equation*}
    \Tss_i\leq\inf\left\{ n\in\bN: \ella_i(n)-\ell_i^0(n)>a \right\}. 
\end{equation*}
In view of assumption \eqref{assump3},  we  apply Lemma \ref{basic lemma for AUB} 
and obtain, as $a \to \infty$,
$$ \Exp_{\cA,\bftheta}[\Tss_i]\lesssim \frac{a}{I_i(\theta_i)}. $$ 
When, additionally, $|\cA|=\ell$, we have 
\begin{equation*}
    \begin{aligned}
        \Tss_i & \leq 
        \inf\left\{ n\in\bN: \lambdas_i(n)>\max\big\{ \lambdas_j(n)+c,\;0 \big\}, \; \lambdas_j(n)< 0, \;\forall\,j\notin\cA \right\} \\
        & \leq \inf\left\{ n\in\bN: \left(\ella_i(n)-\ell_i^0(n)\right)+\left(\ella_j(n)-\ell_j^1(n)\right)> c,\; \forall\,j\notin\cA \right\}.
    \end{aligned}
\end{equation*}
In view of assumption \eqref{assump3},  we  apply Lemma \ref{basic lemma for AUB} and obtain, as $c \to \infty$,
$$ \Exp_{\cA,\bftheta}[\Tss_i]
\lesssim \frac{c}{I_i(\theta_i)+\min_{j\notin\cA}J_j(\theta_j)}= \frac{c}{I_i(\theta_i)+\cJ_\cA(\bftheta)}. $$
The proof is complete, since the selected $a,c$ satisfy $a\sim c\sim |\log\alpha|$ as $\alpha,\beta\to 0$. \\
\end{proof}

\section{} \label{Proofs related to other global error metrics}
 In this Appendix we prove Proposition \ref{prop: GEM}.

\begin{proof} [Proof of Proposition \ref{prop: GEM}]
We only check the part related to type-I errors, as the part related to type-II errors is similar. To lighten the notation, we suppress the dependence on $\chi=(\bfT,\bfD)$, and we let $R \equiv |\bfD|$ and  $V\equiv |\bfD\backslash\cA|$.
Noticing that
\begin{equation*}
    \begin{split}
        \text{pFDR}_\cA^1 & \equiv \Exp_\cA\left[ 
        \left.\frac{V}{R}\right| R\geq 1 \right] = \frac{\Exp_\cA[V/R\cdot\bfone\{R\geq 1\}]}{\Pro_\cA(R\geq 1)} \\
        & =\frac{\text{FDR}_\cA^1}{\Pro_\cA(R\geq 1)}\geq \text{FDR}_\cA^1 \geq \frac{1}{K} \text{FWE}_\cA^1,
    \end{split}
    \end{equation*}
    we have \eqref{C2} holds with $C_2=1/K$.\\

When  $0<l\leq u<K$, then on the event $\{R=0\}$ at least one type-II error is made and, thus,
    $$ \Pro_\cA(R=0)\leq \text{FWE}_\cA^2, $$
    \begin{equation*}
        \text{pFDR}_\cA^1 = \frac{\text{FDR}_\cA^1}{1-\Pro_\cA(R=0)} \leq \frac{\text{FWE}_\cA^1}{1-\text{FWE}_\cA^2}\leq 2\cdot\text{FWE}_\cA^1,
    \end{equation*}
where the last inequality holds  when  $\text{FWE}_\cA^2\leq 1/2$. \\ 
\end{proof}

\bibliographystyle{chicago}
\bibliography{main}

\end{document}